\documentclass[aps,twocolumn,prx,floatfix,preprintnumbers, nofootinbib,10pt]{revtex4-2}
\usepackage{xpatch}
\makeatletter
\patchcmd{\@ssect@ltx}
    {\addcontentsline{toc}{#1}{\protect\numberline{}#8}}
    {}
    {}
    {}
\makeatother
\pdfoutput=1
\usepackage{color}
\usepackage{enumitem}

\usepackage{mathtools}
\usepackage[dvipsnames]{xcolor}
\usepackage{float}

\usepackage{hyperref}
\usepackage{cleveref}
\definecolor{amaranth}{rgb}{0.9, 0.17, 0.31}
\definecolor{forestForestGreen(web)}{rgb}{0.13, 0.55, 0.13}
\definecolor{SymTFTColores}{rgb}{0, 0.33, 70.1}
\definecolor{blue(munsell)}{HTML}{005567}
\definecolor{oxfordblue}{rgb}{0.0, 0.2, 0.4}
\definecolor{bblue}{rgb}{0.0, 0.58, 0.71}
\hypersetup{
pdfstartview={FitH}, 
pdftitle={}, 
colorlinks=true, 
linkcolor=oxfordblue, 
citecolor=oxfordblue, 
filecolor=oxfordblue, 
urlcolor=oxfordblue
}
\usepackage{pgfplots}
\pgfplotsset{compat=1.18}
\usepackage{amssymb}
\usepackage{amsfonts}
\usepackage{graphicx}
\usepackage{epstopdf}
\usepackage{dcolumn}
\usepackage{amsmath}
\usepackage{latexsym,bm}
\usepackage{amsthm}
\usepackage{slashed}
\usepackage{color}
\usepackage{url}
\usepackage{rotating}
\usepackage[normalem]{ulem}
\usepackage{faktor}
\usepackage{tikz-cd}
\usepackage{tensor}
\usepackage{array}
\usepackage{multirow}
\usepackage{comment}

\usepackage{tabularx}
\usepackage{makecell}
\usepackage[normalem]{ulem}
\usepackage{changepage}
\numberwithin{equation}{section}

\usepackage{makecell}
\newtheorem{theorem}{Theorem}[section]
\newtheorem{proposition}[theorem]{Proposition}
\newtheorem{lemma}[theorem]{Lemma}
\newtheorem{corollary}[theorem]{Corollary}
\theoremstyle{definition}
\newtheorem{definition}[theorem]{Definition}
\newtheorem{remark}[theorem]{\bf Remark}
\newtheorem{example}[theorem]{\bf Example}

\usepackage{tikz}
\usepackage{diagbox}
\usetikzlibrary{positioning}
\usetikzlibrary{calc}
\usetikzlibrary{decorations.pathreplacing,calligraphy}

\usepackage{xstring}
\usetikzlibrary{decorations.pathmorphing} 
\usetikzlibrary{decorations.markings} 
\usetikzlibrary{arrows} 
\usetikzlibrary{shapes} 
\usetikzlibrary{matrix} 
\usetikzlibrary{positioning} 
\usepackage[english]{babel} 
\usepackage[autostyle]{csquotes}
\usepackage{pifont}
\usetikzlibrary{shapes.multipart}

\tikzset{->-/.style={decoration={
  markings,
  mark=at position .6 with {\arrow{Latex[length=1.5mm,width=1.5mm]}}},postaction={decorate}}}

\usepackage{quantikz}

\newcommand{\bea}{\begin{eqnarray}}
\newcommand{\eea}{\end{eqnarray}}
\newcommand{\be}{\begin{equation}}
\newcommand{\ee}{\end{equation}}
\newcommand{\ba}{\begin{aligned}}
\newcommand{\ea}{\end{aligned}}
\newcommand{\bit}{\begin{itemize}}
\newcommand{\eit}{\end{itemize}}
\newcommand{\ben}{\begin{enumerate}}
\newcommand{\een}{\end{enumerate}}
\newcommand{\nn}{\nonumber}
\newcommand{\id}{\text{id}}

\newcommand{\Dir}{\text{Dir}}

\newcommand{\Bsym}{\mathfrak{B}^{\text{sym}}}
\newcommand{\Bphys}{\mathfrak{B}^{\text{phys}}}

\newcommand{\lb}{\left(}
\newcommand{\rb}{\right)}

\newcommand{\wt}{\widetilde}
\newcommand{\wh}{\widehat}
\newcommand{\ot}{\otimes}

\newcommand{\Z}{{\mathbb Z}}

\newcommand{\bC}{{\mathbb C}}

\newcommand{\Ind}{\text{Ind}}

\newcommand{\Spin}{\text{Spin}}
\newcommand{\cA}{\mathcal{A}}
\newcommand{\cB}{\mathcal{B}}

\newcommand{\cC}{\mathcal{C}}

\newcommand{\cI}{\mathcal{I}}

\newcommand{\cL}{\mathcal{L}}
\newcommand{\cM}{\mathcal{M}}

\newcommand{\cO}{\mathcal{O}}
\newcommand{\cP}{\mathcal{P}}

\newcommand{\cS}{\mathcal{S}}

\newcommand{\cZ}{\mathcal{Z}}

\newcommand{\A}{\mathsf{A}}

\newcommand{\Tr}{\text{Tr}}

\newcommand{\Hom}{\text{Hom}}
\newcommand{\End}{\text{End}}
\renewcommand{\Vec}{\mathsf{Vec}}
\newcommand{\Rep}{\mathsf{Rep}}
\newcommand{\Mod}{\mathsf{Mod}}
\newcommand{\Bimod}{\mathsf{Bimod}}
\renewcommand{\dim}{\text{dim}}

\newcommand{\Vect}{\mathsf{Vec}}

\newcommand{\Ising}{\mathsf{Ising}}

\newcommand{\sym}{\text{sym}}
\newcommand{\phys}{\text{phys}}

\makeatother

\newcommand{\eps}{\epsilon}
\newcommand{\vp}{\varphi}
\renewcommand{\ol}{\overline}
\definecolor{mygreen}{RGB}{28,172,0}
\definecolor{mybluegreen}{HTML}{10414a}
\definecolor{mycyan}{HTML}{00a5fa}

\makeatletter
\renewcommand\subsubsection{\@startsection{subsubsection}{3}{\z@}%
                                     {-3.25ex\@plus -1ex \@minus -.2ex}%
                                     {2.5ex \@plus .2ex}%
                                     {\centering\itshape\bfseries\small}}
\makeatother

\usepackage{physics}
\renewcommand{\ol}{\overline}

\newcommand{\diag}{\text{diag}}
\newcommand{\folded}{\text{folded}}

\newcommand{\III}{\text{III}}
\newcommand{\Irr}{\text{Irr}}

\def\ad{\mathop{\mathrm{ad}}\nolimits}

\def\unit{{1\kern-.65ex {\rm l}}}
\def\1{{1\kern-.65ex {\rm l}}}

\usepackage{makecell}

\makeatletter
\newcommand\xlabel[2][]{\phantomsection\def\@currentlabelname{#1}\label{#2}}
\makeatother

\makeatletter
\def\l@subsubsection#1#2{}
\makeatother

\begin{document}

\title{Twin Algebras: \\
Condensable Algebras beyond Anyons}

\author{Yuhan Gai}
\author{Sakura Sch\"afer-Nameki}
\author{Alison Warman}

\affiliation{Mathematical Institute, University
of Oxford, Woodstock Road, Oxford, OX2 6GG, United Kingdom}

\begin{abstract}
\noindent 
Condensable algebras in 2+1d non-chiral topological orders characterize gapped boundary conditions and interfaces. 
Applied to the Symmetry Topological Field Theory, they allow classification of symmetric gapped phases and impose sharp constraints on possible phase transitions. 
A condensable algebra is specified not only by its underlying set of anyons, which end on the boundary or interface, but also by its algebra structure. 
We introduce the concept of twin condensable algebras, which have the same anyon decomposition, but inequivalent algebra structure. 
We revisit the classification of condensable algebras in $\mathcal{Z}(\text{Vec}_G^\omega)$, i.e. in group-theoretical topological orders for finite groups $G$ with anomaly $\omega$. 
In this context we are able to identify twin algebras that arise from different mechanisms, such as subgroup data, SPT cocycles, and symmetry actions. 
In particular, we construct infinite families of examples of twins from so-called Gassmann triples, and exhibit cases in which the reduced topological orders are inequivalent despite having identical anyon content. 
Physically, twin algebras describe distinct symmetric phases that
have isomorphic spaces of ground states, but inequivalent order parameters. Such twin phases never exhibit relative spontaneous symmetry breaking, and can be used to construct phase transitions without hidden symmetry breaking, which are intrinsically beyond Landau transitions.
\end{abstract}


\maketitle
\tableofcontents

\section{Introduction}
\label{sec:Intro}



The study of symmetries is vital in the characterization of phases of matter. 
Much of the recent excitement has arisen from considering generalized symmetries \cite{Gaiotto:2014kfa} and in particular non-invertible symmetries (for reviews see \cite{Schafer-Nameki:2023jdn, Shao:2023gho, Bhardwaj:2023kri}). Surprisingly, however, even for finite group symmetries, we can still learn new things. This is the main point of the present paper. 

The insights are however a result of developments of  systematic tools that honed our ability to study symmetries, gauging, gapped and gapless phases, in the realm of categorical symmetries. In particular in 1+1d, the tools for the study of categorical symmetries have led to an improved understanding and approach to so-called fusion category symmetries, which are the natural generalization of finite group symmetries, and include non-invertible symmetries. 

A particularly useful framework is that of the Symmetry topological field theory (SymTFT) \cite{Ji:2019jhk,  Gaiotto:2020iye, Apruzzi:2021nmk, Freed:2022qnc}, which allows a systematic exploration of symmetric gapped phases for fusion categories \cite{Bhardwaj:2024qrf,  Bhardwaj:2024wlr, Bhardwaj:2024kvy, Chatterjee:2024ych, Warman:2024lir, Bottini:2025hri}. Here, we will focus mostly on the case of group-theoretical symmetries, i.e. those obtained by gauging of finite group symmetries, however the framework is applicable to any fusion category. 

\begin{figure}
$$
\begin{tikzpicture}
\begin{scope}[shift={(0,0)}]
\draw [SymTFTColores, fill= SymTFTColores, opacity = 0.2] 
(0,0) -- (0,2) -- (2,2) -- (2,0) -- (0,0) ; 
\draw [white] (0,0) -- (0,2) -- (2,2) -- (2,0) -- (0,0)  ; 
\draw [very thick] (0,0) -- (0,2) ;
\draw [very thick] (2,0) -- (2,2) ;
\draw [ thick] (0,0.8) -- (2,0.8) ;
\fill[] (0,0.8) circle (0.05cm);
\fill[] (2,0.8) circle (0.05cm);
\node[above] at (1,0.8) {$a$};
\node[above] at  (1,1.3)  {$\cZ(\cS)$} ;
\node[above] at (0,2) {$\cL^\sym_{\cS}$}; 
\node[above] at (2.2,2) {$\cL^\phys$}; 
\end{scope}
\begin{scope}[shift={(4,0)}]
\node at (-1,1) {$=$} ;
\draw [very thick] (0,0) -- (0,2) ;
\fill[] (0,0.8) circle (0.05cm);
\node[right]at  (0,0.8) {$\cO_{a}$};
\node[above] at (0.1,2) {$\cP_{\cL}^{\cS}$}; 
\end{scope}
\end{tikzpicture}
$$
\caption{The SymTFT for $\cS$-symmetric phases: $\cL_\cS^\sym$ is the symmetry boundary (Lagrangian algebra), and $\cL^\phys$ the physical, and for gapped phases is chosen to be a Lagrangian algebra. The interval compactification gives rise to the phase $\cP$. Local order parameters are $\cO_a$, that arise from anyons in the SymTFT (generalized charges) that end on  both boundaries. \label{fig:SymTFTAGAIN}}
\end{figure}
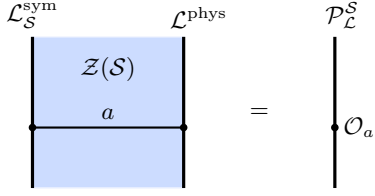

\vspace{2mm}
\noindent{\bf Gapped boundary conditions.}
The SymTFT takes in a symmetry $\cS$ and realizes it on the gapped boundary of a 2+1d TQFT. Putting this theory on a finite interval, with a gapped second boundary, the physical boundary, gives rise to a gapped phase, see  \Cref{fig:SymTFTAGAIN}. Thus, clearly, studying gapped boundary conditions of the SymTFT is paramount to understanding gapped phases in 1+1d. Examples we will focus on are $\cS$ that are obtained from finite group symmetries $G$, possibly after gauging. The SymTFT is the $G$-gauge theory or Dijkgraaf-Witten theory (with a twist, if $G$ has an 't Hooft anomaly $\omega$). The topological line defects of the SymTFT form a braided fusion category, the Drinfeld center $\cZ (\cS)$, and will be referred to as anyons. 
The SymTFT paradigm now states, that gapped boundary conditions are in 1-1 correspondence with so-called Lagrangian algebras in the Drinfeld center $\cZ(\cS)$ of $\cS$. 
Lagrangian algebras are  decomposable as non-negative integer combination of anyons $a_i$ 
\be
\cL = \bigoplus_i n_i a_i \,.
\ee
However crucially, there is an algebra structure, i.e multiplication 
\be
m: \quad \cL \otimes \cL \to \cL \,.
\ee
{The multiplication $m$ is an element in the vector space $\Hom_{\cZ(\cS)}(\cL\otimes \cL, \cL)$.}
It is this algebra structure that we will be most interested in in this paper. 

Similarly condensable algebras more generally are \'etale algebras in $\cZ (\cS)$, which need not be of maximal dimension. See \Cref{app:defn-cond-alg} for definition of condensable and Lagrangian algebras. Lagrangian algebras are condensable algebras of maximal dimension. Non-Lagrangian condensable algebras organize gapless phases \cite{Chatterjee:2022tyg,Bhardwaj:2023bbf, Wen:2023otf, Bhardwaj:2024qrf}. 

\vspace{2mm}
\noindent{\bf Twin Algebras.}  
We define twin algebras to be condensable algebras, which have the same anyon decomposition $\oplus_i n_i a_i$ but a different multiplication $m$. 
Lagrangian algebras with this property are referred to as twin Lagrangian algebras.  
Twin algebras give rise to extremely interesting phases and phase transitions, which will be explored, both here and in the companion paper \cite{Warman:2026gfz}. 

However, 
in general it is hard to find or even classify condensable algebras in modular tensor categories. Even for group-theoretical fusion categories, where the SymTFT has defects in $\cZ (\Vec_G^\omega)$, we found that there are some loopholes in the literature which we will close before embarking on a general study of twin algebras.

\vspace{2mm}
\noindent{\bf Group-theoretical Symmetries.} 
We focus on group-theoretical symmetries 
which can be obtained by gauging (Morita dual\footnote{See \Cref{append:alg_Morita_equiv} for the definition of Morita dual and the notation $\cC(G, \omega, K, \beta):=(\Vec_G^\omega)_{\cM(K, \beta)}^*$. }) from group symmetries $G$ with anomaly $\omega \in H^3 (G, U(1))$, i.e. $\Vec_G^\omega$. Physically, these are symmetries obtained from $\Vec_G^\omega$ by stacking a $K$-SPT $\beta$ and gauging a subgroup $K$ of $G$. 

As the Drinfeld center is invariant under gauging the symmetry, luckily all these categories have the same SymTFT: 
\be 
\cZ(\cC(G, \omega, K, \beta))\cong \cZ(\Vec_G^\omega) \,.
\ee
To study symmetries and phases, we will be interested in the condensable algebras in $\cZ(\Vec_G^\omega)$. These are classified and labeled by 
\be \label{AHNGE}
A(H, N, \gamma, \eps) \,,
\ee 
where $N\triangleleft H \subset G$ are subgroups, $\gamma$ trivializes $\omega$ on $N$, and together with $\eps$, they define $\alpha\in H^3(H/N, U(1))$ that pulls back to $\omega$ on $H$ \cite{Davydov2009ModularIF, davydov2017lagrangian, Hannah:2023tae}.
These give rise to interfaces from $\cZ (\Vec_G^\omega)$ to the reduced TO
\be
\label{ReducedTO}
\cZ(\Vec_{H/N}^\alpha)
\ee
where $\alpha$ depends on $\omega$ and the SPT data of the algebra and is defined in \eqref{eqn:twist_reduced_TO}.
Very broadly speaking, in terms of  $\Vec_G^\omega$-symmetric phases, the ones corresponding to $A(H, N, \gamma, \eps)$ have preserved symmetry $H$, $N$ labels the support of the twisted sectors, $\gamma$ modifies the multiplication of (twisted) operators and $\eps$ modifies the symmetry action on twisted operators. The Lagrangian algebras are singled out by $H=N$ and $\epsilon$ is fixed in terms of $\gamma$. We will denote Lagrangian algebras by 
\be
\cL (H, \gamma) := A (H, H, \gamma, \epsilon(\gamma)) \,.
\ee
These are the gapped boundary conditions, obtained by stacking $\gamma$ and gauging $H$ on the Dirichlet boundary that realizes $\Vec_G^\omega$, i.e. $\cL(1,1)$.

However, this classification has redundancies. We identify the isomorphic one in \Cref{sec:IsomAlg}. This is important for physical applications because, we are interested in distinct, inequivalent symmetric gapped phases.

\vspace{2mm}
\noindent{\bf Twins for Group-theoretical Symmetries.}
With this precise (and free of redundancies) classification at hand, we can then explore the landscape of algebras systematically. In particular we are able to identify criteria for when algebras are twin in terms of the classification data (\ref{AHNGE}).

We find twin algebras in $\cZ(\Vec_G^\omega)$ of different types, which can be distinguished by the data in (\ref{AHNGE}) and will be referred to as $H$-, $N$-, $\gamma$- and $\eps$-type twin algebras, which differ in these entries of the algebra classification datum. Their anyon decomposition is the same, but the algebra structure differs. 

Twin algebras in $\cZ(\Vec_G^\omega)$ are closely related to  {\bf Gassmann triples} \cite{Sunada1985RiemannianCA, Gassmann:1926}, which we  review  in \Cref{sec:Gassman}: let $H_i$ be subgroups of $G$. 
$(G, H_1, H_2)$ form a Gassmann triple if 
\be 
|H_1 \cap [g]|=|H_2\cap [g]|
\ee
for any conjugacy class $[g]$ of $G$. Note that these can also be trivial if $H_1=gH_2g^{-1}$ for some $g\in G$.
For twins in $\cZ(\Vec_G^\omega)$, we show that their $H$ and $N$ subgroups have to form Gassmann triples of $G$.

Twin algebras can be distinguished by the different algebra structure on the same set of anyons, however, we have found other interesting imprints of them. E.g.  we find that non-maximal (i.e. non-Lagrangian) twin algebras of $H$-type can lead to inequivalent reduced TOs (\ref{ReducedTO}) and give examples of these in \Cref{sec:example_Gassmann_Triple_different_reduced_TOs}.

The $\gamma$-type twin algebras include the ones obtained from Bogomolov multiplier \cite{Pollmann:2012, Davydov:2013xov, Kobayashi:2025ykb, Kobayashi:2025pxs}. We also found $\gamma$-type twins beyond Bogomolov multipliers. 
The $\eps$-type twin algebras define twin phases that can be distinguished using generalized string order parameters.

\vspace{2mm}
\noindent{\bf Physical Implications.}
The physical content of twin algebras is developed in \Cref{sec:TwinPhys} and the companion paper \cite{Warman:2026gfz}. Twin Lagrangian algebras $\cL_1$ and $\cL_2$ in the
SymTFT $\cZ(\mathcal{S})$ define \textbf{twin gapped phases}
$\mathcal{P}^{\mathcal{S}}_{\cL_1},\mathcal{P}^{\mathcal{S}}_{\cL_2}$ via
interval compactification with a fixed symmetry boundary. Although these
phases share the same generalized-charges, they are physically
inequivalent: their algebras of order parameters differ. We make this
precise by recovering the (twisted) local-operator algebra of a phase
$\mathcal{P}^{\mathcal{S}}_{\cL^\phys}$ directly from the
condensable algebra $\cL^{\phys}$ via the forgetful functor
\be
F_{{\mathcal{S}}}:\quad \cZ(\mathcal{S})\to\mathcal{S} \,,
\ee
with multiplication
inherited from the comultiplication on $\cL^{\phys}$. Twin algebras,
by definition, agree as objects and disagree precisely on this
multiplication. The main consequence is a sharp dichotomy between twin and non-twin Lagrangians: 
two phases exhibit \textbf{no relative
spontaneous symmetry breaking (SSB)}:
for every fixed choice of symmetry boundary (i.e. for all symmetries related by gauging, or Morita equivalent symmetries), the  twin phases do not have relative SSBing, i.e. they have the same number of vacua (of course the number of vacua depends on the choice of symmetry). Twin gapped phases are therefore the natural candidates for direct transitions 
without hidden symmetry breaking, and, equipped with an anomaly, for intrinsically beyond Landau transitions. 
concrete realizations are the
subject of \Cref{sec:G32-43_examples} and the companion paper \cite{Warman:2026gfz}. These provide new examples of DQCPs \cite{Senthil:2023vqd} in 1+1d, which remain DQCPs for any Morita equivalent symmetry.

\section{Condensable Algebras in $\cZ(\Vec_G^\omega)$}
\label{sec:algebras}

Let $G$ be any finite group and $\omega\in H^3(G, U(1))$.
A condensable algebra in the Drinfeld center $\cZ(\Vec_G^\omega)$ is classified in \cite{davydov2017lagrangian, Hannah:2023tae}, which, however, contains some redundancies, which we will identify in \Cref{sec:IsomAlg}.    
Physically, such equivalence classes characterize all possible inequivalent interfaces from $\cZ(\Vec_G^\omega)$ topological order (TO) to a reduced TO. 
For physical applications, where the topological defects in $\cZ (\Vec_G^\omega)$ play the role of generalized charges, it is useful to have a decomposition of the algebras in terms of anyons. 
In \Cref{sec:AnyonDecomp} we provide a derivation of the anyon decomposition for any condensable algebra in $\cZ (\Vec_G^\omega)$. However, inequivalent algebras can have the same anyon decomposition, meaning they are isomorphic as objects. Such pairs of algebras will be called twin algebras, and are defined in \Cref{sec:TwinAlgebras}. Finally we discuss the partial order (Hasse diagram) on condensable algebras in \Cref{sec:Hasse} and the interfaces and reduced TOs in \Cref{sec:Fold}.

\subsection{Classification of Condensable Algebras} \label{sec:CondAlg_class}

The condensable algebras and their reduced TOs (whose definitions are recalled in \Cref{append:math-details})  
in $\cZ(\Vec_G^\omega)$ were previously discussed in \cite{Davydov2009ModularIF, davydov2017lagrangian, Hannah:2023tae}. 
Here, we review this classification in \Cref{sec:review-algebras}. We then determine the  isomorphism classes of condensable algebras in Lemma~\ref{lem:Morita_equiv} (with proof provided in Appendix~\ref{append:alg_Morita_equiv}).

\subsubsection{Condensable Algebras}
\label{sec:review-algebras}

The following description of condensable algebras in the TO $\cZ(\Vec_G^\omega)$ and their corresponding reduced TO were proven in \cite{davydov2017lagrangian}:

\begin{theorem}[\cite{davydov2017lagrangian}, Thm 3.15]
\label{thm:DS-cond-alg}
    Let $G$ be a finite group and $\omega\in H^3(G, U(1))$.
    A condensable algebra in $\cZ(\Vec_G^\omega)$ is determined by:
    \begin{itemize}
    	\item $H\subseteq G$ is a subgroup;
    	\item $N \triangleleft H$ is a normal subgroup;
    	\item $\gamma: N\times N \rightarrow U(1)$ such that $d\gamma = \omega|_{N\times N \times N}$;
    	\item $\epsilon: H\times N \rightarrow U(1)$ satisfying compatibility conditions (see \Cref{append:math-details}),
    \end{itemize}
    denote the corresponding algebra by $A(H, N, \gamma, \epsilon)$.
\end{theorem}
We refer to $\gamma$ as the \emph{SPT} and $\eps$ the \emph{$H$-action phase}. These algebras give rise to interfaces between $\cZ (\Vec_G^\omega)$ and a reduced TO which is determined as follows: 
\begin{theorem}[\cite{davydov2017lagrangian}, Thm 3.16]
\label{thm:DS-reducedTO}
    The reduced TO corresponding to $A(H,N,\gamma,\epsilon)$ in $\cZ(\Vec_{G}^{\omega})$ is
    \be \label{eq:RedTO}
\cZ(\Vec_{G}^{\omega})_{A(H,N,\gamma,\epsilon)}^{\text{loc}}\cong \cZ(\Vec_{H/N}^{\alpha}),
    \ee
    where $\alpha \in H^3(H/N, U(1))$ is determined by $\omega,\gamma,\epsilon$ as
    \begin{align}
            \alpha(x,y,z) = {} & \omega(s(x), s(y), u)\frac{\gamma(\tau_1, \tau_2)}{\gamma(\tau_4, {\tau_3^{s(z)\tau_4}})}\eps(s(z) \tau_4, \tau_3^{s(z)\tau_4}) \nn \\
            & \times \frac{\omega(s(x)s(y), \tau_3, s(z) \tau_4)}{\omega(s(y), s(y)^{-1}s(yz), \tau_2)\omega(s(z), \tau_1, \tau_2)} \nn \\
            & \times \frac{\omega(s(z), \tau_4, \tau_3^{s(z)\tau_4})}{\omega(s(x), s(y), \tau_3)}, \label{eqn:twist_reduced_TO}
    \end{align}
    where $s:H/N \rightarrow H$ is a section, $\tau_1=\tau(y,z)$, $\tau_2=\tau(x,yz)$, $\tau_3=\tau(x,y)$, $\tau_4={\tau(xy,z)}$ and $u=s(y)^{-1} s(x)^{-1} s(xyz)$. {Here $\tau(a,b)=s(b)^{-1}s(a)^{-1}s(ab)\in N$ for $a,b\in H/N$, and $f^g=g^{-1}fg$.}
\end{theorem}
Two comments to be made on formula \eqref{eqn:twist_reduced_TO}:
\begin{itemize}
    \item As noted in \cite{Hannah:2023tae}, $\alpha$ is such that its pullback $p^*\alpha$ is equivalent to $\omega|_H$.
    \item The formula presented in equation~\eqref{eqn:twist_reduced_TO} is not exactly the same as the one in \cite{davydov2017lagrangian}. In Appendix~\ref{append:twist_reduced_TO}, we follow the explicit steps of \cite{davydov2017lagrangian} to derive \eqref{eqn:twist_reduced_TO}. 
\end{itemize}

A special case are Lagrangian algebras, which are condensable algebras of maximal dimension (see \Cref{append:math-details} for definition) \cite{davydov2013witt}, which give rise to gapped boundary conditions: 
\begin{corollary}[Lagrangian algebras]
    $A(H,N, \gamma, \epsilon)$ is a Lagrangian algebra if and only if $H=N$. In this case $\epsilon$ is uniquely determined by $\gamma$ (see App~\ref{append:math-details}), hence denote a Lagrangian algebra by $\cL(H,\gamma)$.
\end{corollary}
An example of Lagrangian algebras are  
$A(1,1)$, which realizes the $\Vec_G$ symmetry, also known as the Dirichlet boundary. 
For trivial $\omega$, there is also the boundary $\cL(G,1)$ obtained by gauging $G$, which defines the associated Neumann $\Rep(G)$ symmetry boundary. 
More generally $\cL (H, \gamma)$  is the gapped boundary condition, obtained from the $\Vec_G$ Dirichlet boundary after stacking with $\gamma$ and gauging $H$.

The  {\bf algebra structure} of $A(H,N,\gamma,\eps)$ in $\cZ(\Vect_G^\omega)$, whose objects can be modeled as $G$-graded vector spaces equipped with $G$-actions, is made explicit in \cite{Davydov2009ModularIF, Hannah:2023tae}. We summarize the algebra structure as follows:
As a vector space over $\bC$, $A(H,N,\gamma,\epsilon)$ in $\cZ(\Vec_G^\omega)$ is spanned by $v_{g,n}$, $g\in G, n\in N$,
subject to the relation
\be \label{eq:v_equiv_rel}
    v_{gh,n} = \epsilon(h,n)\,\omega(g,h|n)\,v_{g,hnh^{-1}}, \quad \forall h\in H,
\ee
where $\omega(g,h|n)\in U(1)$ is the transgression of $\omega$, whose explicit formula can be found in \eqref{eqn:alpha_proj_action} (see \cite{DIJKGRAAF199160} for its properties). 
Fixing a set $R$ of representatives for $G/H$ fixes a basis 
\be
\label{eqn:basiselm}
    B=\{v_{r,n} \,|\, r\in R, n\in N\}\,.
\ee
which implies that the dimension of the condensable algebra $A(H,N,\gamma,\epsilon)$ in $\cZ(\Vec_G^\omega)$ is
\be
	\dim(A(H,N,\gamma,\eps))=\frac{|G|\times|N|}{|H|}\,.
\ee 
The $G$-grading is 
\be
    |v_{g,n}| = gng^{-1},
\ee
and the $G$-action is
\be
\label{eqn:G-action_algebra_basis}
    g'\cdot v_{g,n} = \omega( g',g|n )^{-1}\, v_{g'g,n},
\ee
where $g,g'\in G$ and $n,n'\in N$.  Multiplication is defined on basis $B$ by
\be
\label{eqn:cond_alg_multiplication}
    v_{r,n}v_{r',n'} = \delta_{r,r'} \frac{\gamma(n,n')}{\Omega_r(n,n')} v_{r,nn'}\,,
\ee
where $r,r'\in R$. Here, $\Omega_g:N \times N \rightarrow U(1)$ trivializes\footnote{That is, $d\Omega_g=\omega/\omega^{g}$. We introduce notations: ${}^gx:=gxg^{-1}$ and $f^g(x_1, \cdots, x_n):=f({}^gx_1, \cdots, {}^gx_n)$ for $n\in \mathbb{N}$, group elements $g,x,x_i\in G$, $i\in \{1, \cdots, n\}$ and function $f$ on $G^n$.} $\omega/\omega^{g}$, whose explicit formula in terms of $\omega$ can be found in \eqref{eq:Omega_f}. The multiplication \eqref{eqn:cond_alg_multiplication} and $G$-action \eqref{eqn:G-action_algebra_basis} specify the algebra structure on $A(H, N, \gamma, \eps)$ in $\cZ(\Vec_G^\omega)$.

The dimension of $A(H,N,\gamma,\epsilon)$ is $|G||N|/|H|$.
    For a fixed $R$, $\{v_{r,1} \,|\, r\in R\}$ forms a complete system of orthogonal idempotents in $A(H, N, \gamma, \epsilon)$.

However, this description contains redundancies. In applications, identifying non-isomorphic algebras is key, e.g., when classifying distinct phases of matter. 
Some of the redundancies were pointed out in \cite{davydov2017lagrangian} as follows:
\begin{lemma}[\cite{davydov2017lagrangian}, Lemma 3.12]
\label{lem:DS-isom-gamma-epsilon}
    Given two condensable algebras $A(H,N, \gamma,\epsilon)$ and $A(H,N,\gamma',\epsilon')$ in $\cZ(\Vec_G^\omega)$, if there exists $c:N \rightarrow U(1)$ such that
    \begin{align}
        \gamma'(n_1,n_2) &= \frac{c(n_1 n_2)}{c(n_1)c(n_2)} \gamma(n_1,n_2), \label{eqn:isom_alg_gamma}\\
        \epsilon'(h,n) & = \epsilon(h,n)\frac{c({}^{h}n)}{c(n)}, \label{eqn:isom_alg_epsilon}
    \end{align}
    for any $n_1,n_2,n\in N$, $h\in H$ and ${}^{h}n:=hnh^{-1}$,
    then $A(H,N, \gamma,\epsilon)$ and $A(H,N,\gamma',\epsilon')$ are isomorphic algebras.
\end{lemma}
From now on, we always take such identification into account. The possible $\gamma$'s in $A(H,N,\gamma,\eps)$ form an $H^2(N, U(1))$-torsor. The above still does not identify all the isomorphic algebras, which we address next.

\subsubsection{Isomorphism Classes of Condensable Algebras}
\label{sec:IsomAlg}

We now identify the complete set of isomorphism classes of condensable algebras in $\cZ(\Vec_G^\omega)$, with proof given in \Cref{append:alg_Morita_equiv}. This will then allow us to identify distinct phases in applications to quantum systems. 

\begin{lemma}[Isomorphisms of Condensable Algebras]
\label{lem:Morita_equiv}
    Two condensable algebras $A(H,N,\gamma,\epsilon)$ and $A(H',N',\gamma',\epsilon')$ in $\cZ(\Vec_G^\omega)$ are isomorphic, if and only if there exists $g\in G$ such that 
    \be\ba 
        H' & = {}^g\!H, & \gamma'({}^g n_1, {}^g n_2) & = \frac{\gamma(n_1,n_2)}{\Omega_g(n_1,n_2)},\\
        N' & = {}^g\!N, & \epsilon'({}^{g}h, {}^{g}n) & =\frac{\Omega_g({}^{h}n,h)}{\Omega_g(h,n)}\epsilon(h,n),\label{eqn:MoritaEq}
    \ea\ee
    where we use notation: ${}^{g}x = gxg^{-1}$, ${}^{g}\!H = gHg^{-1}$, for $h\in H$ and $n_1,n_2,n\in N$ and $\Omega_g$ is defined as 
    \be
    \Omega_f(g,h)=\frac{\omega(f,g,h)\omega({}^{f}\!g,{}^{f}h,f)}{\omega({}^{f}\!g,f,h)}, \label{eq:Omega_f}
    \ee
\end{lemma}
\begin{proof}
    See Appendix~\ref{append:alg_Morita_equiv}.
\end{proof}
The relation for Lagrangian algebras is then simply: 
\begin{corollary}[Isomorphism of Lagrangian Algebras]
    Lagrangian algebras $\cL(H, \gamma)$ and $\cL(H', \gamma')$ in $\cZ(\Vec_G^\omega)$ are isomorphic if and only if $H'={}^{g}\!H$ and  $\gamma'=\gamma^{g^{-1}}\Omega_{g^{-1}}$ for some $g\in G$. 
\end{corollary}
This agrees with the classification of indecomposable module categories over group-theoretical fusion categories \cite{Natale2017} (also see \cite{Ostrikmodule}).

\subsection{Anyon Decomposition of Algebras}  
\label{sec:AnyonDecomp}

For applications, e.g. to the classification of symmetric phases, it is informative to have a decomposition of the algebras in terms of anyons. For instance, they correspond to the generalized charges. 
Here, we provide a formula to decompose each condensable algebra $A(H,N,\gamma,\eps)$ into simple anyons. This is achieved using characters of objects in $\cZ(\Vec_G^\omega)$, studied in \cite{Bantay:1993ku, davydov2017lagrangian}.

Recall that objects in $\cZ(\Vect_G^\omega)$ are $G$-graded vector spaces with $\omega$-projective $G$-actions \cite{davydov2017lagrangian}. A simple object in $\cZ(\Vec_G^\omega)$ is labeled by $([x],\rho)$, where $[x]\subset G$ is a conjugacy class and $\rho$ an irreducible (normalized) $(\omega(\cdot,\cdot|x)^{-1})$-projective representation ($\rho(g)\rho(h)=\omega(g,h|x)^{-1}\rho(gh)$) of the centralizer $C_G(x)$ (of a representative\footnote{For each $f\in[x]$, denote by $\rho_f$ the representation of $C_G(f)$ on the $f$-graded component. For $f={}^kx=kxk^{-1}$ and $g\in C_G(f)$, the characters of $\rho_f$ are related to those of $\rho_x$ by $\Tr\rho_f(g)=\omega(k,k^{-1}gk|x)\,\omega(g,k|x)^{-1}\,\Tr\rho_x(k^{-1}gk)$, see also \cite[eq.~(2.10)]{gruen2021computing}.} $x\in [x]$).

\begin{proposition}[Anyon Decomposition]
\label{prop:anyon}
    The condensable algebra $A(H, N, \gamma, \epsilon)$ in $\cZ(\Vec_G^\omega)$ has decomposition:
    \be
        A(H,N,\gamma,\epsilon) \cong \bigoplus_{([x],\rho)} n_{([x],\rho)} ([x],\rho),
    \ee
    with multiplicity of simple object $([x], \rho)$ being
    \be
        n_{([x],\rho)} = \!\!\!\!\!\sum_{\substack{f\in [x],\\ g\in C_G(f)}}\!\!\! \biggl\{ \frac{\ol{\Tr(\rho_{f}(g))}}{|G|} \!\!\!\!\!\!\!\!\!\sum_{\substack{y\in R,\\ {}^{\bar{y}}f\in N,g^y\in H}}\!\!\!\!\!\!\frac{\omega({}^{\bar{y}}g,y^{-1}|f)}{\omega(y^{-1},g|f)}\epsilon({}^{\bar{y}}g,{}^{\bar{y}}f)\biggr\},\label{eqn:scalar_prod_anyon_decom}
    \ee
    with $\ol{\Tr(\rho_{f}(g))}$ the complex conjugation of $\Tr(\rho_{f}(g))$ and $\bar y=y^{-1}$.
\end{proposition}

\begin{proof}
    Recall from \cite{davydov2017lagrangian} that for any object $V$ in $\cZ(\Vec_G^\omega)$, its character is defined on commuting pairs of elements $f,g\in G$, $fg=gf$, to be
    \be
        \chi_V(f,g) = \Tr_{V_f}(g)\,,
    \ee
    where $V_f$ denotes the $f$-graded component of $V$, and it takes the value of zero on non-commuting elements.
    The scalar product of characters $\chi_V, \chi_W$ of objects $V,W$ of $\cZ(\Vec_G^{\omega})$ is defined by 
    \be
        \langle \chi_V, \chi_W \rangle := \frac{1}{|G|}\sum_{\substack{f,g\in G\\ fg=gf}} \ol{\chi_V(f, g)}\; \chi_W(f,g)
    \ee
and    satisfies the orthogonality relation 
    \be
    \label{eqn:orthogonality_character_Drinfeld_center_of_group}
        \langle \chi_V, \chi_W\rangle = \dim \lb \Hom_{\cZ(\Vec_G^{\omega})}(V,W) \rb \,.
    \ee
    Applying \cite[Lemma 5.6]{davydov2017lagrangian}, we can compute the character of the condensable algebra $A(H,N,\gamma,\epsilon)$ in $\cZ(\Vec_G^\omega)$ to be 
    \be \label{eqn:character_cond_alg}
    	\chi_{A(H,N,\gamma,\epsilon)}(f,g)=\sum_{\substack{y\in R,\\ {}^{\bar{y}}f\in N,\,{}^{\bar{y}}g\in H}}\frac{\omega({}^{\bar{y}}g,y^{-1}|f)}{\omega(y^{-1},g|f)}\epsilon({}^{\bar{y}}g,{}^{\bar{y}}f)\,,
    \ee
    for $f,g\in G$ such that $fg=gf$, and $R$ any complete set of cosets of $G/H$. The character vanishes on pairs of non-commuting group elements. The orthogonality relation \eqref{eqn:orthogonality_character_Drinfeld_center_of_group} leads us to expression~\eqref{eqn:scalar_prod_anyon_decom} as claimed.
\end{proof}

\begin{example}[Abelian Groups]
    To illustrate the above formula \eqref{eqn:scalar_prod_anyon_decom}, we use it to decompose any condensable algebra $A(H,N,\gamma,\epsilon)$ in $\cZ(\Vec_A)$, where $A$ is a finite abelian group, as follows:
    \be
    \label{eqn:abelian_group_trivial_twist_cond_alg}
        A(H,N,\gamma,\epsilon) = \bigoplus_{\substack{a\in N, \, \wh{a}\in \wh{A},\\
        \wh{a}|_H=\epsilon(\cdot, a)}} (a, \wh{a})\,.
    \ee
    This follows by starting with the general expression
    \be
        A(H,N,\gamma,\epsilon) = \bigoplus_{(a,\wh{a})\in A\times \wh{A}} n_{a,\wh{a}}(a,\wh{a}),
    \ee
    and applying equation \eqref{eqn:scalar_prod_anyon_decom}, which simplifies to
    \be
        n_{a,\wh{a}} = \begin{cases}
            0, & a\notin N, \\
            \langle\wh{a}, \epsilon(\cdot, a)\rangle_H, & a\in N,
        \end{cases}
    \ee
    where $\langle\cdot, \cdot\rangle_H$ denotes the scalar product of (class) functions\footnote{More explicitly, let $\chi,\psi:H\rightarrow \bC$ be two (class) functions, define scalar product $\langle\chi, \psi\rangle_{H}:=1/|H|\sum_{h\in H}\chi(h^{-1})\psi(h)$.} on $H$. Hence equation \eqref{eqn:abelian_group_trivial_twist_cond_alg} holds.
\end{example}

\begin{example}[Lagrangian algebras]
    Note that for a Lagrangian algebra $\cL(H, 1)$ in $\cZ(\Vec_G)$, the multiplicity $n_{([1],\rho)}$ of anyon $([1],\rho)$, where $\rho$ denotes an irreducible representation of $G$, agrees with the multiplicity of the trivial representation of $H$ in $\text{Res}^G_H \rho$. This follows from:
    \begin{equation}
        \begin{split}
            n_{([1],\rho)} & = \frac{1}{|G|}\sum_{g\in G} \overline{\Tr(\rho(g))} \sum_{\substack{y\in R\\{}^{\bar{y}}g\in H}}1\\
            & = \langle \chi_\rho, \text{Ind}^G_H 1 \rangle_{G} \\
            & = \langle \text{Res}^G_H\chi_\rho, 1\rangle_H,
        \end{split}
    \end{equation}
    where $\chi_\rho$ denotes the character of $\rho$ and the final line follows from Frobenius reciprocity.

    For more specific examples,
    \be 
     \cL(1,1)\cong\!\!\bigoplus_{r_i \in \Rep (G)} \text{dim} (r_i) ([1], r_i)
    \ee
    in $\cZ(\Vec_G^\omega)$
    defines the Lagrangian algebra for the Dirichlet boundary for $\Vec_G^\omega$ symmetry, and
    \begin{equation}
        \cL(G,1)\cong \bigoplus_{[g]}([g],1)
    \end{equation}
    in $\cZ(\Vec_G)$ defines the Neumann symmetry boundary for $\Rep(G)$ symmetry.
\end{example}

\subsection{Twin Algebras}
\label{sec:TwinAlgebras}

It is rather common to refer to algebras simply in terms of their anyon decomposition. However, there can be algebras whose anyon decomposition is the same, but whose multiplicative structure is different. I.e., 
there can be non-isomorphic condensable algebras but isomorphic as objects:
\begin{definition}
\label{defn:twins}
    Two condensable algebras $A_1$ and $A_2$ in a modular tensor category $\cB$ are called {\bf twin algebras} (or simply {\bf twins}) if $A_1\cong A_2$ as objects but $A_1\not \cong A_2$ as algebras in $\cB$.
\end{definition}
Twin algebras exist in $\cZ(\Vec_G^\omega)$, both Lagrangian and non-maximal condensable algebras. 
We provide a characterization of twins in terms of the classification data $A (H, N, \gamma, \epsilon)$, in \Cref{sec:examples-of-twin-algs}.

\subsection{Interfaces from Condensable Algebras}
\label{sec:Fold}

One of the main applications of non-maximal condensable algebras is that they define topological interfaces between topological orders. 
Here we will describe interfaces $\cI_{A(H,N,\gamma,\epsilon)}$ associated to $A(H, N, \gamma, \epsilon)$, which define interfaces from 
$\cZ(\Vec_G^{\omega})$  to the reduced topological order, which is given by $\cZ(\Vec_{H/N}^{\alpha})$ by Theorem~\ref{thm:DS-reducedTO}: 
\be\label{GLAGR}
\begin{tikzpicture}[baseline={(current bounding box.center)}]
\begin{scope}[shift={(0,0)}]
\draw[SymTFTColores, fill= SymTFTColores, opacity = 0.2]   (0,0) -- (0,2) -- (2,2) -- (2,0) --(0,0); 
\draw[SymTFTColores, fill= SymTFTColores, opacity = 0.1]  (2,0) -- (2,2) -- (4,2) -- (4,0) --(2,0); 
\draw [very thick]  (2,0) -- (2,2); 
\node[above] at (2,2) {$\cI_{A(H, N, \gamma, \epsilon)}$};
\node at (1,1) {$\cZ(G,{\omega})$}; 
\node at (3,1) {$\cZ({H/N},{\alpha})$};
\end{scope}
\end{tikzpicture}
\ee
A description in terms of anyons was given in many instances in \cite{Beigi:2010htr,Chatterjee:2022tyg,Bhardwaj:2023bbf,Bhardwaj:2024qrf,Bhardwaj:2025jtf}. 
Anyon $([x], \rho)$ in $\cZ(G, \omega)$ can confine on the interface $\cI_{A(H, N, \gamma, \epsilon)}$ if 
\begin{itemize}
    \item $[x]\cap H=\emptyset$ and there is $([1],\rho_A)\in A(H, N, \gamma, \eps)$ such that $\rho_A(x)\neq1$; or
    \item $(\Ind^G_{C_G(x)}\rho)|_{N}\neq 1$, more precisely, if there is $n\in N$ such that $([n],\rho_A)\in A(H, N, \gamma,\eps)$ and $\Ind^G_{C_G(x)}\rho(n)\neq 1$.
\end{itemize}
By ``folding'' along the interface $\cI_{A(H, N, \gamma, \epsilon)}$, we can associate to the condensable algebra $A(H, N, \gamma, \epsilon)$
a Lagrangian algebra in the folded setup 
\be\ba \label{eq:Afolded}
    \cL^\folded(H^\diag,\vp) & \in \cZ(\Vec_G^{\omega})\boxtimes\overline{\cZ(\Vec_{H/N}^{\alpha})}\\
    & \cong \cZ(\Vec_{G\times H/N}^{\omega\times\overline{\alpha}})\,,
\ea\ee
where
\begin{itemize}
    \item denoting by $p:H\rightarrow H/N$ the projection, 
    \be \label{eq:Hdiag}
       H^\diag=\{(h,p(h))\;:\;h\in H\}\,.
    \ee
    
    \item $\vp$ is a 2-cocycle on $H^\diag$ such that $A^\folded(H^\diag, \varphi)$ contains
    \be\ba \label{eq:alg_to_include}
    A(H, N, \gamma, \epsilon) \boxtimes A(H/N, 1, 1, 1) \\
    = A(H \times H/N, N \times 1, \wt{\gamma}, \wt{\epsilon}),
    \ea\ee
    as a subalgebra. Here, $A(H/N, 1, 1, 1)$ is the trivial condensable algebra in $\overline{\cZ(\Vec_{H/N}^{\alpha})}$
    and 
    \be\ba
    \wt{\gamma}((n_1,1), (n_2,1)) &= \gamma(n_1,n_2)\,,\\
    \wt{\epsilon}((h_1, p(h_2)), (n,1)) &= \epsilon(h_1, n)\,,
    \ea\ee
    for $n_1,n_2,n\in N$, $h_1,h_2\in H$. 
\end{itemize}
Note that $\omega\times\overline{\alpha}$ trivializes on $H^\diag$ because $p^*\alpha$ is equivalent to $\omega|_H$ \cite{Hannah:2023tae}. Hence $A^{\folded}(H^{\diag}, \varphi)$ exists.

Lagrangian algebras in the folded theory corresponding to a condensable algebra are not unique. They can be related via an autoequivalence on the reduced TO. For example, consider the subgroup 
\be
    H^2(H/N, U(1)) \subset \text{Aut}^{\text{br}}(\cZ(\Vec_{H/N}^\alpha)),
\ee
an element $\phi$ in this subgroup $H^2(H/N, U(1))$ induces a 2-cocycle $p^*\phi$ on $H^{\diag}$ via pullback along the projection $p:H^{\diag} \rightarrow H/N$. $p^*\phi$ defines an autoequivalence that relates the two Lagrangian algebras 
\be
    \cL(H^{\diag}, \varphi), \quad \cL(H^{\diag}, \varphi (p^*\phi)),
\ee
in the folded theory $\cZ(\Vec_{G\times H/N}^{\omega\times \overline{\alpha}})$.

\vspace{2mm}
\noindent{\bf Map of Anyons.} As mentioned, a condensable algebra defines a map on anyons from the reduced TO to the original TO, which can be read off from the folded Lagrangian. The folded Lagrangian $\cL^\folded(H^\diag,\vp)$ can be decomposed in terms of the anyons using \eqref{eqn:scalar_prod_anyon_decom}
\be \label{eq:L_folded}
\cL^{\folded} (H^\diag, \vp) = \bigoplus_{\substack{a \in \cZ(\Vec_G^{\omega})\\ b \in \cZ(\Vec_{H/N}^{\alpha})}} n_{ab} \;{a} \otimes \ol{b} \,.
\ee
The anyon decomposition of $A^\folded(H^\diag,\vp)$ encodes a map of anyons from the reduced TO $\cZ(\Vec_{H/N}^{\alpha})$ to the original TO $\cZ(\Vec_G^{\omega})$, 
\begin{align}
    \cZ(\Vec_{H/N}^{\alpha}) & \rightarrow \cZ(\Vec_G^{\omega})\nn\\
    b & \mapsto a \,.
\end{align}
The anyons in $\cZ(\Vec_G^\omega)$ that do not appear in the image of this map are confined and do not pass through the interface: for example, the anyons $([x],\rho)$ in $\cZ(\Vec_G^\omega)$ such that $[x]\cap H = \emptyset$ are confined.
Note that as an equivalence of modular tensor categories, such a map preserves spins and braidings, i.e. it needs to be compatible with modular data as discussed in \cite{Chatterjee:2022tyg,Ng:2023wsc}, in particular:
\be\ba
    Sn=nS'\,,\quad Tn=nT'\,,
\ea\ee
where $S,T$ and $S',T'$ denote the modular matrices of $\cZ(\Vec_G^\omega)$ and $\cZ(\Vec_{H/N}^{\alpha})$ respectively and $n$ is the matrix of non-negative integer coefficients appearing in \eqref{eq:L_folded}.

Note that the compatibility conditions on objects and modular data do not capture the full algebra structure of condensable algebras, and the modular data does not fully specify an MTC \cite{mignard2021modular,Delaney:2021owx}, which is why we impose the stronger requirement of algebra inclusion of \eqref{eq:alg_to_include} in $A^\folded(H^\diag,\vp)$.

\subsection{Partial Order and Hasse Diagram}
\label{sec:Hasse}

Condensable algebras for a fixed topological order have a partial order on them. This has surprising connections to the structure of phases and phase-transitions, and provides a useful way to characterize all symmetric phases for a given  symmetry \cite{Bhardwaj:2024qrf}. In particular given that it is independent of the symmetry boundary and only depends on the bulk topological order, it is the same Hasse diagram for all Morita equivalent symmetries. 

We now define the partial order on condensable algebras in $\cZ(\Vec_G^\omega)$ and the resulting Hasse diagram: 
\begin{proposition}
\label{prop:partial-order}
    There is a  partial order on condensable algebras in $\cZ(\Vec_G^\omega)$ defined by: 
\be 
A(H, N, \gamma, \epsilon ) \preceq A(H', N', \gamma', \epsilon ') 
\ee
if     
\be \label{eqn:partial_order_N}
H = H'\,,\quad N < N' \,,\quad  
\gamma = \gamma'|_{N} \,, \quad \epsilon= \epsilon'|_{H \times N}
\ee
or 
\be\label{eqn:partial_order_H}
H' < H \,,\quad N= N' \,,\quad \gamma= \gamma' \,,\quad \epsilon' = \epsilon|_{H'\times N'} \,.
\ee
\end{proposition}
\begin{proof}
    See Appendix~\ref{append:partialorder}.
\end{proof}
A useful depiction of the partial order is the Hasse diagram of condensable algebras \cite{Bhardwaj:2024qrf}, which uses the partial order and arranges the algebras by quantum dimension 
where 
\be
\dim(A (H, N, \gamma, \epsilon))={|G||N|\over |H|} \,.
\ee
The Lagrangian algebras, as the category, have quantum dimension $|G|$. The Hasse diagram is obtained by arranging  algebras by quantum dim along the vertical axis, starting at the top with $A(G,1,1,1)= 1$ and the Lagrangian algebras at the bottom, and two algebras are connected, if they satisfy $A \preceq A'$.

There are interesting features of the Hasse diagram, which have deep implications in the structure of phases and transitions. 

An example considers the application to direct phase transitions: 
Define $A (H, N, \gamma, \epsilon)$ as a \emph{sub-Lagrangian algebra}, if it is non-Lagrangian,   $\dim (A (H, N, \gamma, \epsilon)) < |G|$, but there are no non-Lagrangian algebras that it is contained in. Then:
\begin{corollary}
Any sub-Lagrangian algebra is contained in one or two Lagrangian algebras. If it is only contained in one Lagrangian, the reduced topological order is anomalous. 
\end{corollary}
\begin{proof}
    Fix any non-Lagrangian algebra $A=A(H, N, \gamma, \eps)$ in $\cZ(\Vec_G^\omega)$. As a consequence of Proposition~\ref{prop:partial-order}, it is always a subalgebra of the Lagrangian algebra of the form $\cL(N, \gamma)$.  
    If $A$ is also a subalgebra of another Lagrangian algebra $\cL(H', \gamma')$ for some subgroup $H'\subset G$ such that $N \subsetneq H'$. We have $\gamma'|_N=\gamma$. Then by Proposition~\ref{prop:partial-order}, $A$ is also a subalgebra of $A(H', N, \gamma, \eps'|_{H'\times N})$, where $\eps'$ is fully determined by $\gamma'$ via \eqref{eps_from_gamma}. 
By assumption $A$ is a sub-Lagrangian algebra, and thus is not contained in any other non-Lagrangian algebra, either: 

\begin{itemize}[itemsep=3.5pt,topsep=2ex, parsep=2pt, partopsep=5pt,leftmargin=2.ex]
    \item $\cL(H', \gamma')$ does not exist, in which case $A$ is a subalgebra of a single Lagrangian algebra. In this case, the reduced TO has a unique Lagrangian algebra, hence the anomaly $\alpha\in H^3(H/N, U(1))$ (expressed in terms of $\omega, \gamma, \eps$ in \eqref{eqn:twist_reduced_TO}) appearing in the reduced TO $\cZ(\Vec_{H/N}^{\alpha})$ is non-trivial. Or:
    \item If $\cL(H', \gamma')$ exists, then $H'=H$, $A$ is a subalgebra of two  Lagrangian algebras: $\cL(N, \gamma)$ and $\cL(H', \gamma')$.
\end{itemize}
\end{proof}
This means the bottom of the Hasse diagram -- comprised of the sub-Lagrangian and Lagrangian algebras  --  has only two types of connections:
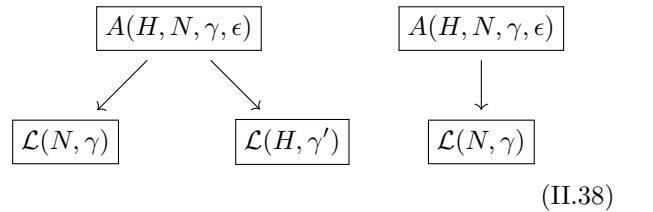
\be
   \begin{tikzpicture}
   \begin{scope}
        \node[] (n1) at (0.00,0.00) {$\boxed{A(H, N, \gamma, \epsilon)}$};
        \node[] (n2) at (-1.5,-1.5) {$\boxed{\cL(N, \gamma)}$};
        \node[] (n3) at (1.5,-1.5) {$\boxed{\cL(H, \gamma')}$};
        \draw[->] (n1) -- (n2);
        \draw[->] (n1) -- (n3);
        \end{scope}
     \begin{scope}[shift= {(4,0)}]
        \node[] (n1) at (0.00,0.00) {$\boxed{A(H, N, \gamma, \epsilon)}$};
        \node[] (n2) at (0,-1.5) {$\boxed{\cL(N, \gamma)}$};
        \draw[->] (n1) -- (n2);
        \end{scope}      
    \end{tikzpicture}
\ee
If a sub-Lagrangian algebra is contained in two Lagrangians, it has an interpretation in terms of a direct phase transition between the gapped phases, that these two Lagrangian algebras represent.

If the sub-Lagrangian algebra is contained in only one Lagrangian algebra  then its reduced TO $\cZ(\Vec_{H/N}^{\alpha})$ has non-trivial anomaly $\alpha$ \eqref{eqn:twist_reduced_TO} and in fact it is possible to interpret this (with a particular choice of symmetry) as an igSPT \cite{Wen:2023otf,Bhardwaj:2024qrf}.  It only admits two condensable algebras: the trivial algebra and the Lagrangian algebra $\cL(1,1)$. Equivalently, $\Vec_{H/N}^{\alpha}$ only admits a single topological SSB phase. Hence if the sub-Lagrangian algebra of interest results in a gapless SPT with respect to some symmetry, then it is intrinsically gapless.

\section{Twin Algebras}
\label{sec:examples-of-twin-algs}

We discussed the comprehensive classification (and crucially, the identification up to isomorphism) of condensable algebras in $\cZ(\Vec_G^\omega)$. Also we defined the concept of twin algebras in \Cref{sec:TwinAlgebras}. These algebras are thus far only sporadically studied. 
Given the classification of algebras in terms of $(H, N, \gamma, \epsilon)$, we can divide them as follows -- most of these have to our knowledge not been discussed in the literature before:

\begin{enumerate}
\item  {\bf Twins with non-conjugate subgroups $H$.} Any pair of condensable algebras with the same anyon decomposition but whose subgroups are not conjugate, define a pair of twin algebras. One family consists of algebras with non-conjugate subgroups $H\subset G$, constructed from so-called Gassmann triples, which we study in \Cref{sec:Gassmann_Triple_twin_alg}. These new twin algebras can be either maximal or non-maximal (non-Lagrangian).  We call these {\bf twins of $H$-type}.

\item {\bf Twins with non-conjugate subgroups $N_i$.}
In turn we could have the same group $H$, but normal subgroups $N_1,N_2\triangleleft H$ that are not related by conjugation with any element of the normalizer $N_G(H)=\{g\in G : gHg^{-1}=H\}$. The $N_i$ may still be conjugate in $G$. For such twins, the subgroups $N_i$ necessarily also form a (possibly trivial) Gassmann triple. Again, such examples have not been discussed to our knowledge. We call these {\bf twins of $N$-type}.

\item {\bf Twins with inequivalent SPTs $\gamma$.}
Fixing the same $H, N$ but allowing $\gamma$ to differ, there can be two types of twin algebras: First type contains Lagrangian algebras, which have been previously considered, whose 2-cocycles $\gamma$ differ by a Bogomolov multiplier \cite{Pollmann:2012,Davydov:2013xov, Cong:2017ffh, Kobayashi:2025ykb,Kobayashi:2025pxs}. We review this in \Cref{sec:SPT-twins}. 
Secondly, we find new, non-maximal condensable algebras for anomalous symmetries whose 2-cocycles do not differ by a Bogomolov multiplier. We call these {\bf twins of $\gamma$-type}.

\item {\bf Twins with inequivalent $H$-action phase $\epsilon$.}
Finally, in \Cref{sec:twin_alg_epsilon}, we construct twin algebras that differ by $\epsilon$ but have the same $H,N,\gamma$. These new twin algebras are non-maximal. We call these {\bf twins of $\epsilon$-type}.
\end{enumerate}

We will present detailed examples for each family of twin algebras in \Cref{sec:DoubTroub}.

\subsection{Twins of $H$-Type}
\label{sec:Gassmann_Triple_twin_alg}

The first class of twin algebras results from a choice of non-conjugate subgroups $H$. Our construction relies on so-called Gassmann triples, which allow us to identify both Lagrangian and non-maximal condensable twin algebras. {In fact, we show that the subgroups for twin algebras necessarily form Gassmann triples.} Examples will be discussed in \Cref{sec:Ex-Htype}.

\subsubsection{General Considerations}

Consider a group $G$ and two non-conjugate subgroups $H_1$ and $H_2$ of $G$. 
Consequently by Lemma~\ref{lem:Morita_equiv} it follows that  $A(H_1,N_1,\gamma_1,\eps_1)$ and $A(H_2, N_2, \gamma_2, \eps_2)$ in $\cZ(\Vec_G^\omega)$ are not isomorphic condensable algebras. 
Here we are interested in the question when they have the same anyon decomposition, i.e. the same character~\eqref{eqn:character_cond_alg}.
{In general this condition can be satisfied in various ways. However, if $\omega=1$ then any Lagrangian algebra with trivial $\gamma$ arising from non-conjugate subgroups $H_1$ and $H_2$ has to come from a so-called Gassmann triple -- which we will discuss in the next section.}


{We find the following: 
if $N=1$  then with or without anomaly $\omega$, $H_i$ have to form a Gassmann triple, in order for the associated algebras be twins. 
For non-trivial $N$, then both $N$ and $H$ have to form Gassmann triples.}
We will provide more precise statements in the following, in particular Proposition \ref{prop:HNGassmann} and Proposition \ref{prop:Gassmann}.

\subsubsection{Gassmann Triple Twin Algebras} \label{sec:Gassman}

One approach to construct twin algebras from non-conjugate subgroups uses so-called Gassmann triples.
 The non-maximal ones provide examples of twins with possibly inequivalent reduced TOs, see \Cref{sec:example_Gassmann_Triple_different_reduced_TOs} for examples.

Let us review the basic notions of Gassmann triples first, which appeared in the study of arithmetically equivalent number fields in \cite{Gassmann:1926}. We will adopt the following definition \cite{Sunada1985RiemannianCA}.

\begin{definition}
    \label{def:GassmannTriple}
    Let $G$ be a finite group, $H_1,H_2\subset G$ two subgroups. The triple $(G, H_1, H_2)$ is called a \textbf{Gassmann triple} if 
    \be
        |H_1\cap [g]|=|H_2\cap [g]|
    \ee
    for all conjugacy classes $[g]\subset G$. A Gassmann triple $(G,H_1,H_2)$ is said to be \emph{non-trivial} if $H_1$ and $H_2$ are not conjugate to each other. 
\end{definition}
Equivalently, $(G, H_1, H_2)$ is a Gassmann triple if and only if $\Ind^G_{H_1}1\cong \Ind^G_{H_2}1$ as $G$-representations.
\begin{remark}
\label{rmk:Hallsmarriage}
    Note that for a Gassmann triple $(G,H_1,H_2)$, $H_1$ and $H_2$ are of the same order, hence the Hall's marriage theorem \cite{hall_combinatorial_1998} guarantees a common set $R$ of coset representatives for both $G/H_1$ and $G/H_2$ to exist.
\end{remark}

{Gassmann triples are closely related to twin algebras in $\cZ(\Vec_G^\omega)$. In fact, twin algebras necessarily carry information of Gassmann triples.
\begin{proposition}
\label{prop:HNGassmann}
    If two condensable algebras $A(H_i,N_i,\gamma_i,\epsilon_i)$, $i=1,2$, in $\cZ(\Vec_G^\omega)$ are twins, then $(G, H_1, H_2)$ and $(G, N_1, N_2)$ are necessarily (possibly trivial) Gassmann triples.
\end{proposition}
\begin{proof}
    Twin algebras have the same anyon decompositions, equivalently, they have the same characters. That is, $\chi_{A(H_1,N_1,\gamma_1,\epsilon_1)}(f,g)=\chi_{A(H_2,N_2,\gamma_2,\epsilon_2)}(f,g)$ for any commuting $f,g\in G$. 
    For $i=1,2$, fix $R_i$ a set of representative of $G/H_i$.
    Evaluate the character~\eqref{eqn:character_cond_alg} on $(1,g)$ for $g\in G$,
    \begin{equation}
        \begin{split}
            \chi_{A(H_i, N_i, \gamma_i, \epsilon_i)}(1, g) & = \sum_{\substack{y\in R_i\\ {}^{\bar{y}}g\in H_i}} 1 \\
            & = \frac{|[g]\cap H_i|\times |C_G(g)|}{|H_i|} \nn\\
            & = \chi_{\text{Ind}^G_{H_i}1}(g).
        \end{split}
    \end{equation}
    Requiring $\chi_{A(H_1, N_1, \gamma_1, \epsilon_1)}(1, g)=\chi_{A(H_2, N_2, \gamma_2, \epsilon_2)}(1, g)$ for all $g\in G$ is the same as $\Ind^G_{H_1}1\cong \Ind^G_{H_2}1$ as $G$-representations, hence $(G, H_1, H_2)$ forms a Gassmann triple.

    Similarly, on $(g,1)$, 
    \begin{equation}
        \begin{split}
            \chi_{A(H_i, N_i, \gamma_i, \epsilon_i)}(g, 1) & = \sum_{\substack{y\in R_i\\ {}^{\bar{y}}g\in N_i}} 1 \\
            & = \frac{|N_i|}{|H_i|} \chi_{\text{Ind}^G_{N_i}1}(g).
        \end{split}
    \end{equation}
    Being twin algebras, two algebras need to be isomorphic as $G$-graded vector spaces. This implies $|N_1|=|N_2|$ and hence $|H_1|=|H_2|$. Requiring $\chi_{A(H_1, N_1, \gamma_1, \epsilon_1)}(g, 1)=\chi_{A(H_2, N_2, \gamma_2, \epsilon_2)}(g, 1)$ for all $g\in G$ implies $\Ind^G_{N_1}1\cong \Ind^G_{N_2}1$ as $G$-representations. We conclude $(G, N_1, N_2)$ is a Gassmann triple.
\end{proof}
}

{The previous result says twin algebras are necessarily labeled by Gassmann triples. For the converse direction,} each non-trivial Gassmann triple $(G, H_1, H_2)$ defines a pair of non-maximal twin algebras of the form $A(H_i,1,1,1)$ in $\cZ(\Vec_G^\omega)$, whose anyon decompositions are given by the induced representations\footnote{Given a representation $\rho$ of $H\subseteq G$, one can define the induced representation $\text{Ind}^G_H \rho = \bC[G]\otimes_{\bC[H]}\rho$ of $G$ \cite{serre_linear_1977}. In terms of simple $G$-representations $\{\pi_i\}_i$, $\text{Ind}^G_H \rho \cong \bigoplus_{i}n_i \pi_i$ with multiplicity $n_i\in \Z_{\geq 0}$ can be computed by the Frobenius reciprocity $n_i = \langle \pi_i, \text{Ind}^G_H \rho\rangle_G = \langle\text{Res}^G_H \pi_i, \rho\rangle_H$. Here $\text{Res}^G_H $ denotes the restriction of a $G$-representation to $H$.} $\text{Ind}^G_{H_i}1$, $i=1,2$. The precise meaning will be made explicit in the proof of the following proposition.

\begin{proposition}[Nonmaximal Gassmann-triple twin algebras]
\label{prop:Gassmann}
    Let $G$ be a finite group and $\omega\in H^3(G, U(1))$.
    Two non-maximal condensable algebras $A(H_1,1,1,1)$ and $A(H_2,1,1,1)$
    in $\cZ(\Vec_G^\omega)$ are twin algebras if and only if 
    $(G,H_1,H_2)$ is a non-trivial Gassmann triple. 
\end{proposition}
\begin{proof}
    Let $H\subset G$ be any subgroup of $G$, fix a set $R$ of representative of coset $G/H$.
    For $A(H,1,1,1)$ in $\cZ(\Vec_G^\omega)$, $\omega$ assumed to be normalized, the $\omega$- and $\epsilon$-phases in its character $\chi_{A(H,1,1,1)}$ formula~\eqref{eqn:character_cond_alg} are trivial. $\chi_{A(H,1,1,1)}$ evaluates to the following expression which is closely related to $\chi_{\text{Ind}^G_{H}1}$, character of the induced $G$-representation $\text{Ind}^G_{H}1$:
    \begin{align}
        \chi_{A(H,1,1,1)}(f,g) & = \delta_{f,1} \sum_{\substack{y\in R,\\ {}^{\bar{y}}g\in H}}1 \nn\\
        & = \delta_{f,1} \frac{|[g]\cap H|\times |C_G(g)|}{|H|} \nn\\
        & = \delta_{f,1} \chi_{\text{Ind}^G_{H}1}(g),
    \end{align}
    where $f,g\in G$ are such that $fg=gf$. From definition of a Gassmann triple $(G, H_1, H_2)$, $\text{Ind}^G_{H_1}1\cong \text{Ind}^G_{H_2}1$. Hence, we conclude that $A(H_1,1,1,1)$ and $A(H_2,1,1,1)$ have the same anyon decomposition if and only if $(G, H_1, H_2)$ is a Gassmann triple. Explicitly, 
    \be
        A(H_i,1,1,1) \cong \bigoplus_{\rho \in \text{Irrep}(G)}\langle\rho, \text{Ind}^G_{H_i}1\rangle_G ([1],\rho),
    \ee
    for $i=1,2$, where $\langle\cdot, \cdot \rangle_G$ denotes the inner product of characters over $G$. 
    
    $A(H_1,1,1,1)$ and $A(H_2,1,1,1)$ with the same character are non-isomorphic if and only if $(G, H_1, H_2)$ is a non-trivial Gassmann triple.
\end{proof}

    Gassmann triples provide a rich source of examples of twin algebras. A summary of Gassmann triples $(G,H_1,H_2)$ with $|G/H_i|\leq 15$ can be found in \cite{Bosma:2002GassmannNF}.
    The smallest-order example can be found for $\cZ(\Vec_{(\Z_2\times \Z_2) \ltimes \Z_8})$, which we study in detail in \Cref{sec:Gassmann_example}.

If a Gassmann triple $(G, H_1, H_2)$ has additional properties, they give rise to twin Lagrangian algebras $A(H_1,1)$ and $A(H_2,1)$ in $\cZ(\Vec_G)$:
\begin{proposition}[Gassmann-triple twin Lagrangians]
\label{prop:GassmannLagrangian}
    Two Lagrangian algebras $\cL(H_1,1)$ and $\cL(H_2,1)$ in $\cZ(\Vec_G)$ are twins if and only if the non-trivial Gassmann triple $(G, H_1, H_2)$ satisfies that for all abelian subgroup $A$ of $G$ generated by two elements, the number of $A$-fixed cosets
    \begin{equation}
    \label{eqn:fixed-cosets}
        |(G/H_1)^{A}|=|(G/H_2)^{A}|.
    \end{equation}
\end{proposition}
\begin{proof}
    Character formula~\eqref{eqn:character_cond_alg} for $\cL(H_i,1)$ in $\cZ(\Vec_G)$ simplifies to
    \begin{equation}
        \chi_{\cL(H_i,1)}(f,g)=\sum_{\substack{y\in G/H \\ f^{y},g^y\in H}}1
    \end{equation}
    for commuting pair $f,g\in G$,
    where the right hand side computes the number of cosets of $H_i$ fixed simultaneously by $f$ and $g$. That is,
    \begin{equation}
        \chi_{\cL(H_i,1)}(f,g)=|(G/H_i)^{\langle f, g\rangle}|\,.
    \end{equation}
    In particular, for $g=1$, $\chi_{\cL(H_i,1)}(f,1)=\Tr(\Ind^G_{H_i}1(f))$. Requiring $\chi_{\cL(H_1,1)}(f,1)=\chi_{\cL(H_2,1)}(f,1)$ for all $f\in G$ is equivalent to $(G, H_1, H_2)$ forming a Gassmann triple. Further requiring the characters to agree on any commuting pair $f,g\in G$ is equivalent to Gassmann triple $(G,H_1,H_2)$ satisfying
    \begin{equation}
        |(G/H_1)^{A}|=|(G/H_2)^{A}|\,,
    \end{equation}
    for all abelian subgroup $A$ of $G$ generated by two elements. $\cL(H_1,1)$ being non-isomorphic to $\cL(H_2,1)$ is equivalent to $(G, H_1, H_2)$ being a non-trivial Gassmann triple.
\end{proof}
\begin{remark}
    Condition~\eqref{eqn:fixed-cosets} can be formulated in terms of (Burnside) marks, see \cite{Pfeiffer01011997}. In this language, $\chi_{A(H,1)}(f,g)=|(G/H)^{\langle f, g\rangle}|$ is the number of cosets of $H$ fixed simultaneously by $f$ and $g$, which is the \emph{mark} of transitive $G$-set $G/H$ at $\langle f,g \rangle \subset G$. Condition~\eqref{eqn:fixed-cosets} requires $G/H_1$ and $G/H_2$ to have the same mark at any abelian subgroup of $G$ generated by two elements.
    Gassmann-triple twin Lagrangian algebras exist, e.g., for group $G=GL(2,3)$ or $G$ the Clifford group, which we study in detail in the companion paper \cite{Warman:2026gfz}.
\end{remark}

Recall from \eqref{eq:RedTO} that non-maximal Gassmann-triple twin algebras $A(H_1,1,1,1)$ and $A(H_2,1,1,1)$ in $\cZ(\Vec_G^\omega)$ 
are associated to potentially inequivalent reduced TOs $\cZ(\Vec_{H_1}^{\omega|_{H_1}})$ and $\cZ(\Vec_{H_2}^{\omega|_{H_2}})$, respectively. 
\be\label{GLAGR-Gassmann}
\begin{tikzpicture}
\begin{scope}[shift={(0,0)}]
\draw[SymTFTColores, fill= SymTFTColores, opacity = 0.2]   (0,0) -- (0,2) -- (2,2) -- (2,0) --(0,0); 
\draw[SymTFTColores, fill= SymTFTColores, opacity = 0.1]  (2,0) -- (2,2) -- (4,2) -- (4,0) --(2,0); 
\draw [very thick]  (2,0) -- (2,2); 
\node[above] at (2,2) {$\cI_{A(H_i,1,1,1)}$};
\node at (1,1) {$\cZ(G,\omega)$}; 
\node at (3,1) {$\cZ({H_i},{\omega|_{H_i}})$};
\end{scope}
\end{tikzpicture}
\ee
Such examples with inequivalent reduced TOs do exist, we study them in the next section.

\subsubsection{Example with Inequivalent Reduced TOs}
\label{sec:example_Gassmann_Triple_different_reduced_TOs}

Non-trivial Gassmann triples allow us to find examples of twin algebras associated to inequivalent reduced TOs:
\be
    \cZ(\Vec_{H_1}^{\omega|_{H_1}}) \not\cong \cZ(\Vec_{H_2}^{\omega|_{H_2}}) \,.
\ee
This is achieved using non-trivial Gassmann triples $(G,H_1,H_2)$ with $H_1\not\cong H_2$. Hence we conclude that the same collection of anyons can lead to potentially inequivalent\footnote{It is possible to have $\cZ(\Vec_{H_1}^{\omega|_{H_1}}) \cong \cZ(\Vec_{H_2}^{\omega|_{H_2}})$ for non-isomorphic groups $H_1 \not \cong H_2$, see \cite{Naidu2006CategoricalME}.} reduced TOs. 

\begin{example}
    The smallest example of a Gassmann triple $(G, H_1, H_2)$ with $H_1\not\cong H_2$ \cite[Remark 2.9]{Sutherland2021Stronger} has\footnote{The GAP ID of these groups are $G:=\text{SmallGroup}(384, 5755)$, $H_1:=\text{SmallGroup}(16,3)$ and $H_2:=\text{SmallGroup}(16,10)$.} 
\be\ba
G &=((\Z_4^2 \rtimes \Z_4) \rtimes \Z_2) \rtimes \Z_3\,, \cr 
H_1 &= (\Z_4\times \Z_2)\rtimes \Z_2\,,\cr 
H_2&= \Z_4\times\Z_2 \times \Z_2\,.
\ea\ee
The twin algebras $A(H_1,1,1,1)$ and $A(H_2,1,1,1)$ in $\cZ(\Vec_G)$, lead to inequivalent reduced TOs:
\be 
\ba
    \cZ(\Vec_{H_1}) & = \cZ(\Vec_{(\Z_4\times \Z_2)\rtimes \Z_2}) \,, \cr 
    \cZ(\Vec_{H_2}) & = \cZ(\Vec_{\Z_4\times\Z_2 \times \Z_2}) \,.
\ea
\ee
\end{example}

There is furthermore a family of Gassmann triples $(S_{p^3}, \Z_p^3, H_3(\Z_p))$ for any odd prime integer $p$, see \Cref{sec:Sp^3Gassmann}.

These twin algebras are necessarily not related to each other by any braided autoequivalences on $\cZ(\Vec_G)$. This is because two algebras related by braided autoequivalences results in equivalent reduced TOs\footnote{This statement can be checked by concatenation of commutative diagrams.}. This is in contrast to twin algebras studied in the literature constructed from Bogomolov-multiplier autoequivalence \cite{Pollmann:2012,Davydov:2013xov, Cong:2017ffh, Kobayashi:2025ykb,Kobayashi:2025pxs}, which we review in \Cref{sec:SPT-twins}.

\subsection{Twins of $N$-Type}

We can consider algebras, non-maximal ones where $H$ is the same, and $N_i$ are normal subgroups in $H$ with $N_2\neq{}^gN_1$ for all $g\in N_G(H)=\{g\in G: gHg^{-1}=H\}$. Examples will be discussed in \Cref{sec:Ex-Ntype}.

{As a result of Prop.~\ref{prop:HNGassmann}, we know that if two condensable algebras $A(H, N_i,\gamma_i,\epsilon_i)$, $i=1,2$, in $\cZ(\Vec_G)$ are twins, then $(G,N_1,N_2)$ is necessarily a Gassmann triple.

However, not every Gassmann triple $(G, N_1, N_2)$ results in $A(H, N_i,\gamma_i, \epsilon_i)$ being twins for $i=1,2$. For example, consider  $H=G$, $\gamma_i=1$ and $\epsilon_i=1$, then $A(G, N_1,1,1)$ and $A(G, N_2, 1,1)$ have the same anyon decompositions if and only if $N_1=N_2$, i.e., they are isomorphic algebras. For subgroups $H\subsetneq G$, it is possible for $(G, N_1, N_2)$ to be non-trivial Gassmann triples and we will provide examples in \Cref{sec:Ex-Ntype}.}

\subsection{Twins of $\gamma$-Type}
\label{sec:SPT-twins}

We now turn to algebras that have the same subgroups $H$ and $N$, but differ in the cocycle $\gamma$, in the classification data of $A(H,N,\gamma,\eps)$. 
Examples will be discussed in \Cref{sec:Ex-ganmmatype}.

The anyon decomposition is derived from the $G$-action, equivalently by the character. It depends on $\gamma$ in the following way: a 2-cochain $\gamma$ on $N$ governs the algebra multiplication~\eqref{eqn:cond_alg_multiplication}. The $G$-action is required to be compatible with the multiplication, hence algebras with different $\gamma$ generally have different anyon decomposition. However, a $\gamma$ with special properties can result in non-isomorphic algebras with the same anyon decomposition.

Twin {\bf Lagrangian} algebras that we discuss here are closely related to ones previously studied for $\cL(G, \gamma_i)\in\cZ(\Vec_G)$, where the $\gamma$'s in the algebra data differ by a Bogomolov multiplier \cite{Pollmann:2012,Davydov:2013xov, Cong:2017ffh, Kobayashi:2025ykb,Kobayashi:2025pxs}. We generalize the previously studied twins slightly to $\cL(H, \gamma_i)\in \cZ(\Vec_G^\omega)$, allowing $H\neq G$ and non-trivial $\omega$. Moreover we found non-maximal $\gamma$-type twin algebras that are different from the Bogomolov-multiplier type ones in \Cref{sec:Ex-ganmmatype}.

Consider the following special 2-cocycle \cite{FABogomolov_1988, Davydov:2013xov}:
\begin{definition}
    Let $H$ be a finite group, a (representative of) $\beta\in H^2(H,U(1))$ satisfying
    \begin{equation}
    \label{eqn:Bogomolov_multiplier}
        \beta(h_1,h_2)=\beta(h_2,h_1)
    \end{equation}
    for any commuting pair $h_1,h_2\in H$ is called a \emph{Bogomolov multiplier of $H$}.
\end{definition}

\begin{proposition}[Bogomolov-multiplier twin Lagrangian algebras]
\label{prop:Bogomolovtwins}
    Two Lagrangian algebras $\cL(H,\gamma_1)$ and $\cL(H,\gamma_2)$ in $\cZ(\Vec_{G}^{\omega})$ such that 
    \begin{equation}
        \gamma_1=\gamma_2\beta
    \end{equation}
    for some non-trivial Bogomolov multiplier $\beta$ of $H$ are twin algebras.
\end{proposition}
\begin{proof}
    The existence of Lagrangian algebras $A(H,\gamma_i)$, $i=1,2$, in $\cZ(\Vec_G^\omega)$ requires $\omega|_H$ to be trivial, hence we further assume $\omega$ is trivialized on $H$. This guarantees that $\beta\neq \Omega_g|_H$ (recall $\Omega_g$ from equation~\eqref{eq:Omega_f} and Lemma~\ref{lem:Morita_equiv}) for any $g$ in the normalizer of $H$, hence that the two Lagrangian algebras are not isomorphic, by Lemma~\ref{lem:Morita_equiv}.

    The following computation (using formula~\eqref{eqn:character_cond_alg}) shows that these algebras have the same character hence are defined on the same anyons: 
    \begin{align}
        \chi_{A(H,\gamma_1\beta)}(f,g) & = \!\! \sum_{\substack{y\in R,\\ {}^{\bar{y}}f,{}^{\bar{y}}g\in H}} \!\! \frac{\omega({}^{\bar{y}}g, y^{-1} | f)}{\omega(y^{-1}, g|f)} \frac{\gamma_1({}^{\bar{y}}g, {}^{\bar{y}}f)}{\gamma_1({}^{\bar{y}}f,{}^{\bar{y}}g)}\frac{\beta({}^{\bar{y}}g, {}^{\bar{y}}f)}{\beta({}^{\bar{y}}f,{}^{\bar{y}}g)}\nn\\
        & = \!\! \sum_{\substack{y\in R,\\ {}^{\bar{y}}f,{}^{\bar{y}}g\in H}} \!\! \frac{\omega({}^{\bar{y}}g, y^{-1} | f)}{\omega(y^{-1}, g|f)} \frac{\gamma_1({}^{\bar{y}}g, {}^{\bar{y}}f)}{\gamma_1({}^{\bar{y}}f,{}^{\bar{y}}g)}\nn\\
        & = \chi_{A(H,\gamma_1)}(f,g)\,, 
    \end{align}
    for $f,g\in G$ such that $fg=gf$. 
\end{proof}

We comment that Bogomolov multipliers are used in \cite{Davydov:2013xov} to construct non-trivial autoequivalences on $\cZ(\Vec_G)$ that map each object back to itself. This is further used in \cite{Kobayashi:2025pxs, Kobayashi:2025ykb} to find twin algebras of the form $\cL(G,1)$ and $\cL(G,\gamma)$ in $\cZ(\Vec_G)$. (Similar structure has been noticed in \cite{Pollmann:2012}.) Such twin algebras are special cases of Proposition~\ref{prop:Bogomolovtwins} for $H=G$ and $\omega=1$.

The physical implication of the different multiplication structures of such twin algebras with $H=G$ in $\cZ(\Vec_G)$ was studied in \cite{Pollmann:2012,Kobayashi:2025pxs, Kobayashi:2025ykb}.
The corresponding twin phases $\cP^{\Vec_G}_{\cL(H,\gamma)}$ and $\cP^{\Vec_G}_{\cL(H,\gamma\beta)}$ can be distinguished from generalized string order parameters \cite{Pollmann:2012}.

\subsection{Twins of $\epsilon$-Type}
\label{sec:twin_alg_epsilon}

Finally, there are twins that differ only by the $\epsilon$ datum. 
These are necessarily non-maximal, since for all Lagrangian algebras, $\epsilon$ is uniquely fixed by $\gamma$. Examples will be discussed in \Cref{sec:Ex-epstype}.

We find the following twin algebras that differ in $\epsilon$:
\begin{proposition}
\label{prop:twins-epsilon}
    For non-isomorphic condensable algebras $A(G, N, \gamma, \epsilon_1)$ and $A(G, N, \gamma, \epsilon_2)$ in $\cZ(\Vec_G^\omega)$, i.e. when $\epsilon_{1}$ and $\epsilon_{2}$ do not satisfy \eqref{eqn:MoritaEq}, if 
    \be
    \label{eqn:twin_alg_epsilon_cond}
        \epsilon_1(g,n)=\epsilon_2(g,n)
    \ee
    for all $n\in N$ and $g\in C_G(n)$, then they are twins.
\end{proposition}
\begin{proof}
    Equation~\eqref{eqn:scalar_prod_anyon_decom} implies the multiplicity $n_{([x], \rho)}$ of the simple anyon $([x],\rho)$ in the algebra $A(G, N, \gamma, \epsilon_i)$, $i=1,2$, is $0$ for $[x]\not\subset N$, and for $[x]\subset N$, the multiplicity $n_{([x], \rho)}$ evaluates to\footnote{For $H=G$, we fix the set of representative of $G/H$ to be $R=\{1\}$.}
    \be
        \sum_{\substack{f\in [x],\\ g\in C_G(f)}} \biggl\{ \frac{\ol{\Tr(\rho_{f}(g))}}{|G|} \epsilon_i(g,f)\biggr\}.
    \ee
    Note that because $\epsilon_1(g,n)=\epsilon_2(g,n)$ for any $n\in N$ and $g\in C_G(n)$, the above number is the same for $i=1,2$. This shows that the two algebras have the same anyon decomposition.

    Alternatively, one can observe directly from the character formula~\eqref{eqn:character_cond_alg} that \eqref{eqn:twin_alg_epsilon_cond} guarantees that $A(G, N, \gamma, \epsilon_1)$ and $A(G, N, \gamma, \epsilon_2)$ have the same character, hence the same anyon decomposition.
\end{proof}

The smallest rank example for such twins to exist is in TO $\cZ(\Vec_{(\Z_2\times \Z_2) \ltimes \Z_8})$, the same group $(\Z_2\times \Z_2) \ltimes \Z_8$ that we will discuss in detail in \Cref{sec:G32-43_examples} which is also the smallest group admitting a non-trivial Gassmann triple.

These twins can be detected using generalized string order parameters similar to the ones studied in \cite{Haegeman:2012pex, Pollmann:2012}. 
To this end, we recall the physical interpretation of the $H$-action phase $\epsilon$ studied in \cite{Wen:2023otf}, which measures the charge of the localized edge action. More explicitly, an $N$-symmetry action $U_n$ splits into edge actions $U_n=U_n^L \otimes U_n^R$. Introduce notation $\wt{U}_n:=U_n^L$. The edge $\wt{U}_n$ of $U_n$ transforms under $U_h$, for $h\in H$, as
\be
    U_h^{\phantom{-1}}\!\!\!\! \wt{U}_n^{\phantom{-1}}\!\!\!\! U_h^{-1} = \epsilon(h, n) \wt{U}_{{}^hn}.
\ee
Note that if $h$ and $n$ commute and \eqref{eqn:twin_alg_epsilon_cond} holds then, 
\be \label{eq:Ucomm}
    U_{h}^{\phantom{-1}}\!\!\!\! \wt{U}_{n}^{\phantom{-1}}\!\!\!\! U_{h}^{-1} \wt{U}_{n}^{-1} = \epsilon_1(h, n) = \epsilon_2(h, n),
\ee
hence if $\eps_i$'s satisfy equation~\eqref{eqn:twin_alg_epsilon_cond}, the commutator on the left hand side of \eqref{eq:Ucomm} does not detect the difference between the above twin algebras. 

As commented in \cite{Wen:2023otf}, the $\gamma$ and $\epsilon$ data can be used as invariants for a gSPT phase, however the charges defined by $\epsilon$ on commuting group elements do not determine the phase completely. Our new examples presented in the current section provide gSPTs of this type for a group of order 32.

To distinguish two gSPT twin phases $\cP^{\Vec_G}_{A_i}$, $i=1,2$, where $A_1=A(H, N, \gamma, \epsilon_1)$ and $A_2=A(H, N, \gamma, \epsilon_2)$ are twin algebras satisfying conditions in Proposition~\ref{prop:twins-epsilon}, one needs generalized string order parameters of the form 
\be
\label{eqn:generalized_string_order_par}
    U_{h_1}^{\phantom{-1}}\!\!\wt{U}_{n_1}^{\phantom{-1}}\!\!U_{h_1}^{-1}\wt{U}_{n_1}^{-1}\cdots U_{h_k}^{\phantom{-1}}\!\!\wt{U}_{n_k}^{\phantom{-1}}\!\!U_{h_k}^{-1}\wt{U}_{n_k}^{-1} = \prod_{i}^{k}\epsilon(h_i, n_i)\id\,.
\ee
for $h_i\in H$, $n_i\in N$, $i\in \{1, \cdots, k\}$, $k\in \mathbb{Z}_{\geq1}$, such that $h_1^{\phantom{-1}}\!\!\!\! n_1^{\phantom{-1}}\!\!\!\! h_1^{-1} n_1^{-1}\cdots h_k^{\phantom{-1}}\!\!\!\! n_k^{\phantom{-1}}\!\!\!\! h_k^{-1} n_k^{-1}=1$, guaranteeing that the product~\eqref{eqn:generalized_string_order_par} is proportional to the identity operator, with gauge-invariant coefficient, hence can be used to detect the different twin algebras. We illustrate this with examples in \Cref{sec:twinexample_epsilon}.

\section{Twin Phases and Physical Implications}
\label{sec:TwinPhys}

There are several general implications of twin algebras when plugged into the SymTFT paradigm. We will first discuss these in generality, before considering concrete examples. 

\subsection{Twin Phases}
\label{sec:twin-phase}

\begin{definition}[Twin gapped phases.]
Let $\cL_i$, $i=1,2$ be twin Lagrangian algebras for any SymTFT $\cZ (\cS)$ for a fusion category $\cS$. The interval compactification of the SymTFT with $\cL^\phys = \cL_i$ and a fixed symmetry boundary $\cL_\cS^\sym$ give rise to gapped phases $\cP_{\cL_i}^{\cS}$ with symmetry $\cS$, which are called \emph{twin gapped phases}. 
\end{definition}
The schematic picture is as follows:
\be\begin{split}
\begin{tikzpicture}
\begin{scope}[shift={(0,0)}]
\draw [SymTFTColores, fill= SymTFTColores, opacity = 0.2] 
(0,0) -- (0,2) -- (2,2) -- (2,0) -- (0,0) ; 
\draw [white] (0,0) -- (0,2) -- (2,2) -- (2,0) -- (0,0)  ; 
\draw [very thick] (0,0) -- (0,2) ;
\draw [very thick] (2,0) -- (2,2) ;
\draw [ thick] (0,0.8) -- (2,0.8) ;
\fill[] (0,0.8) circle (0.05cm);
\fill[] (2,0.8) circle (0.05cm);
\node[above] at (1,0.8) {$a$};
\node[above] at  (1,1.3)  {$\cZ(\cS)$} ;
\node[above] at (0,2) {$\cL^\sym_{\cS}$}; 
\node[above] at (2.2,2) {$\cL^\phys_{i=1,2}$}; 
\end{scope}
\begin{scope}[shift={(4,0)}]
\node at (-1,1) {$=$} ;
\draw [very thick] (0,0) -- (0,2) ;
\fill[] (0,0.8) circle (0.05cm);
\node[right]at  (0,0.8) {$\cO_{a}^{(i)}$};
\node[above] at (0.1,2) {$\cP_{\cL_{i=1,2}}^{\cS}$}; 
\end{scope}
\end{tikzpicture}\end{split}
\ee
In particular, the inequivalence of the Lagrangians implies that the resulting phases are inequivalent: One way to argue this is that the algebra of order parameters is different. We will discuss this in the next subsection. 
An alternative way to phrase this is, that the associated module categories of the symmetry $\cS$ are inequivalent -- this uses the identification of gapped phases with module categories. 

We can argue this for group-theoretical symmetries, i.e. $\cS$ such that $\cZ (\cS) =\cZ(\Vec_G^\omega)$. We focus on two non-conjugate subgroups $H_1$ and $H_2$ that are preserved in the gapped phase $\cP^{\Vec_G^{\omega}}_{\cL(H_i,1)}$, $i=1,2$ (this will be the situation for several twin phases and thus instructive to recall). 
For each phase, we can take the ground state basis of local operators to be\footnote{Here, for each object $a$ in $\cZ(\Vec_G^\omega)$ and for $g\in G$, we identify the vector space $\Hom_{\Vec_G^\omega}(F_{A(1,1)}(a),\bC_g)$ with the $g$-graded vector space component of $a$.} $\{v_{r,1}\}_{r\in R}$, following notation from \eqref{eqn:basiselm}, where $R$ is a common set of representatives of left cosets $G/H_i$ for $i=1,2$ (whose existence is guaranteed, see Remark~\ref{rmk:Hallsmarriage}).
From the multiplication structure \eqref{eqn:cond_alg_multiplication} the $v_{r,1}$ are idempotents 
\be
    v_{r,1}v_{r',1}=\delta_{rH_i,r'H_i}\;v_{r,1}\,,
\ee
and thus physically correspond to projectors onto orthogonal vacua (also called universes in the gapless setting).
On these idempotents, it is explicit to check that subgroups conjugate to $H_i$ act trivially on these bases: this follows from the $G$-action \eqref{eqn:G-action_algebra_basis} and relation \eqref{eq:v_equiv_rel} 
\be\ba
    (rh_ir^{-1})\cdot v^i_{r,1}=v^i_{rh_i,1}=v^i_{r,1}\,,
\ea\ee
for all $h_i\in H_i$. This shows that in the vacuum defined by $v^i_{r,1}$ the symmetry $rH_ir^{-1}$ is preserved. Conversely, for $k_i\notin H_i$, one has
\be\ba
    (rk_ir^{-1})\cdot v^i_{r,1}=v^i_{rk_i,1}\neq v^i_{r,1}\,,
\ea\ee
since $rk_i\notin rH_i$. This shows that group elements that are not in $rH_ir^{-1}$ act non-trivially on $v_{r,1}^i$: the symmetries they generate are therefore spontaneously broken therein. In particular the phases associated to Lagrangian algebras $\cL (H_1, 1)$ and $\cL (H_2, 1)$ with $H_1$ and $H_2$ non-conjugate in $G$ are not related by a continuous path of symmetric Hamiltonians without closing the gap, i.e. are inequivalent gapped phases. We can argue similarly for projective representations, i.e. in the presence of $\gamma$. For more general symmetries that may not be group-theoretical, one can argue that the algebra of order parameters is different, which we will discuss in the next subsection. 

\vspace{2mm}
\noindent{\bf Gapless Twin Phases.}
We can furthermore study interfaces and associated gapless phases, that are twins. We will discuss these in the examples more concretely. 
In particular we find that $\epsilon$-type twin algebras, which necessarily are non-Lagrangian algebras, give rise to twin gapless SPT (gSPT) phases. In order to distinguish such gapless phases, we need to determine a generalized string order parameter, which does not rely on the anyons, but is detected from the action of the symmetry on the order parameters. 

\subsection{The Algebra of Order Parameters}

One way to distinguish twin phases is to consider the operator product algebra of order parameters (OPs), see e.g., \cite{Chavda:2026mpz}. In fact the definition of twins as algebras that are isomorphic as objects, i.e. arising from the same generalized charges, but being not isomorphic as algebras in $\cZ(\cS)$ (multiplicative structure), indicates this already. 
Here we will first discuss the algebra structure on condensable algebras in some more detail before applying these to twins. It turns out that in order to study the algebra of OPs, it is in fact entirely unnecessary to consider anyons.

Algebra structure of condensable algebras in $\cZ(\Vec_G^\omega)$ is determined by viewing them as $G$-graded vector spaces \eqref{eqn:cond_alg_multiplication}. This translates to multiplication on the vector space of local operators in a $\cS$-symmetric phase $\cP^{\cS}_{A}$, for any group-theoretical fusion category $\cS$.

The space of local operators in $\cP^{\cS}_{\cL_\phys}$ can be recovered from the algebras $\cL_\sym$ and $\cL_\phys$ in the following way. The symmetry Lagrangian algebra for $\cS$, denoted $\cL_\cS=\cL_\sym$, in $\cZ(\cS)$ singles out a monoidal forgetful functor\footnote{This can be obtained from the free module functor $-\otimes \cL_{\cS}$ \cite{davydov2013witt}.} 
\begin{equation}
    F_{\cL_\cS}: \cZ(\cS) \rightarrow \cS,
\end{equation}
which projects the anyon lines in the TO to the symmetry boundary $\Bsym$. For each symmetry defect $X$ in $\cS$, the space of $X$-twisted local operators in $\cP_{\cL_\phys}^{\cS}$ is
\be
    (V^{\cS}_{\cL_\phys})_{X}:=\Hom_{\cS}\bigl(F_{\cL_\cS}(\cL_\phys), X\bigr).
\ee

For a simple anyon $a$ in $\cL_\phys$, the $\cS$ symmetry action on an $X$-twisted operator $\cO\in \Hom_{\cS}(F_{A_\cS}(a),X)\subset (V^{\cS}_{\cL_\phys})_{X}$ depends only on\footnote{See, e.g., \cite{Bhardwaj:2023idu, Bhardwaj:2023ayw}. This statement is further made explicit in terms of representations of tube categories in \cite{Bartsch:2026wqq}.} $a$. 
Hence the space of $X$-twisted operators can be further decomposed with respect to the $\cS$ symmetry action\footnote{Equivalently, decompose with respect to generalized charge/tube representation \cite{Bartsch:2022mpm, Bartsch:2023wvv}.} into
\begin{equation}
    (V^{\cS}_{\cL_\phys})_{X} \cong \bigoplus_{a\in \Irr(\cZ(\cS))} (V^{\cS}_{\cL_\phys})_{a, X}
\end{equation}
where $(V^{\cS}_{\cL_\phys})_{a, X}:=n_{\cL_\phys,a}\Hom_{\cS}(F_{\cL_{\cS}}(a),X)$.
In particular, the space of local operators (i.e., local operators in untwisted sector) in $\cP_{\cL_\phys}^{\cS}$ is
\begin{equation}
    \begin{split}
        (V^{\cS}_{\cL_\phys})_{1} & =\Hom_{\cS}\bigl(F_{\cL_\cS}(\cL_\phys), 1\bigr)\\
        & \cong \bigoplus_{a\in \Irr(\cZ(\cS))}(V^\cS_{\cL_\phys})_{a,1}\,.
    \end{split}
\end{equation}

Multiplication on the space of local operators can be recovered from the coalgebra structure on $\cL_\phys=A$. Let 
$F^2_{\cL_\sym}$ denote the monoidal structure of the forgetful functor $F_{\cL_\sym}$. For $X,Y$ in $\cC$, two local operators
\begin{align}
    \cO_{1} & \in \Hom_{\cC}(F_{\cL_{\sym}}(A), X),\\
    \cO_{2} & \in \Hom_{\cC}(F_{\cL_{\sym}}(A), Y),
\end{align}
in the $X$- and $Y$-twisted sectors, respectively, 
can be multiplied to a local operator
\begin{equation}
    (\cO_{1}\otimes\cO_{2}) \circ (F^2_{\cL_\sym})_{A, A}^{-1}\circ F(\Delta),
\end{equation}
in the $X\otimes Y$-twisted sector $\Hom_{\cC}(F_{\cL_{\sym}}(A), X\otimes Y)$, where $\Delta:A\rightarrow A\otimes A$ is the comultiplication.

For $\cS=\Vec_G^\omega$, the symmetry algebra is $\cL(1,1)$ and the forgetful functor $F_{\cL(1,1)}:\cZ(\Vec_G^\omega)\rightarrow \Vec_G^\omega$ forgets the group action. In this case we can identify the vector space of $g$-twisted operator $\Hom_{\cS}(F_{\cL(1,1)}(A),g)$ with the $g$-graded component of $A$, and the multiplication on the vector space of $g$-twisted operators is given by \eqref{eqn:cond_alg_multiplication}.

\subsection{No Relative SSB and non-Landau Transitions}
\label{sec:no-rel-SSB}

One crucial implication of twin gapped phases is that irrespective of the symmetry boundary the symmetry breaking patterns for the phases $\cP_i$ are the same. In particular, twin phases cannot be distinguished by relative spontaneous symmetry breaking (SSBing). 
Here we define {\bf no relative SSBing} as the property of two gapped phases, which for any choice of gauge related (Morita equivalent) symmetry category, have the same symmetry breaking pattern. In particular they both have the same number of vacua.   

Naturally, phases with no relative SSBing are perfect candidates for {\bf phase transitions with no hidden symmetry breaking}\footnote{Hidden symmetry breaking occurs when two topological $\cC$-symmetric phases are indistinguishable by local operators, but their gauged phases are distinguishable by symmetry breaking patterns. E.g., this is the case between the $\Z_2\times \Z_2$-symmetric trivial and non-trivial SPT phases \cite{Kennedy:1992ifl, kennedy1992hidden, Oshikawa:1992smv}, and similarly for $\Z_n\times \Z_n$ \cite{PhysRevB.88.085114, PhysRevB.88.125115}.}, although in order for there to be a second order transition, we require some additional structure (e.g. an anomaly). This is discussed in the companion paper \cite{Warman:2026gfz}. 

\begin{proposition}[No relative SSBing]
\label{prop:no-hidden-SB}
    Let $\cL_1$ and $\cL_2$ be Lagrangian algebras in $\cZ(\cC)$. $\cP^{\cC}_{\cL_1}$ and $\cP^{\cC}_{\cL_2}$ have no relative SSBing
    if and only if $\cL_1 \cong \cL_2$ as objects in $\cZ(\cC)$, i.e. they are twin Lagrangian algebras.
    
    In particular, for any symmetry $\cS$ that is Morita equivalent to $\cC$ the phases $\cP^{\cS}_{\cL_1}$ and $\cP^{\cS}_{\cL_2}$ have the same symmetry breaking pattern. 
\end{proposition}
\begin{proof}
    Assuming $\cL_1\not\cong \cL_2$ as objects in $\cZ(\cC)$, then there exists $a\in \cL_1$ such that 
    \be
    n_{\cL_1,a}\neq n_{\cL_2,a}.
    \ee
    Let $\cC_1$ denote the fusion category corresponding to $\cL_1$, then the two subspaces 
    \be
        (V_{\cL_1}^{\cC_1})_{a,1} \not\cong (V_{\cL_2}^{\cC_1})_{a,1}
    \ee
    are inequivalent, showing that $\cP^{\cC_1}_{\cL_1}$ and $\cP^{\cC_1}_{\cL_2}$ admit distinct symmetry breaking patterns. We conclude either $\cP^{\cC}_{\cL_1}$ and $\cP^{\cC}_{\cL_2}$ have distinct symmetry breaking pattern, or if they have the same symmetry breaking pattern, they can be gauged to $\cP^{\cC_1}_{\cL_1}$ and $\cP^{\cC_1}_{\cL_2}$ with different symmetry breaking pattern, hence admitting relative SSBing. 
    For any $\cS$ symmetry dual to $\cC$ and $a\in \Irr(\cZ(\cC))$, $\cP^{\cS}_{\cL_1}$ and $\cP^{\cS}_{\cL_2}$ have isomorphic spaces $(V_{\cL_1}^{{\cS}})_{a,1}\cong (V_{\cL_2}^{{\cS}})_{a,1}$ of local operators, with respect to the $\cS$ symmetry. Therefore there is no relative SSBing between $\cP^{\cC_1}_{\cL_1}$ and $\cP^{\cC_1}_{\cL_2}$ as they admit the same symmetry breaking pattern for all symmetry categories related to $\cC$ by gauging. 
\end{proof}

There is a bijection between Lagrangian algebras in $\cZ(\cC)$ and indecomposable $\cC$-module categories \cite{davydov2013witt}. Denote by  $\cM_{\cL}^{\cC}$ the $\cC$-module category corresponding to Lagrangian algebra $\cL$ in $\cZ(\cC)$. 
Let $\cC$ be a fusion category. The topological $\cC$-symmetric phases $\cP^{\cC}_{\cL}$ are in bijection to $\cC$-module categories $\cM_{\cL}^{\cC}$ \cite{Thorngren:2019iar}. Phase $\cP^{\cC}_{\cL}$ has vector space of vacua of dimension $\rank \cM_{\cL}^{\cC}$.
Hence the twin topological phases defined above have the following properties:

    If $\cL_1$ and $\cL_2$ are twin Lagrangian algebras in $\cZ(\cC)$, then
    \begin{equation}
        \rank \cM_{\cL_1}^{\cS} =\rank \cM_{\cL_2}^{\cS}
    \end{equation}
    for any fusion category $\cS$ Morita dual\footnote{See Appendix~\ref{sec:induced_functor_on_Morita_dual} for the definition of Morita duals.} to $\cC$.

\section{Examples of Twin Algebras} 
\label{sec:DoubTroub}

We now give examples of twin algebras of all the different types we studied in general in section \ref{sec:examples-of-twin-algs}. We focus on twins for groups of order $< 64$ as those will not fall into the category of known twins arising from Bogomolov multipliers. We will consider group with and without anomalies. The search was done using GAP.
A complete list of twins in $\cZ (\Vec_{G})$ for $|G|\leq 48$ is provided in \Cref{app:TwinLists}.

\subsection{Examples of $H$-type Twins}
\label{sec:Ex-Htype}

As we explained these $H$-type twins arise from non-conjugate subgroups. For groups without anomalies, by Proposition \ref{prop:Gassmann}, they come from Gassmann triples. 

\subsubsection{$G=(\Z_2 \times \Z_2) \ltimes \Z_8$ with and without $\omega$}
A systematic scan with GAP reveals, that the smallest order of a group $G$ for which twin algebras occur in $\cZ(\Vec_G^\omega)$ is $32$.
The SmallGroup(32,43) is given by 
\be\label{G32Def}
    G:=(\Z_2 \times \Z_2) \ltimes \Z_8 = \Z_8^*\ltimes \Z_8  \,.
\ee
We will discuss some of the properties of the twin phases for this group in \Cref{sec:G32-43_examples}. We list all the twins for this symmetry (and any Morita equivalent symmetry) in \Cref{tab:twinsG32-43} without anomaly and \Cref{tab:TwinG3243omega} with $\omega$.

\widetext{\onecolumngrid 
\begin{center}
\begin{table}
\begin{minipage}{\textwidth}
$$
\begin{array}{|c|c|c|}
\hline 
\makecell{\text{Twin Algebras in } \cZ(\Vec_{G_{48,29}}) \\ G_{48,29} = GL(2,3) = \langle h, c, a, (xz), (iz) \rangle } &  \makecell{\text{Reduced} \\ \text{TO}}  & \makecell{\text{Algebra} \\ \text{Inequivalence}}   \\
\hline \hline 
A(S_3^{(1)}=\langle a, h \rangle, 1, 1, 1)  &
\cZ(\Vec_{S_3})  & \multirow{3}{*}{$ H'\neq{}^g\!H$} \\
A(S_3^{(2)}=\langle a,hc\rangle, 1, 1, 1)  &
\cZ(\Vec_{S_3})  & \\
([1],1)\oplus([1],\rho_3)\oplus([1],\rho_4)  &  &  \\
\hline

A(D_{12}=\langle a,h,hc \rangle, S_3^{(1)}=\langle a, h \rangle, 1, 1)  &
\cZ(\Vec_{\Z_2})  & \multirow{3}{*}{$ N'\neq{}^g\!N$} \\
A(D_{12}=\langle a,h,hc \rangle, S_3^{(2)}=\langle a,hc\rangle, 1, 1)  &
\cZ(\Vec_{\Z_2})  & \\
([1],1)\oplus([1],\rho_3)\oplus([h],1_{++})\oplus([a],1_{++})  &  &  \\
\hline

A(D_{12}=\langle a,h,hc \rangle, S_3^{(1)}=\langle a, h \rangle, 1, \eps_1)  &
\cZ(\Vec_{\Z_2})  & \multirow{3}{*}{$ N'\neq{}^g\!N$} \\
A(D_{12}=\langle a,h,hc \rangle, S_3^{(2)}=\langle a,hc\rangle, 1, \eps_2)  &
\cZ(\Vec_{\Z_2})  & \\
([1],1)\oplus([1],\rho_3)\oplus([h],1_{+-})\oplus([a],1_{++})  &  &  \\
\hline

A(S_3^{(1)}=\langle a, h \rangle, \Z_3=\langle a\rangle, 1, 1)  &
\cZ(\Vec_{\Z_2})  & \multirow{3}{*}{$ H'\neq{}^g\!H$} \\
A(S_3^{(2)}=\langle a,hc\rangle, \Z_3=\langle a\rangle, 1, 1)  &
\cZ(\Vec_{\Z_2})  & \\
([1],1)\oplus([1],\rho_3)\oplus([1],\rho_4)\oplus([a],1_{++})\oplus([a],1_{+-})  &  &  \\
\hline

A(S_3^{(1)}=\langle a, h \rangle, S_3^{(1)}=\langle a, h \rangle, 1, 1)  & \text{Trivial}    & \multirow{3}{*}{$ H'\neq{}^g\!H$} \\
A(S_3^{(2)}=\langle a,hc\rangle, S_3^{(2)}=\langle a,hc\rangle, 1, 1)  & \text{Trivial}    & \\
([1],1)\oplus([1],\rho_3)\oplus([1],\rho_4)\oplus([h],1_{++})\oplus([h],1_{+-})\oplus([a],1_{++})\oplus([a],1_{+-})  &  &  \\
\hline
\end{array}
$$
\caption{The complete list of twin algebras in $\cZ(\Vec_{GL(2,3)})$. $(G, S_3^{(1)}, S_3^{(2)})$ is a non-trivial Gassmann triple. For the epsilon data, $\eps_1(\cdot,h)$ maps $h\mapsto1,c\mapsto-1$  while $\eps_2(\cdot,hc)$ maps $hc\mapsto1,c\mapsto-1$. For the character table of $GL(2,3)$ and anyon notation, see \cite{Warman:2026gfz}.
\label{tab:twinsGL32}}
\end{minipage}
\end{table}
\end{center}}

{\onecolumngrid
\begin{center}
\begin{table}
\begin{minipage}{\textwidth}
$$
\begin{array}{|c|c|c|c|}
\hline 
\makecell{\text{Twin Algebras in } \cZ(\Vec_{G_{32,43}}) \\ G_{32,43}=(\Z_2 \times \Z_2) \ltimes \Z_8 } &  \makecell{\text{Reduced} \\ \text{TO}}  & \makecell{\text{Algebra} \\ \text{Inequivalence}}  & 
\text{Section}\\
\hline \hline 
\makecell{A(G, \langle g_1\rangle, 1, 1) =: A_1^{\langle g_1\rangle}\\ 
A(G, \langle g_1\rangle, 1, \epsilon^{\langle g_1\rangle}_2) =: A_2^{\langle g_1\rangle} \\
([1],1)\oplus ([g_1^4],1) \oplus ([g_1^2], 1) \oplus ([g_1],1)
} & \makecell{\cZ(\Vec_{\Z_2\times \Z_2}) \\ \cZ(\Vec_{\Z_2\times \Z_2}) \\ { }} & \makecell{\eps_1^{g}\neq\eps_2 } & \makecell{\text{\Cref{sec:twinexample_epsilon}}}\\
\hline
\makecell{
A(G, \langle g_1 g_3 \rangle, 1, 1) =: A_1^{\langle g_1 g_3 \rangle} \\
A(G, \langle g_1 g_3 \rangle, 1, \epsilon^{\langle g_1 g_3 \rangle}_2) =: A_2^{\langle g_1 g_3 \rangle} \\
([1],1)\oplus ([g_1^4],1) \oplus ([g_1^2], 1) \oplus ([g_1g_3],1)
} & \makecell{\cZ(\Vec_{\Z_2\times \Z_2}) \\ \cZ(\Vec_{\Z_2\times \Z_2}) \\ { }} & \makecell{\eps_1^{g}\neq\eps_2} & \makecell{\text{\Cref{sec:twinexample_epsilon}}} \\
\hline
\makecell{
{A(H_3, \langle g_3 \rangle, 1, \epsilon)} \\
{A(H_3, \langle g_3 g_1^4 \rangle, 1, \epsilon')} \\
{([1],1)\oplus ([1],r_1) \oplus ([1],r_8) \oplus ([g_3],2)}
} & \makecell{{\cZ(\Vec_{\Z_2\times \Z_2})} \\ {\cZ(\Vec_{\Z_2\times \Z_2})} \\ { }} & \makecell{{N'\neq{}^g\!N} \\ {g\in N_G(H_3)}} & \makecell{{\text{\Cref{sec:Ex-Ntype}}}} \\
\hline
\makecell{ 
A(H_1, 1, 1, 1)  =: $\nameref{alg:4}$ \\
A(H_2, 1, 1, 1) =: $\nameref{alg:5}$\\
([1],1)\oplus([1],r_1)\oplus([1],r_8)\oplus([1],r_{10})}
 & \makecell{\cZ(\Vec_{\Z_2\times \Z_2})\\\cZ(\Vec_{\Z_2\times \Z_2})\\{ }}& \makecell{\text{Gassmann Triple:\vphantom{$I^{I}$}} \\ H_1\neq{}^g\!H_2 \\ { }}  &
  \makecell{{\text{\Cref{sec:Gassmann_example}}}} \\

\hline
\end{array}
$$
\caption{Table of the complete list of twin algebras in $\cZ(\Vec_{(\Z_2 \times \Z_2) \ltimes \Z_8 })$. For the first two pairs of twin algebras, $g_1, g_3\in G$ satisfy equation~\eqref{eqn:generatorsG32-43}, and $\epsilon^{\langle g_1\rangle}_2, \, \epsilon^{\langle g_1 g_3\rangle}_2$ are defined in equations \eqref{eqn:G32,43esp_2^g1} and \eqref{eqn:G32,43esp_2^g1g3}, respectively. For the third pair, $H_3=\langle g_2,g_3,g_1^4\rangle\cong\Z_2^3$ and $\epsilon,\epsilon'$ are defined in \eqref{eqn:G32Ntype_eps}. The subgroups $\langle g_3\rangle$ and $\langle g_3g_1^4\rangle={}^{g_1}\langle g_3\rangle$ are conjugate in $G$, but not by an element of the normalizer $N_G(H_3)$ so the algebras are inequivalent.
For the last pair of twin algebras, $(G, H_1, H_2)$ is a non-trivial Gassmann triple, where the two subgroups $H_1\cong H_2\cong \Z_2\times \Z_2$ are defined in (\ref{eqn:HinG32-43}). Here, $r_i$, $i\in \{1,8,10\}$ denotes irreducible $(\Z_2 \times \Z_2) \ltimes \Z_8$-representations whose characters can be found in Table~\ref{tab:G3243chars}. 
\label{tab:twinsG32-43}}
\end{minipage}
\end{table}
\end{center}

{\onecolumngrid
\begin{center}
\begin{table}
\begin{minipage}{\textwidth}
$$
\begin{array}{|c|c|c|}
\hline 
\makecell{\text{Twin Algebras in } \cZ(\Vec_{G_{32,43}}^{\omega}) \\ G_{32,43}=(\Z_2 \times \Z_2) \ltimes \Z_8 } &  \makecell{\text{Reduced} \\ \text{TO}}  & \makecell{\text{Algebra} \\ \text{Inequivalence}}   \\
\hline \hline 
A(H_1, 1, 1, 1)=\text{\nameref{Aom:4}}  &
\cZ(\Vec_{\Z_2 \times \Z_2})  & \text{Gassmann Triple:} \\
A(H_2, 1, 1, 1)=\text{\nameref{Aom:5}}  &
\cZ(\Vec_{\Z_2 \times \Z_2})  &  H_1\neq{}^g\!H_2 \\
([1],1) \oplus ([1],r_1) \oplus ([1],r_8) \oplus ([1],r_{10})  &  &  \\
\hline

A(H_1, \langle g_2g_3\rangle, 1, 1)=\text{\nameref{Aom:6}}   &
\cZ(\Vec_{\Z_2})  &  \text{Gassmann Triple:}  \\
A(H_2, \langle g_2g_3\rangle, 1, 1)=\text{\nameref{Aom:11}}  &
\cZ(\Vec_{\Z_2})  & H_1\neq{}^g\!H_2 \\
([1],1) \oplus ([1],r_1) \oplus ([1],r_8) \oplus ([1],r_{10}) \oplus ([g_2g_3],2)  &  &  \\
\hline

A(H_1, \langle g_2\rangle, 1, 1)=\text{\nameref{Aom:8}}  &
\cZ(\Vec_{\Z_2})  & \text{Gassmann Triple:}  \\
A(H_2, \langle g_2g_1^4\rangle, 1, 1)=\text{\nameref{Aom:9}}  &
\cZ(\Vec_{\Z_2})  & H_1\neq{}^g\!H_2  \\
([1],1) \oplus ([1],r_1) \oplus ([1],r_8) \oplus ([1],r_{10}) \oplus ([g_2],2)  &  &  \\
\hline

A(H_1, \langle g_3\rangle, 1, 1)=\text{\nameref{Aom:7}}  &
\cZ(\Vec_{\Z_2})  & \text{Gassmann Triple:}  \\
A(H_2, \langle g_3g_1^4\rangle, 1, 1)=\text{\nameref{Aom:10}}  &
\cZ(\Vec_{\Z_2})  & H_1\neq{}^g\!H_2  \\
([1],1) \oplus ([1],r_1) \oplus ([1],r_8) \oplus ([1],r_{10})&  &  \\
 \oplus ([g_3],2\;g_1^6=-\sqrt{2}\zeta_{16}) \oplus ([g_3],2\;g_1^6=\sqrt{2}\zeta_{16})  &  &  \\
\hline

A(H_1, H_1, 1, 1)=\text{\nameref{Aom:15}}  
& \text{Trivial}  & \text{Gassmann Triple:} \\
A(H_2, H_2, 1, 1)=\text{\nameref{Aom:16}} 
& \text{Trivial}  & H_1\neq{}^g\!H_2 \\
([1],1) \oplus ([1],r_1) \oplus ([1],r_8) \oplus ([1],r_{10}) \oplus ([g_2g_3],2) \oplus ([g_2],2)& & \\ \oplus ([g_3],2\;g_1^6=-\sqrt{2}\zeta_{16}) \oplus ([g_3],2\;g_1^6=\sqrt{2}\zeta_{16})  &  &  \\
\hline
\end{array}
$$
\caption{Table of twin algebras in $\cZ(\Vec^{\omega}_{(\Z_2 \times \Z_2) \ltimes \Z_8 })$. The two subgroups $H_1\cong H_2\cong \Z_2\times \Z_2$ defined in (\ref{eqn:HinG32-43}) are such that $(G, H_1, H_2)$ form a non-trivial Gassmann triple. The Hasse diagram for these algebras is in Figure~\ref{fig:Hasse3243_om}. \label{tab:TwinG3243omega}}
\end{minipage}
\end{table}
\end{center}
}
\twocolumngrid}

\twocolumngrid

\subsubsection{$G=GL(2,3)$ with and without $\omega$}
This group is the main focus of the companion paper 
\cite{Warman:2026gfz}. It is the general linear group on two by two matrices 
\be 
G=GL(2, 3)= Q_8 \rtimes S_3 
\ee 
over the field  $\mathbb{F}_3$. The key property is that it has two non-conjugate subgroups, that are isomorphic as groups (permutation groups), but non-conjugate in $G$ 
\be
S_3^{(1)} \not \sim S_3^{(2)} \,.
\ee
We find that there are Gassmann triple algebras for both $G$ and $G$ with anomaly. The details are provided in Tab.~\ref{tab:twinsGL32} for $G$ without anomaly. With anomaly, the condensable algebras are fewer, and the twins are:
\be
\begin{array}{|c|c|c|} \hline
\text{Twin 1} & \text{Twin 2} & \text{Reduced TO} \cr \hline\hline 
A(S_3^{(1)},1,1,1) & A(S_3^{(2)},1,1,1) & D({S_3}) \cr \hline
A(S_3^{(1)},\Z_3,1,1) & A(S_3^{(2)},\Z_3,1,1) & D(\Z_2) \cr \hline 
A(S_3^{(1)},S_3^{(1)},1,1) &A(S_3^{(2)},S_3^{(2)},1,1) & \text{Trivial} \cr \hline
\end{array}
\ee
We study the physical implications of the twin gapped phases in \cite{Warman:2026gfz}, where in the presence of a mixed-anomaly, we note that the transition between the twin gapped phases is beyond-Landau without hidden symmetry breaking.

\subsubsection{$G=S_{p^3}$} 
\label{sec:Sp^3Gassmann}

    There is a family of non-trivial Gassmann triples $(G,H_1,H_2)$, indexed by an odd prime, with $H_1$ abelian and $H_2$ non-abelian \cite{Komatsu:1976adelering}. For each odd prime number $p$, $G=S_{p^3}$ denotes the symmetric group. Consider 
    \be
    H_1\cong(\Z_{p})^3\,,\qquad 
    H_2\cong \bm{H}_3(\Z_{p}) \,,
    \ee
    where $\bm{H}_3$ is the Heisenberg group, i.e., the group of $3\times 3$ upper triangular matrices over $\Z_{p}$, with diagonal entries being 1. Regarding both groups as subgroups of $S_{p^3}$, where both groups act on themselves as a set of cardinality $p^3$, via left multiplication, $(S_{p^3}, (\Z_{p})^3, \bm{H}_3(\Z_{p}))$ is a non-trivial Gassmann triple. Hence the condensable algebras $A((\Z_{p})^3,1,1,1)$ and $A(\bm{H}_3(\Z_{p}),1,1,1)$ in $\cZ(\Vec_{S_{p^3}})$ are twins, but their condensations result in inequivalent reduced TOs, as discussed in \Cref{sec:example_Gassmann_Triple_different_reduced_TOs}:
    \be 
    \ba
        \cZ(\Vec_{H_1}) & = \cZ(\Vec_{(\Z_{p})^3}) \cr 
        \cZ(\Vec_{H_2}) & = \cZ(\Vec_{\bm{H}_3(\Z_{p})}) \,.
    \ea 
    \ee

\subsection{Examples of $N$-type Twins}
\label{sec:Ex-Ntype}

The group $GL(2,3)$ with trivial anomaly $\omega=1$ also provides examples of $N$-type twins: 
\be
A(D_{12}, S_3^{(1)}, 1, 1) \,,\qquad 
A(D_{12}, S_3^{(2)}, 1, 1)
\ee
are twins that have the same $H$, but non-conjugate $N$. 

They have the anyon decomposition
\begin{equation}
    ([1],1)\oplus([1],\rho_3)\oplus([h],1_{++})\oplus([a],1_{++})
\end{equation}

Another group admitting twins of $N$-type is 
\be
    G_{48,49} = \Z_2^2 \times A_4 = \langle f_1, f_2, f_3, f_4, f_5 \rangle
\ee
for the algebras
\be\ba
    A(\Z_2^4=\langle f_1,f_2,f_4,f_5\rangle, \Z_2^2=\langle f_1 f_4,f_2 f_5\rangle, 1, 1)\,,\\
    A(\Z_2^4=\langle f_1,f_2,f_4,f_5\rangle, \Z_2^2=\langle f_4 f_1 f_5,f_2 f_5\rangle, 1, 1)
\ea\ee
with anyon decomposition
\be\ba
    A& \cong([1],1)\oplus ([1],1 f_3=e^{4\pi i/3})\oplus([1],1 f_3=e^{2\pi i/3})\\
    & \quad \oplus([f_1 f_4],1)\oplus([f_2 f_4],1)\oplus([f_1 f_2 f_4],1) \,.
\ea\ee

The group $G_{32,43}$ of \eqref{G32Def}, with the generators \eqref{eqn:generatorsG32-43}, admits an $N$-type pair in $\cZ(\Vec_{G_{32,43}})$ whose normal subgroups are conjugate in $G$. Let
\be
\ba
H_3&=\langle g_2,g_3,g_1^4\rangle\cong\Z_2^3\,, \cr
N&=\langle g_3\rangle\,,\qquad N'=\langle g_3g_1^4\rangle={}^{g_1}N\,.
\ea
\ee
Choose characters $\chi,\chi'$ of $H_3$ with 
\be\ba
\chi(g_3)&=1\,,\quad &\chi(g_1^4)&=-1\,,\\
\chi'(g_3g_1^4)&=1\,,\quad &\chi'(g_1^4)&=-1\,,
\ea\ee
and set
\be \label{eqn:G32Ntype_eps}
\epsilon(h,g_3)=\chi(h)\,,\quad \epsilon'(h,g_3g_1^4)=\chi'(h)\,,\quad h\in H_3\,.
\ee
Then $A(H_3,N,1,\epsilon)$ and $A(H_3,N',1,\epsilon')$ are twins with the anyon decomposition
\be\ba \label{eqn:G32Ntype_anyons}
A&\cong([1],1)\oplus([1],r_1)\oplus([1],r_8)\oplus([g_3],2)\,,
\ea\ee
and both have reduced TO $\cZ(\Vec_{\Z_2\times\Z_2})$. Here $([g_3],2)$ denotes the $2$-dimensional irrep of $C_G(g_3)$ in which $g_3$ is represented trivially.

This twin pair does not survive in $\cZ(\Vec_{G_{32,43}}^\omega)$ for the 3-cocycle $\omega$ of \texttt{G32-43-omega.g}, whose restriction to $H_3$ is the type-III cocycle $\omega_{\text{III}}$, see \Cref{sec:G32-43_examples}. Indeed, since $H_3$ is abelian, \eqref{eqn:epsilon_relation_transgression} requires $\omega(\cdot,\cdot|n)$ to be a coboundary on $H_3$ for every $n\in N$. For $\omega_{\text{III}}$ and $n\neq1$, the commutator $\omega(h_1,h_2|n)/\omega(h_2,h_1|n)=(-1)^{\det(n,h_1,h_2)}$, with $n,h_1,h_2$ viewed as additive vectors in $\Z_2^3$ (in the $a,b,c$ basis of \Cref{sec:G32-43_examples}), is nontrivial. Hence no $\epsilon$ exists, and the only condensable algebra on $H_3$ in $\cZ(\Vec_{G_{32,43}}^\omega)$ is $A(H_3,1,1,1)$, consistent with Table~\ref{tab:TwinG3243omega}.

For other groups, there are also $N$-type pairs whose normal subgroups are conjugate in $G$ but not by $N_G(H)$: two in $G_{48,32}=\Z_2\times SL(2,3)$, with $H\cong Q_8$ and $N\cong N'\cong\Z_4$, and two in $G_{48,33}$, with $H\cong D_8$ and $N\cong N'\cong\Z_2\times\Z_2$, as shown in Table~\ref{tab:twins_sweep_no_omega}.

\subsection{Examples of $\gamma$-type Twins}
\label{sec:Ex-ganmmatype}

The algebras that differ by the SPT $\gamma$ and $H=G$ inside $\cZ(\Vec_G)$ were studied previously in \cite{Pollmann:2012,Kobayashi:2025pxs, Kobayashi:2025ykb}.

Performing a computer search using GAP \cite{GAP4,HAP,mignard2017moritaequivalencepointedfusion,gruen2021computing} we find the following example of twin algebras that are not maximal, and differ in $\gamma$ not by a Bogomolov multiplier. 
The example exists for 
\be
G_{32,6} = (\Z_2 \times  \Z_2 \times \Z_2) \rtimes \Z_4 \,.
\ee
with a non-trivial anomaly $\omega$ provided in the data file \texttt{G32-6-omega.g}. 
The subgroups 
\be
\ba
H&= \Z_2 \times  \Z_2 \times \Z_2=\langle f_2,f_3,f_5\rangle\cr 
N&= \Z_2 \times \Z_2=\langle f_2f_3,f_3f_5\rangle \,,
\ea
\ee
where $f_i$ are the various $\Z_2$ generators,
define two inequivalent non-maximal algebras $A(H, N, 1, \eps)$ and $A(H, N, \gamma', \eps')$.
Here, $\gamma'$ is the standard non-trivial 2-cocycle on $\Z_2\times\Z_2$ and where 
\be
 \gamma^{\prime g} \Omega_g \neq 1 \,,
\ee
which implies that the algebras are non-isomorphic, where $\Omega_g$ is defined in equation~\eqref{eq:Omega_f}. Both reduced TO of either algebra is equivalent to $\cZ (\Vec_{\Z_2})$. 
These algebras are twins, decompose into simple anyons as: 
\be
\ba
&([1],1) \oplus ([1],f_1=i) \oplus ([1], f_1=-1) \cr
& \oplus ([1],f_1=-i) \oplus ([f_2],f_5=-1) \cr 
&\oplus ([f_2],f_3=-1 ,f_5=-1) \oplus ([f_3],2)  \,,
\ea
\ee
where we use the standard $([g],R)$ notation for anyons in group theoretical TOs, denoting the representation with value of its character on elements.

\subsection{Examples of $\epsilon$-type Twins}
\label{sec:Ex-epstype}

There are twins that only differ in $\epsilon$. Examples arise again in the group (\ref{G32Def}), and are shown in \Cref{tab:twinsG32-43}. Details of these algebras were studied in \Cref{sec:twin_alg_epsilon}.

\subsection{All Twins in $\cZ(\Vec_G)$ up to $|G|\leq48$}
\label{app:TwinLists}

Performing a systematic search with GAP, we find that twin algebras occur in $\cZ(\Vec_G)$ for the groups listed in Tab.~\ref{tab:twins_sweep_no_omega}. This table is reproduced by the sweep at the end of \texttt{examples/7\_twin\_algebras\_G32.g} of \texttt{AlgebrasTQD} \cite{AlgebrasTQD}.

\begin{table}[H]
    \centering
    \begin{tabular}{|l|l|}
    \hline
    \text{Group $G$ with twin algebras in $\cZ(\Vec_G)$} & \text{\;\;\;Type of twins} \\
    \hline
    $G_{32,43}=\Z_8 \rtimes  (\Z_2 \times \Z_2)$ & \makecell[l]{Non-Lag\quad$H'\neq{}^g\!H$ \\ {Non-Lag\quad$N'\neq{}^g\!N$} \\ Non-Lag\quad$\eps^{\prime g}\neq\eps$} \\
    \hline 
    $G_{32,44}=(\Z_2 \times Q_8) \rtimes  \Z_2$ & Non-Lag\quad$\eps^{\prime g}\neq\eps$ \\
    \hline
    $G_{48,29}=GL(2,3)$ & \makecell[l]{Non-Lag\quad$H'\neq{}^g\!H$ \\ Non-Lag\quad$N'\neq{}^g\!N$ \\\phantom{Non-}Lag\quad$H'\neq{}^g\!H$} \\
    \hline
    $G_{48,31}=\Z_4 \times A_4$ &  Non-Lag\quad$\eps^{\prime g}\neq\eps$ \\
    \hline
    $G_{48,32}=\Z_2 \times SL(2,3)$ & \makecell[l]{{Non-Lag\quad$N'\neq{}^g\!N$} \\ Non-Lag\quad$\eps^{\prime g}\neq\eps$} \\
    \hline
    $G_{48,33}=((\Z_4 \times \Z_2) \rtimes  \Z_2) \rtimes  \Z_3$ & \makecell[l]{{Non-Lag\quad$N'\neq{}^g\!N$} \\ Non-Lag\quad$\eps^{\prime g}\neq\eps$} \\
    \hline
    $G_{48,49}=\Z_2 \times \Z_2 \times A_4$ & \makecell[l]{Non-Lag\quad$H'\neq{}^g\!H$ \\ Non-Lag\quad$N'\neq{}^g\!N$ \\ Non-Lag\quad$\eps^{\prime g}\neq\eps$} \\
    \hline
    \end{tabular}
    \caption{Groups of order $|G|\leq 48$ for which twin algebras occur in $\cZ(\Vec_G)$ and their type. Here we use the notation $N'\neq{}^g\!N$ for all $g\in N_G(H)$ but $N,N'$ may be conjugate in $G$.}
    \label{tab:twins_sweep_no_omega}
\end{table}

\section{Twin Algebras and Phases for $G_{32}=(\Z_2 \times \Z_2) \ltimes \Z_8$}
\label{sec:G32-43_examples}

\subsection{Properties of $\cZ(G_{32})$}

Let us return now to the smallest group $G$ for which $\cZ(\Vec_G^\omega)$ has twins. It has order $32$ and is given by\footnote{This group has GAP ID $(32,43)$ and can be defined in GAP as $G:=\text{SmallGroup}(32, 43)$.} 
\be
    G_{32}:=(\Z_2 \times \Z_2) \ltimes \Z_8 = \Z_8^*\ltimes \Z_8  \,.
\ee
Here, $\Z_8^*=\{1,3,5,7\}\cong \Z_2\times \Z_2$ is the group of units of the ring $\Z_8$, in terms of which the group multiplication is defined by $(a,b)(c,d)=(ac,bc+d)$, for $a,c\in \Z_8^*$ and $b,d\in \Z_8$. We fix generators of this group to be 
\be
\label{eqn:generatorsG32-43}
    g_1=(1,1), \quad g_2 = (3,0), \quad g_3 = (5,0).
\ee
In terms of these generators,
\be
G_{32}=(\Z_2^{(g_2)} \times \Z_2^{(g_3)}) \ltimes \Z_8^{(g_1)}\,,\\
\ee
with conjugation action (denoting ${}^{g}\!f = gfg^{-1}$)
\be
{}^{g_2}\!g_1=(g_1)^3\,,\quad {}^{g_3}\!g_1=(g_1)^5\,.
\ee
The group $G_{32}$ has character table shown in Table~\ref{tab:G3243chars}.

The TO $\cZ(\Vec_{G_{32}})$ admits four pairs of twin algebras (all non-maximal) as shown in Table~\ref{tab:twinsG32-43}. 
Each algebra defines a map of anyons, which can be found in Appendix~\ref{app:all-maps-of-anyons}. These include examples with non-conjugate $(H, N)$ data, as well as ones that differ in their $H$-action phase $\epsilon$ data and can be detected using generalized string order parameters.

We also study twins in $\cZ(\Vec_{G_{32}}^\omega)$ with 
\be \omega \in H^3(G_{32}, U(1)) = \Z_2^4 \times \Z_8 \,.
\ee
The twins are summarized in table \ref{tab:TwinG3243omega}. 
Their phases are particularly interesting as they provide non-Landau DQCP type transitions -- thanks to the anomaly.  
  Here, $\omega$ is a representative of a non-trivial 3-cocycle computed with GAP \cite{GAP4,HAP,mignard2017moritaequivalencepointedfusion}, which we provide in \texttt{G32-43-omega.g}. We choose it to have the following property
  \be
\omega_{\Z_2^3} = \omega_{\III} \,,\qquad 
  \ee
  where $\omega_{\III}$ is the type III cocycle on $\Z_2^3$, 
and $\omega$ is order $8$ on the $\Z_8$ factor.

\subsection{Gapless SPT Twins and Generalized String Order Parameters}
\label{sec:twinexample_epsilon}

We start with $\cZ(\Vec_{G_{32}})$ with no anomaly and discuss $\epsilon$-type twin algebras as discussed in Proposition~\ref{prop:twins-epsilon}. The algebras are shown in \Cref{tab:twinsG32-43}. They are non-Lagrangian algebras, and thus define interfaces. For concreteness we will consider the group symmetry $G_{32}$, but any other Morita equivalent symmetry is equally admissible and will give rise to twin phases. 
The symmetry Lagrangian $G_{32}$ is 
\be
\cL_{G_{32}} = \cL(1,1) = \bigoplus_{\rho = G_{32}\text{-irreps}} \dim(\rho) ([1],\rho) \,.
\ee
And we consider the following interfaces: 
\be
\label{fig:G3243noomega}
		\begin{split}
			\begin{tikzpicture}
				\begin{scope}[shift={(0,0)}, scale = 0.8] 
					\pgfmathsetmacro{\del}{-3}
					\draw [SymTFTColores, fill= SymTFTColores, opacity = 0.2]
					(4,0) -- (0,0)--(0,4)--(4,4)--(4,0);
					\draw [SymTFTColores, fill= SymTFTColores, opacity = 0.1]
					(8,0) -- (4,0)--(4,4)--(8,4)--(8,0);
                    	\draw[very thick,black] (0,0) -- (0,4);
                        \draw[very thick,black] (4,4) -- (4,0);
                        \draw[very thick,black] (8,4) -- (8,0);
					\node[above] at (4, 4) {$\cI_{A_i^{\langle g_1\rangle}}$};
					\node[above]  at (0, 4) {$\cL_{G_{32}} $};
                    	\node[above]  at (8, 4) {$\Bphys=\Ising^2$};
					\node  at (2, 2) {${\cZ({(\Z_2 \times \Z_2) \ltimes \Z_8})}$};
					\node  at (6, 2) {${\cZ({\Z_2\times \Z_2})}$};
				\end{scope}
			\end{tikzpicture}    
		\end{split}
	\ee
On the physical boundary, we input a $\Z_2\times \Z_2$ symmetric CFT such as $\Ising^2$. The resulting 
 $G_{32}$-symmetric phases corresponding to the twin algebras $A_i^{\langle g_1 \rangle}$ and $A_i^{\langle g_1g_3 \rangle}$ for $i=1,2$ have a unique (up to scalar) untwisted sector local operator: they are thus {\bf gapless SPT} (gSPT) phases.

As  gSPT phases they cannot be distinguished by local operators. Furthermore, since twin algebras have the same anyon decomposition, twin phases cannot be distinguished by ordinary string order parameters either. 
However, the gapless twin phases can be detected by generalized string order parameters~\eqref{eqn:generalized_string_order_par}, which we discuss now.

\vspace{2mm}
\noindent\textbf{Example 1: Gapless SPT Twins.} We first focus on the twin algebras
\begin{align}
    A_1^{\langle g_1\rangle} & :=A(G_{32},\langle g_1\rangle,1,1), \\
    A_2^{\langle g_1\rangle} & :=A(G_{32},\langle g_1\rangle,1,\epsilon^{\langle g_1\rangle}_2),
\end{align}
where $\langle g_1 \rangle \cong \Z_8$, and $\epsilon^{\langle g_1\rangle}_2$ is such that
\be
\label{eqn:G32,43esp_2^g1}
    \epsilon^{\langle g_1\rangle}_2(g_2g_1^{n}, g_1^{2m+1})=\epsilon^{\langle g_1\rangle}_2(g_3g_1^{n}, g_1^{2m+1}) = -1,
\ee
for $n\in \{0,1,\cdots, 7\}$, $m\in \{0,1,2,3\}$\footnote{The conjugacy class $[g_1]=\{g_1^{2m+1}\}_{m=0}^{3}$ in $G_{32}$.}, and takes value of 1 on other elements in $G_{32}\times \langle g_1\rangle$. For both $i=1,2$, $\epsilon_i^{\langle g_1\rangle}$ evaluates to 1 on commuting group elements (i.e., they satisfy condition~\eqref{eqn:twin_alg_epsilon_cond}), hence the above two algebras $A_i^{\langle g_1\rangle}$ are twin algebras. As in Table~\ref{tab:twinsG32-43}, they have anyon decomposition
\be
    A_i^{\langle g_1\rangle} \cong 1 \oplus ([g_1^4],1) \oplus ([g_1^2], 1) \oplus ([g_1],1).
\ee
Hence, the $G_{32}$-symmetric twin gapless phases $\cP^{\Vec_{G_{32}}}_{A_i^{\langle g_1 \rangle}}$, $i=1,2$, cannot be distinguished by local or ordinary string order parameters. However, they can be distinguished by generalized string order parameters as discussed around equation~\eqref{eqn:generalized_string_order_par}. 
Concretely, the generalized string order parameter as defined in (\ref{eqn:generalized_string_order_par}) is 
\be
      (U_{g_3}^{\phantom{-1}}\!\!\wt{U}_{g_1}^{\phantom{-1}}\!\!U_{g_3}^{-1}\wt{U}_{g_1}^{-1})\,
     (U_{g_2^{\phantom{2}}}^{\phantom{-1}}\!\!\wt{U}_{g_1^2}^{\phantom{-1}}\!U_{g_2^{\phantom{2}}}^{-1}\wt{U}_{g_1^2}^{-1}) 
\ee
is a phase 
\be
    \epsilon_i^{\langle g_1\rangle}(g_3, g_1)\epsilon_i^{\langle g_1\rangle}(g_2, g_1^2) = (-1)^{i-1}\,.
\ee
It can distinguish the gapless SPT twin phases $\cP^{\Vec_{G_{32}}}_{A_i^{\langle g_1\rangle}}$, $i=1,2$.

\vspace{2mm}
\noindent\textbf{Example 2: Gapless SPT Twins.} Since these are very peculiar phases, let us consider another example for gapless SPT twins: consider the non-maximal twins
\begin{align}
    A_1^{\langle g_1 g_3 \rangle} & := A(G_{32}, \langle g_1 g_3 \rangle, 1, 1), \\ 
    A_2^{\langle g_1 g_3 \rangle} & := A(G_{32}, \langle g_1 g_3 \rangle, 1, \epsilon^{\langle g_1 g_3 \rangle}_2),
\end{align}
where  
$\langle g_1 g_3 \rangle \cong \Z_8$, and $\epsilon^{\langle g_1 g_3 \rangle}_2$ is defined as follows:
\be
\label{eqn:G32,43esp_2^g1g3}
    \epsilon_2^{\langle g_1 g_3 \rangle}(g,n)=\begin{cases}
        -1, & g\in S,\, n\in [g_1 g_3], \\
        1, & \text{otherwise},
    \end{cases}
\ee
where $S=\{ g_1^{2n+1}, g_2 g_1^{2n+1}, g_3 g_1^{2n}, g_3 g_2 g_1^{2n} \}_{n=0}^{3}$. Similarly to the previous example, both $\epsilon_1^{\langle g_1 g_3 \rangle}$ and $\epsilon_2^{\langle g_1 g_3 \rangle}$ evaluate to 1 on commuting pairs of elements, satisfying condition~\eqref{eqn:twin_alg_epsilon_cond}. For $i=1,2$, $A_i^{\langle g_1 g_3 \rangle}$ in TO $\cZ(\Vec_{(\Z_2\times \Z_2) \ltimes \Z_8})$ are twins, defined on the anyons
\be
    A_i^{\langle g_1 g_3 \rangle} \cong 1 \oplus ([g_1^4],1) \oplus ([g_1^2], 1) \oplus ([g_1g_3],1)\, .
\ee
In this case, the generalized string order parameter
\be
      (U_{g_3}^{\phantom{-1}}\!\wt{U}_{g_1g_3}^{\phantom{-1}}U_{g_3}^{-1}\wt{U}_{g_1g_3}^{-1})\,
     (U_{g_2^{\phantom{2}}}^{\phantom{-1}}\!\!\wt{U}_{g_1^2}^{\phantom{-1}}\!U_{g_2^{\phantom{2}}}^{-1}\wt{U}_{g_1^2}^{-1}) 
\ee 
is a phase 
\be
    \epsilon_i^{\langle g_1g_3\rangle}(g_3, g_1g_3)\epsilon_i^{\langle g_1g_3\rangle}(g_2, g_1^2) = (-1)^{i-1}\,.
\ee
This distinguishes the gapless SPT twin phases $\cP^{\Vec_{G_{32}}}_{A_i^{\langle g_1 g_3\rangle}}$.

\subsection{$H$-type Twin Algebras with Non-Trivial $\omega$} 
\label{sec:Gassmann_example}
We now consider non-maximal Gassmann-triple twin algebras (see Proposition~\ref{prop:Gassmann}) in $\cZ(\Vec_{G_{32}}^\omega)$, and their corresponding $\Vec_{G_{32}}^\omega$-symmetric gapless twin phases, for any $\omega\in H^3(G_{32},U(1))$. We discuss the symmetry action and OPE of local operators in two different bases:
$G_{32}$ admits a non-trivial Gassmann triple $(G_{32}, H_1, H_2)$ with 
\be
\label{eqn:HinG32-43}
    \begin{split}
        H_1 & =\{(1,0),\,g_2=(3,0),\,g_3=(5,0),\,
        g_2 g_3=(7,0)\}\cr 
        &\cong \Z_2 \times \Z_2,\\
    H_2 & =\{(1,0),\,
    g_2 g_1^4= (3,4),\, g_3 g_1^4=(5,4), \,g_2g_3=(7,0)\}\cr 
    &\cong \Z_2 \times \Z_2.
    \end{split}
\ee
Following Proposition~\ref{prop:Gassmann}, $A(H_i,1,1,1)$ for $i=1,2$ in $\cZ(\Vec_{G_{32}}^\omega)$ are twins for any 3-cocycle $\omega\in H^3(G_{32},U(1))$. They are defined on anyons 
\be \label{eq:AH1H2_anyons}
    A(H_i,1,1,1) \cong 1 \oplus ([1],r_1) \oplus ([1],r_8) \oplus ([1],r_{10})\,.
\ee
Here $r_n$ denotes irrep of $G_{32}$ whose character can be found in Table~\ref{tab:G3243chars}.
For the following discussion, we fix the representatives for both $G_{32}/H_i$, $i=1,2$, to be
\be \label{eq:R_Gassman}
    R=\{g_1^n \;|\; n\in\Z_8\}\,.
\ee
This example illustrates that twin phases, can have two equivalent descriptions: 
\begin{itemize}
\item Idempotent basis: here the symmetry action is different, but the OPE of OPs is the same. 
\item Local operator basis: Here the symmetry action is the same, but we have different OPEs of OPs. 
\end{itemize}
In the following we discuss both points of view in turn. 

\subsubsection{Basis of Idempotents: Different $G_{32}$-action but Same OPEs}
\label{sec:basis-diff-action-same-OPE}
\begin{table}
$$
{\scriptsize\begin{array}{|c|ccccccccccc|}
    \hline
     & \![1]\! & \!\!\![g_1 g_2 g_3]\!\! & \![g_2 g_3]\! & \![g_3] & \![g_1^2] & \![g_1^4] & [g_1] & \![g_1 g_2]\! & \![g_2] & \![g_3 g_1^2]\! & \![g_1 g_3]\! \cr
    \hline
   1 &  1 & 1 & 1 & 1 & 1 & 1 & 1 & 1 & 1 & 1 & 1 \cr
 r_1 & 1 & -1 & 1 & 1 & 1 & 1 & -1 & -1 & 1 & 1 & -1  \cr
r_2& 1 & 1 & -1 & 1 & 1 & 1 & -1 & 1 & -1 & 1 & -1  \cr
r_3& 1 & -1 & -1 & 1 & 1 & 1 & 1 & -1 & -1 & 1 & 1  \cr
r_4& 1 & 1 & 1 & -1 & 1 & 1 & 1 & -1 & -1 & -1 & -1  \cr
r_5&  1 & -1 & 1 & -1 & 1 & 1 & -1 & 1 & -1 & -1 & 1  \cr
r_6& 1 & 1 & -1 & -1 & 1 & 1 & -1 & -1 & 1 & -1 & 1  \cr
r_7& 1 & -1 & -1 & -1 & 1 & 1 & 1 & 1 & 1 & -1 & -1  \cr
r_8& 2 & 0 & 0 & 2 & -2 & 2 & 0 & 0 & 0 & -2 & 0  \cr
r_9& 2 & 0 & 0 & -2 & -2 & 2 & 0 & 0 & 0 & 2 & 0  \cr
r_{10} & 4 & 0 & 0 & 0 & 0 & -4 & 0 & 0 & 0 & 0 & 0  \cr
\hline
\end{array}
}
$$
\caption{Character table for the group $G_{32}=(\Z_2\times \Z_2)\ltimes \Z_8$. $r_i$ label the irreps, and $g_i$ are the generators in \eqref{eqn:generatorsG32-43}. \label{tab:G3243chars}
}
\end{table}

Following \Cref{sec:CondAlg_class}, a basis for $A(H_i,1,1,1)$, $i=1,2$, in $\cZ(\Vec_{G_{32}}^\omega)$, is
\be
\label{eqn:G3243GassmannBasis}
	A(H_i,1,1,1) =\text{span}\{v^{i}_{g_1^k,1}\, | \, k\in \Z_8\}.
\ee
We introduce short-hand notation 
\be
\label{notation:vik}
    v^i_k:=v^i_{g_1^k,1} 
\ee
with $i=1,2$ labeling the algebra and $k=0,\cdots,7$ labeling the basis elements.
Physically, in the $\Vec_{G_{32}}^\omega$-symmetric gapless phase obtained by condensing $A(H_i,1,1,1)$ on $\Bphys$, each $v^i_k$ is a universe\footnote{We use the term ``universe'' instead of ``vacuum'' since the phase is gapless.} labeled by the left coset $g_1^kH_i$.

The OPEs~\eqref{eqn:cond_alg_multiplication} read:
\be
v^i_k\, v^i_m=\delta_{k,m}\,v^i_{k}\,,
\ee
for $i=1,2$ and $k,m\in\Z_8$. The $v^i_k$'s are therefore idempotents, in agreement with their physical interpretation as projectors onto orthogonal vacua/universes.

The $G_{32}$-action can be computed from \eqref{eqn:G-action_algebra_basis} and \eqref{eq:v_equiv_rel}: 
\begin{itemize}
    \item $g_1$ permutes the idempotents:
    \be \label{eq:g1_act_H1H2}
    g_1\cdot v^i_{k} = v^i_{k+1}
    \ee
    (where $k+1$ is taken mod 8)
    \item $g_2g_3 \in H_1\cap H_2$ so in both phases each universe preserves a symmetry conjugate to $g_2g_3$:
    \be
     g_1^k\,(g_2g_3)\,g_1^{-k}\cdot v^i_{k} = v^i_{k}
    \ee
    \item $g_2,g_3\in H_1$ so each universe in phase 1 preserves a symmetry conjugate to $g_2$ and $g_3$
    \be\ba
    g_1^k\,(g_2)\,g_1^{-k}\cdot v^1_{k} &= v^1_{k} \\
    g_1^k\,(g_3)\,g_1^{-k}\cdot v^1_{k} &= v^1_{k}
    \ea\ee
     \item $g_2g_1^4,g_3g_1^4\notin H_1$, so elements conjugate to them act non-trivially in phase 1
         \be\ba
     g_1^k\,(g_2g_1^4)\,g_1^{-k}\cdot v^1_{k} &= v^1_{k+4}\\
     g_1^k\,(g_3g_1^4)\,g_1^{-k}\cdot v^1_{k} &= v^1_{k+4}
     \ea\ee
    \item $g_2g_1^4,g_3g_1^4\in H_2$ so each universe in phase 2 preserves a symmetry conjugate to $g_2g_1^4$ and $g_3g_1^4$
    \be\ba
    g_1^k\,(g_2g_1^4)\,g_1^{-k}\cdot v^2_{k} &= v^2_{k} \\
    g_1^k\,(g_3g_1^4)\,g_1^{-k}\cdot v^2_{k} &= v^2_{k}
    \ea\ee
    \item $g_2,g_3\notin H_2$, so elements conjugate to them act non-trivially in phase 2
        \be\ba
    g_1^k\,(g_2)\,g_1^{-k}\cdot v^2_{k} &= v^2_{k+4} \\
    g_1^k\,(g_3)\,g_1^{-k}\cdot v^2_{k} &= v^2_{k+4}\,.
    \ea\ee
\end{itemize}
Since $g_1$ permutes the idempotents, both phases are $\Z_8^{(g_1)}$-SSB phases. However, with the basis choice \eqref{eqn:G3243GassmannBasis}, the corresponding phases $\cP^{\Vec_{G_{32}}^\omega}_{A_i}$ exhibit distinct and inequivalent symmetry-breaking patterns. The preserved symmetry in $\cP^{\Vec_{G_{32}}^\omega}_{A_1}$ and $\cP^{\Vec_{G_{32}}^\omega}_{A_2}$ are subgroups conjugate to $H_1$ and $H_2$ respectively, which, we recall, are not conjugate to each other.

However, it is always possible to find alternative basis for the space of local operators in twin phases on which the $\Vec_{G_{32}}^\omega$ actions take identical form. In the 
next section we compute this alternative basis and show how the different algebra structures manifest in different OPEs.

\subsubsection{Basis of Local Operators: Same $G_{32}$-action but Different OPEs} \label{sec:OPEs_Gassman}
Here we provide an alternative basis of local operators on which the symmetry action takes identical form. This basis also makes the anyon decomposition~\eqref{eq:AH1H2_anyons} of $A(H_i,1,1,1)$ more transparent. From the $G_{32}$-action presented above, 
consider the following change of basis
\be
\label{eqn:anyon_basis_Gassmann_example}
	w^{i}_{m}:=\sum_{n=0}^{7} \zeta_{8}^{nm}v^{i}_n,
\ee
where $i=1,2$, $m\in \Z_8$ and $\zeta_{8}=\exp(2\pi i/8)$. With this basis, $A(H_i,1,1,1)$ decomposes into simple anyons:
\be
    A \cong 1 \oplus ([1],r_1) \oplus ([1],r_8) \oplus ([1],r_{10}),
\ee
where as $\bC$-vector spaces,
\begin{align} \label{eq:rw}
    1 & = \text{span}\{w^{i}_{0}\}, \nn\\
    ([1],r_1) & = \text{span}\{w^{i}_{4}\},\\
    ([1],r_8) & = \text{span}\{w^{i}_{2},w^{i}_{6}\},\nn\\
    ([1],r_{10}) & =\text{span}\{w^{i}_1,(-1)^{i-1}w^{i}_3,(-1)^{i-1}w^{i}_5,w^{i}_7\}.\nn
\end{align}

Note that here we made an alternative choice of basis elements such that
the symmetry $G_{32}$ action on them takes the same form for both algebras. 
We spell the action out of the generators $g_1,g_2,g_3$ \eqref{eqn:generatorsG32-43} of $G_{32}$ as follows.
\begin{itemize}
    \item $G_{32}$ acts on $([1],r_1) = \text{span}\{w^{i}_{4}\}$ as the irrep $r_1: G_{32} \rightarrow U(1)\subset \text{GL}(\bC)$
\be
\label{eqn:1_dim_irrep}
    r_1(g_1)= -1, \quad r_1(g_2)=r_1(g_3)=1.
\ee
\item $G_{32}$ acts on $([1],r_8) = \text{span}\{w^{i}_{2},w^{i}_{6}\}$ as 2-dim irrep $r_8: G_{32} \rightarrow U(2)\subset \text{GL}(\bC^2)$
\be
\label{eqn:2_dim_irrep}
        \begin{split}
    	r_8(g_1) & = \begin{pmatrix}
    		\zeta_{8}^{6} & 0 \\
    		0 & \zeta_{8}^{2}
    	\end{pmatrix}, \\
        r_8(g_2) & = \begin{pmatrix}
    	0 & 1 \\
    	1 & 0
    	\end{pmatrix}, \\
        r_8(g_3) & = \begin{pmatrix}
    	    1 & 0\\
            0 & 1
    	\end{pmatrix}.
    \end{split}
\ee
\item Finally, $G_{32}$ acts on $([1],r_{10}) =\text{span}\{w^{i}_1,(-1)^{i-1}w^{i}_3,(-1)^{i-1}w^{i}_5,w^{i}_7\}$ as the 4-dim irrep $r_{10}: G_{32} \rightarrow \text{SU}(4)\subset \text{GL}(\bC^4)$ reads
\be
    \begin{split}
        \label{eqn:4_dim_irrep}
	r_{10}(g_1) & = \begin{pmatrix}
		\zeta_{8}^{7} & 0 & 0 & 0 \\
		0 & \zeta_{8}^{5} & 0 & 0 \\
		0 & 0 & \zeta_{8}^{3} & 0 \\
		0 & 0 & 0 & \zeta_{8}
	\end{pmatrix}, \\ 
    r_{10}(g_2) & = \begin{pmatrix}
		0 & 1 & 0 & 0 \\
		1 & 0 & 0 & 0 \\
		0 & 0 & 0 & 1 \\
		0 & 0 & 1 & 0
	\end{pmatrix}, \\
    r_{10}(g_3) & = \begin{pmatrix}
	0 & 0 & 1 & 0 \\
	0 & 0 & 0 & 1 \\
	1 & 0 & 0 & 0 \\
	0 & 1 & 0 & 0 
	\end{pmatrix}.
    \end{split}
\ee
\end{itemize}
The fusion rules of the anyons in the algebras $A(H_i,1,1,1)$ are:
\be\ba
    ([1], r_1) \ot ([1], r_1) & \cong 1\\
    ([1], r_1) \ot ([1], r_8) & \cong ([1], r_8) \\
    ([1], r_1) \ot ([1], r_{10}) & \cong ([1], r_{10}) \\
    ([1],r_8) \ot ([1],r_8) & \cong 1 \oplus ([1],r_1) \oplus \cdots,\\
    ([1], r_8) \ot ([1], r_{10}) & \cong 2 ([1], r_{10}), \\
    ([1],r_{10}) \ot ([1],r_{10}) & \cong 1 \oplus ([1],r_1) \oplus 2 ([1],r_8) \oplus \cdots \,,
\ea\ee
where the dots $\cdots$ denote anyons beyond the ones appearing in the algebras $A(H_i,1,1,1)$.

The multiplications $\mu_i$ on $A(H_i,1,1,1)$ on the level of anyons can be computed from equation~\eqref{eqn:cond_alg_multiplication}, they differ for $i=1,2$ when restricted to the following subspaces: 
\be
\ba
    \mu_{i}|_{([1],r_1) \ot ([1],r_{10})} & = (-1)^{i-1}\id_{([1],r_{10})},\\
    \mu_i|_{([1],r_8) \ot ([1],r_{10})} & = \begin{pmatrix}
        \id_{([1],r_{10})} & (-1)^{i-1} \id_{([1],r_{10})}
    \end{pmatrix},\\
    \mu_{i}|_{([1],r_1) \subset ([1],r_{10}) \ot ([1],r_{10})} & = (-1)^{i-1} \id_{([1],r_1)},\\
    \mu_{i}|_{2([1],r_8) \subset ([1],r_{10}) \ot ([1],r_{10})} & = \begin{pmatrix}
        \id_{([1],r_8)} & (-1)^{i-1} \id_{([1],r_8)}
    \end{pmatrix}.
\ea
\ee

The different multiplication structures manifest in distinct OPEs of local operators. To make it explicit, consider the SymTFT with the Dirichlet symmetry algebra $\cL(1,1)$ for $\Vec_{G_{32}}^\omega$ symmetry.
Algebra $A (H_i, 1, 1, 1)$ defines an interface between TO $\cZ(G_{32})$ and the reduced TO $\cZ(\Vec_{H_i})$ on which the anyons appearing in $A (H_i, 1, 1, 1)$ can end. The local operators arise from anyons ending on both the symmetry boundary and the interface $\cI_{A(H_i,1,1,1)}$: 
\be
\label{fig:G3243condquiche}
		\begin{split}
			\begin{tikzpicture}[scale=0.8]
				\begin{scope}[shift={(0,0)},scale=0.8] 
					\pgfmathsetmacro{\del}{-3}
					\draw [SymTFTColores, fill= SymTFTColores, opacity = 0.2]
					(4,0-1.2) -- (0,0-1.2)--(0,4)--(4,4)--(4,0-1.2);
					\draw [SymTFTColores, fill= SymTFTColores, opacity = 0.1]
					(8,0-1.2) -- (4,0-1.2)--(4,4)--(8,4)--(8,0-1.2);
					\draw[very thick] (0,-1.2) -- (0,4) ;
					\draw[thick] (0,3) -- (4,3);
					\draw[thick] (0,2) -- (4,2);
					\draw[thick] (0,1.2) -- (4,1);
					\draw[thick] (0,0.8) -- (4,1);
					\draw[thick] (0,-0.3) -- (4,0);
					\draw[thick] (0,-0.1) -- (4,0);
					\draw[thick] (0,0.1) -- (4,0);
					\draw[thick] (0,0.3) -- (4,0);
					\draw[thick,black,line width=1pt] (4,4) -- (4,0-1.2);
					\node  at (4, 4.5) {$\cI_{A(H_i,1,1,1)}$};
					\node  at (0, 4.5) {$\Bsym_\Dir $};
					\node  at (2, -1.7) {${\cZ({(\Z_2 \times \Z_2) \ltimes \Z_8})}$};
					\node  at (2, 3.45) {$1$};
					\node  at (2, 2.35) {$([1],r_1)$};
					\node  at (2, 1.45) {$([1],r_8)$};
					\node  at (2, -0.7) {$([1],r_{10})$};
					\node  at (6, -1.7) {${\cZ({\Z_2\times \Z_2})}$};
				\end{scope}
			\end{tikzpicture}    
		\end{split}.
	\ee
In this case, since $A(H_i,1,1,1) \subset A(1,1)$ as subalgebra for $i=1,2$, we can compute the OPE from the algebra multiplication~\eqref{eqn:cond_alg_multiplication}. 
Fix notations for the local operators obtained from condensing $A(H_i,1,1,1)$, $i=1,2$. 
The unit anyon $1=([1],1)$ gives rise to the unit local operator
\be
\label{eqn:localop1}
    \cO^i_{1} = w_0^i.
\ee
The local operator $\cO^i_{r_1}$ associated to the anyon $([1],r_1)$ is chosen to be (unique choice up to rescaling)
\be
\label{eqn:localopr1}
    \cO^i_{r_1} = w_4^i.
\ee
The anyon $([1],r_{8})$ gives rise to 2 linearly-independent local operators $\cO^i_{ r_{8},n}$, for $n\in \{2, 6\}$, defined in terms of the basis elements to be
\be
\label{eqn:localopr8}
    \cO^i_{r_{8},n} = w_n^i.
\ee
The anyon $([1],r_{10})$ gives rise to 4 linearly-independent untwisted local operators $\cO^i_{r_{10},n}$, defined to be 
\be
\label{eqn:localopr10}
    \cO^i_{r_{10},n} = w_n^i, \quad \cO^i_{r_{10},m} = (-1)^{i-1}w_m^i,
\ee
for $n=1,7$ and $m=3,5$. 

As shown before, the $\Vec_{G_{32}}^\omega$ symmetry action on local operators $\cO^i_{r_{n}}$ does not depend on the label $i$: as matrices, the $G_{32}$ action is determined by \eqref{eqn:2_dim_irrep} and \eqref{eqn:4_dim_irrep}. Taking the trace of these matrices recovers the braiding of the anyons $([1],r_8)$ and $([1],r_{10})$ with each group element $g\in G_{32}$, respectively. The action of the generators $g_1=(1,1), \, g_2=(3,0),\, g_3=(5,0)\in \Z_8^* \ltimes \Z_8$ is shown in Table~\ref{tab:G3243action}.  We remark that this choice of $\Vec_{G_{32}}^\omega$ symmetry action on local operators does not depend on the 3-cocycle $\omega$.

The algebra multiplication \eqref{eqn:cond_alg_multiplication} and equivalence relation on basis vectors \eqref{eq:v_equiv_rel} can be used to compute the OPEs of these local operators, recalling the basis change  \eqref{eqn:anyon_basis_Gassmann_example}. Condensable algebras are commutative, hence the local operators also have commutative multiplication. The OPEs involving the $\cO^i_{r_{10},n}$ operators differ for the twin algebras $A(H_i,1,1,1)$:
\be
    \cO^i_{r_1} \cO^i_{r_{10},n} = (-1)^{i-1} \cO^i_{r_{10},m},
\ee
for $n=1,3,5,7$, and $m=5,7,1,3$, respectively.
\be
    \cO^i_{r_8, 2} \cO^i_{r_{10},n} = (-1)^{(i-1)\frac{(n+1)}{2}}\cO^i_{r_{10},m}
\ee
for $n=1,3,5,7$, and $m=3,5,7,1$, respectively.
\be
    \cO^i_{r_8, 6} \cO^i_{r_{10},n} = (-1)^{(i-1)\frac{(n+3)}{2}}\cO^i_{r_{10},m}
\ee
for $n=1,3,5,7$, and $m=7,1,3,5$, respectively.

Finally, the OPEs of $\cO^i_{r_{10},n}$ with $\cO^i_{r_{10},m}$ are shown in Table~\ref{tab:G3243r10OPE}.

\begin{table}
$$\def\arraystretch{1.3}
\begin{array}{c|c|c}
     \cO & g & g\cdot \cO \cr
     \hline
     \hline
     \cO^i_{1} & g\in G_{32} & \cO^i_{1} \cr
     \hline
     \cO^i_{r_1} & \makecell[c]{g_1 \\ g_2,g_3} & \makecell[c]{-\cO^i_{r_1}\\ \hspace{1.8ex}\cO^i_{r_1}} \cr
     \hline
     & g_1 & \zeta_8^{-n} \cO^i_{r_{8},n} \cr
     \cO^i_{r_{8},n} & g_2 & \hspace{6.8ex}\cO^i_{r_{8},{n+4}} \cr
     & g_3 & \hspace{4.5ex}\cO^i_{r_{8},{n}} \cr
     \hline
     & g_1 & \zeta_8^{-m}\cO^i_{r_{10},m} \cr
     \cO^i_{r_{10},m} & g_2 & \hspace{8.3ex}\cO^i_{r_{10},P_2(m)} \cr
     & g_3 & \hspace{8.3ex}\cO^i_{r_{10},P_3(m)}
\end{array}
$$
\caption{Symmetry $\Vec_{G_{32}}^\omega$ action on the local operators. Here, $n=2,6$ for $\cO^i_{r_{8},n}$, and the $n+4$ in the subscript of $\cO^i_{r_{8},{n+4}}$ is understood as $n+4 \equiv \text{ mod }{8}$. For the operators $\cO^i_{r_{10},m}$, $m=1,3,5,7$, $P_2=(13)(57)$ and $P_3=(15)(37)$ denote the permutations on the index set $\{1,3,5,7\}$ for $m$. Note that, in this basis, the group action takes the same form, independent on the label $i=1,2$ for the condensable algebra. \label{tab:G3243action}}
\end{table}

\begin{table}
\centering
\scalebox{0.9}{$
\begin{array}{|c||c|c|c|c|}\hline
     & \cO^i_{r_{10},1} & \cO^i_{r_{10},3} & \cO^i_{r_{10},5} & \cO^i_{r_{10},7} \cr
    \hline\hline
    \cO^i_{r_{10},1} &  \cO^i_{r_8, 2} & (-1)^{i-1}\cO^i_{r_1} & (-1)^{i-1}\cO^i_{r_8, 6} & \cO^i_1 \cr\hline
    \cO^i_{r_{10},3} & (-1)^{i-1}\cO^i_{r_1} & \cO^i_{r_8, 6} & \cO^i_1 & (-1)^{i-1}\cO^i_{r_8, 2} \cr\hline
    \cO^i_{r_{10},5} & (-1)^{i-1}\cO^i_{r_8, 6} & \cO^i_1 & \cO^i_{r_8, 2} & (-1)^{i-1}\cO^i_{r_1} \cr\hline
    \cO^i_{r_{10},7} & \cO^i_1 & (-1)^{i-1}\cO^i_{r_8, 2} & (-1)^{i-1}\cO^i_{r_1} & \cO^i_{r_8, 6} \cr\hline
\end{array}
$}
\caption{Multiplication table of the local operators $\cO^i_{r_{10},n}$ obtained from the $([1],r_{10})$ anyon in $A(H_i,1,1,1)$, $i=1,2$.\label{tab:G3243r10OPE}}
\end{table}

The $(-1)^{i-1}$ factor could be eliminated by re-defining $\cO^i_{r_{10},m}$ in \eqref{eqn:localopr10} with no $(-1)^{i-1}$ on the right-hand-side. Such a re-definition, would, however, introduce $(-1)^{i-1}$ factors in the symmetry action in Table~\ref{tab:G3243action} where $\cO^i_{r_{10},m}$ appears. In general, as a result of Schur's lemma, the local operators $\cO^i_{r_{10},m}$ cannot be redefined to produce the same OPEs while respecting identical symmetry $\Vec_{G_{32}}^\omega$ actions for $i=1$ and $i=2$. Hence, we again conclude that although the algebras $A(H_i,1,1,1)$ are defined on the same anyons, they correspond to physically distinct phases.

To summarize, we have two choices of basis for the local operators obtained from condensing $A(H_i,1,1,1)$. One choice \eqref{eqn:G3243GassmannBasis} results in local operators with identical multiplication structure but different $\Vec_{G_{32}}^\omega$ symmetry action. The other choice, related to the first via \eqref{eqn:anyon_basis_Gassmann_example} and \eqref{eqn:localop1}-\eqref{eqn:localopr10} gives local operators on which $\Vec_{G_{32}}^\omega$ acts identically, but whose OPEs are distinct.

\subsection{Phase Transitions Between Gapped Phases without Hidden Symmetry Breaking}

Perhaps the most interesting applications for physical systems of twin algebras is the possibility of constructing intrinsically non-Landau transition, in the sense that they are not SSB breaking transitions, in any gauge related frame. This is the topic of the companion paper \cite{Warman:2026gfz}. There we discuss such DQCP transitions for the group $G^\omega= GL(2,3)^\omega$, which are intrinsically non-Landau in this sense.

Here we will focus on the smallest group with such twin Lagrangian algebras and transitions:  $\cZ(\Vec_{G_{32}}^\omega)$, shown in table \ref{tab:TwinG3243omega}, and their corresponding phases and transitions. 
Here, $\omega$ is a representative of a non-trivial 3-cocycle computed with GAP \cite{GAP4,HAP,mignard2017moritaequivalencepointedfusion}, which we provide in \texttt{G32-43-omega.g}. 

As summarized in table \ref{tab:TwinG3243omega}, $\cZ(\Vec_{G_{32}}^\omega)$ has a pair of non-maximal Gassmann-triple twin algebras $A(H_i,1,1,1)$. The subgroups $H_i \subset G_{32}$ were defined in \eqref{eqn:HinG32-43}, which we recall for convenience 
\be
    \begin{split}
        H_1 & =\{1,\,g_2,\,g_3,\,
        g_2 g_3\}\cong \Z_2 \times \Z_2,\\
    H_2 & =\{1,\, g_2 g_1^4,\, g_3 g_1^4, \,g_2g_3\} \cong \Z_2 \times \Z_2.
    \end{split}
\ee 
We summarize the sub-diagram of the Hasse diagram of condensable algebras in $\cZ(\Vec_{G_{32}}^\omega)$ connected to the twin algebras $A(H_i,1,1,1)$ in Figure \ref{fig:Hasse3243_om}. 
{\twocolumngrid
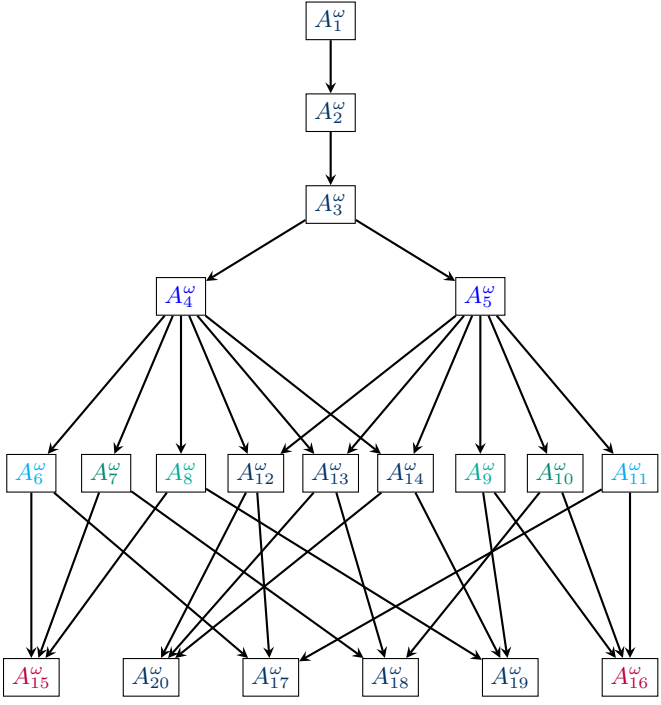
\begin{figure}
\begin{center}
\begin{tikzpicture}[vertex/.style={draw}, scale=0.9]
 \begin{scope}[shift={(0,0)}]
    \foreach \coord/\i/\j in {
(0.,-1.35)/1/{\nameref{Aom:1}},
(0.,-2.7)/2/{\nameref{Aom:2}},
(0.,-4.05)/3/{\nameref{Aom:3}},
(-2.2,-5.4)/4/{\nameref{Aom:4}},
(2.2,-5.4)/5/{\nameref{Aom:5}},
(-4.4,-8)/6/{\nameref{Aom:6}},
(-3.3,-8)/7/{\nameref{Aom:7}},
(-2.2,-8)/8/{\nameref{Aom:8}},
(2.2,-8)/9/{\nameref{Aom:9}},
(3.3,-8)/10/{\nameref{Aom:10}},
(4.4,-8)/11/{\nameref{Aom:11}},
(-1.1,-8)/12/{\nameref{Aom:12}},
(0.,-8)/13/{\nameref{Aom:13}},
(1.1,-8)/14/{\nameref{Aom:14}},
(-4.4,-11)/15/{\nameref{Aom:15}},
(4.4,-11)/16/{\nameref{Aom:16}},
(-0.880,-11)/17/{\nameref{Aom:17}},
(0.880,-11)/18/{\nameref{Aom:18}},
(2.64,-11)/19/{\nameref{Aom:19}},
(-2.64,-11)/20/{\nameref{Aom:20}}}
     {
      \node[vertex,align=center] (p\i) at \coord {\j};
     }
  \foreach [count=\r] \row in 
{{0,1,0,0,0,0,0,0,0,0,0,0,0,0,0,0,0,0,0,0},
{0,0,1,0,0,0,0,0,0,0,0,0,0,0,0,0,0,0,0,0},
{0,0,0,1,1,0,0,0,0,0,0,0,0,0,0,0,0,0,0,0},
{0,0,0,0,0,1,1,1,0,0,0,1,1,1,0,0,0,0,0,0},
{0,0,0,0,0,0,0,0,1,1,1,1,1,1,0,0,0,0,0,0},
{0,0,0,0,0,0,0,0,0,0,0,0,0,0,1,0,1,0,0,0},
{0,0,0,0,0,0,0,0,0,0,0,0,0,0,1,0,0,1,0,0},
{0,0,0,0,0,0,0,0,0,0,0,0,0,0,1,0,0,0,1,0},
{0,0,0,0,0,0,0,0,0,0,0,0,0,0,0,1,0,0,1,0},
{0,0,0,0,0,0,0,0,0,0,0,0,0,0,0,1,0,1,0,0},
{0,0,0,0,0,0,0,0,0,0,0,0,0,0,0,1,1,0,0,0},
{0,0,0,0,0,0,0,0,0,0,0,0,0,0,0,0,1,0,0,1},
{0,0,0,0,0,0,0,0,0,0,0,0,0,0,0,0,0,1,0,1},
{0,0,0,0,0,0,0,0,0,0,0,0,0,0,0,0,0,0,1,1},
{0,0,0,0,0,0,0,0,0,0,0,0,0,0,0,0,0,0,0,0},
{0,0,0,0,0,0,0,0,0,0,0,0,0,0,0,0,0,0,0,0},
{0,0,0,0,0,0,0,0,0,0,0,0,0,0,0,0,0,0,0,0},
{0,0,0,0,0,0,0,0,0,0,0,0,0,0,0,0,0,0,0,0},
{0,0,0,0,0,0,0,0,0,0,0,0,0,0,0,0,0,0,0,0},
{0,0,0,0,0,0,0,0,0,0,0,0,0,0,0,0,0,0,0,0}}
    {
     \foreach [count=\c] \cell in \row{
            \ifnum\cell=1%
                \draw[-stealth] (p\r) edge [thick] (p\c);
            \fi
        }
    }
\end{scope}
\end{tikzpicture}
\caption{Partial Hasse diagram for $\cZ(\Vec_{G_{32}}^{\omega})$. The algebras are detailed in Table \ref{tab:G_om}. The top most algebra is the identity algebra. The bottom row are the Lagrangian algebras, that correspond to topological boundary conditions. 
Twin algebras are (\nameref{Aom:4},\nameref{Aom:5}), (\nameref{Aom:6},\nameref{Aom:11}), (\nameref{Aom:8},\nameref{Aom:9}), (\nameref{Aom:7},\nameref{Aom:10}), (\nameref{Aom:15},\nameref{Aom:16}). \label{fig:Hasse3243_om}}
\end{center}
\end{figure}}

\subsubsection{Twin Lagrangian Algebras from Lagrangians in Reduced TO}
For the purpose of studying gapped phases, we focus on twin Lagrangian algebras
\be\ba
\text{\nameref{Aom:15}}&=A(H_1, H_1, 1, 1)\\
\text{\nameref{Aom:16}}&=A(H_2, H_2, 1, 1)
\ea\ee
whose anyon decomposition is ($i=15,16$)
\begin{align}
\label{eqn:anyon-decomp-Aom}
A^\omega_{i} &\cong ([1],1) \oplus ([1],r_1) \oplus ([1],r_8) \oplus ([1],r_{10}) \nn\\ 
& \oplus ([g_2g_3],2) \oplus ([g_3],2\;g_1^6=-\sqrt{2}\zeta_{16})  \nn \\
& \oplus ([g_2],2) \oplus ([g_3],2\;g_1^6=\sqrt{2}\zeta_{16})\,. 
\end{align}

To study phase transitions, we consider their common subalgebra\footnote{As illustrated in Figure~\ref{fig:Hasse3243_om}, all pairs of twin algebras (\nameref{Aom:4},\nameref{Aom:5}), (\nameref{Aom:6},\nameref{Aom:11}), (\nameref{Aom:8},\nameref{Aom:9}), (\nameref{Aom:7},\nameref{Aom:10}), (\nameref{Aom:15},\nameref{Aom:16}) admit \nameref{Aom:3} as a common subalgebra.} 
\begin{align}
\text{\nameref{Aom:3}}&=A(H, 1, 1, 1) \nn\\
&\cong ([1],1) \oplus ([1],r_1) \oplus ([1],r_8)\,,
\end{align}
where
\be
    H=\langle g_2g_3,g_3,g_1^4 \rangle\cong\Z_2\times\Z_2\times\Z_2\,.
\ee
The reduced TO associated to $\text{\nameref{Aom:3}}$ is $\cZ(\Vec_{\Z_2^3}^{\omega_{\text{III}}})$ with type-III anomaly $\omega_\III:=\omega|_H$ defined as
\be
    \omega_\III(a^{i_1}b^{i_2}c^{i_3},a^{j_1}b^{j_2}c^{j_3},a^{k_1}b^{k_2}c^{k_3})=(-1)^{i_1j_2k_3}
\ee
for generators of each $\Z_2$ factor
\be
    a=g_2g_1^4\,,\quad b=g_3\,, \quad c=g_2g_3\,.
\ee
Note that $\omega|_{H_i}\equiv1$ in $H^3(H_i, U(1))$ for each $i=1,2$.

Denoting the electric and magnetic anyon generators of $\cZ(\Vec_{\Z_2^3}^{\omega_{\text{III}}})$ by $e_a,e_b,e_c$ and $([a],2),([b],2),([c],2)$ respectively, \noindent\nameref{Aom:3} defines the following map of anyons:\footnote{Note that this is only part of the map of anyons data. The full map can be found in Appendix~\ref{app:map-of-anyon-G32omega}.}
\begin{align}
\label{eqn:map-of-anyon-A3omega}
    \cZ(\Vec_{\Z_2^3}^{\omega_{\text{III}}}) & \rightarrow \cZ(\Vec_{G_{32}}^\omega) \cr
    e_{a} &\mapsto ([1],r_{10}) \cr 
    e_{b} &\mapsto ([1],r_{10}) \cr 
    e_{c} &\mapsto ([1],r_{10}) \cr 
    ([a],2) &\mapsto ([g_2],2) \cr 
    ([bc],2) &\mapsto ([g_2],2) \cr 
    ([c],2) &\mapsto ([g_2g_3],2) \cr 
    ([b],2) &\mapsto ([g_3],2\, g_1^6=-\sqrt{2}\zeta_{16}) \oplus ([g_3],2\, g_1^6=\sqrt{2}\zeta_{16}) \cr 
    ([ac],2) &\mapsto ([g_3],2\, g_1^6=-\sqrt{2}\zeta_{16}) \oplus ([g_3],2\, g_1^6=\sqrt{2}\zeta_{16}) \cr 
\end{align}
In particular, it sends the two (non-twin) Lagrangian algebras in $\cZ(\Vec_{\Z_2^3}^{\omega_{\text{III}}})$
\begin{align} \label{eq:tilde_A15}
    A_{15}^{\Z_2^3,\omega_{\text{III}}} & \cong 1\oplus e_a\oplus([c],2)\oplus([b],2)\oplus([bc],2)\\
    A_{16}^{\Z_2^3,\omega_{\text{III}}} & \cong 1 \oplus e_b \oplus ([c],2) \oplus ([ac],2) \oplus ([a],2)\label{eq:tilde_A16}
\end{align}
to the twin Lagrangian algebras \nameref{Aom:15} and \nameref{Aom:16} in the original $\cZ(\Vec_{G_{32}}^\omega)$, whose anyon decomposition is \eqref{eqn:anyon-decomp-Aom}.

\subsubsection{Phase Transitions Respecting Group Symmetries}
Fixing the symmetry Lagrangian algebra of the SymTFT to be \nameref{Aom:20}$=\cL(1,1)$ defines the Dirichlet $\Vec_{G_{32}}^\omega$ symmetry boundary.
Denote the gapped $\Vec_{G_{32}}^\omega$-symmetric twin phases obtained from $A^\omega_{15}$ and $A^\omega_{16}$ by $\cP^{\Vec_{G_{32}}^\omega}_{A^\omega_{15}}$ and $\cP^{\Vec_{G_{32}}^\omega}_{A^\omega_{16}}$, respectively. 
As in \eqref{notation:vik}, denote vacua obtained from untwisted local operators in the two phases by $v_k^i$ ($i=1$ corresponds to phase $\cP^{\Vec_{G_{32}}^\omega}_{A^\omega_{15}}$ and $i=2$ for $\cP^{\Vec_{G_{32}}^\omega}_{A^\omega_{16}}$), $k\in \{0, 1, \cdots, 7\}$. 
The anomalous ${G_{32}}$-action on vacuum $v_k^i$ takes the same form as that in \Cref{sec:basis-diff-action-same-OPE}. E.g., $g_1\in {G_{32}}$ permutes the vacua cyclically $g_1\cdot v^i_{k} = v^i_{k+1}$ for both $i=1,2$, see equation~\eqref{eq:g1_act_H1H2}.

The gapped phase $\cP^{\Vec_{G_{32}}^\omega}_{A^\omega_{15}}$ is:
\be
\label{eqn:VecG32A15phase_omega}
\begin{tikzpicture}[baseline]
\begin{scope}[shift={(-1,2.8)}]
\node at (1.2,-3) {$\cP^{\Vec_{G_{32}}^\omega}_{A^\omega_{15}}=$};
\node (1) at (2.5,-3) {$v^1_{0}$};
\node at (2.5,-4) {$g_2g_3$};
\node[RoyalBlue] at (2.5,-2) {$g_3$};
\node at (3.5,-3) {$\oplus$};
\node[red] at (3.5,-3.7) {$g_1$};
\node (2) at (4.5,-3) {$v^1_{1}$};
\node at (4.5,-4) {${}^{g_1}(g_2g_3)$};
\node[RoyalBlue] at (4.5,-2) {${}^{g_1}g_3$};
\node (2b) at (4.55,-3.2) {};
\node at (5.5,-3) {$\oplus$};
\node[red] at (5.5,-3.7) {$g_1$};
\node (3) at (6.5,-3) {$\cdots$};
\node (3b) at (6,-3.3) {};
\node (3c) at (7,-3.3) {};
\node at (7.5,-3) {$\oplus$};
\node[red] at (7.5,-3.7) {$g_1$};
\node (4) at (8.5,-3) {$v^1_{7}$};
\node at (8.5,-3.9) {${}^{g_1^7}(g_2g_3)$};
\node[RoyalBlue] at (8.5,-1.95) {${}^{g_1^7}g_3$};
\node (4b) at (6.55,-3.2) {};
\draw[red,-stealth] (1) [bend right=30] to (2);
\draw[red,-stealth] (2) [bend right=30] to (3b);
\draw[red,-stealth] (3c) [bend right=30] to (4);
\draw[black,-stealth] (1) edge [in=-70, out=-110,looseness=5] (1);
\draw[black,-stealth] (2) edge [in=-70, out=-110,looseness=5] (2);
\draw[black,-stealth] (4) edge [in=-70, out=-110,looseness=5] (4);
\draw[RoyalBlue,-stealth] (1) edge [in=70, out=110,looseness=5] (1);
\draw[RoyalBlue,-stealth] (2) edge [in=70, out=110,looseness=5] (2);
\draw[RoyalBlue,-stealth] (4) edge [in=70, out=110,looseness=5] (4);
\end{scope}
\end{tikzpicture} 
\ee
while the phase $\cP^{\Vec_{G_{32}}^\omega}_{A^\omega_{16}}$ corresponding to $A_{16}$ is
\be
\label{eqn:VecG32A16phase_omega}
\begin{tikzpicture}[baseline]
\begin{scope}[shift={(-1,2.8)}]
\node at (1.2,-3) {$\cP^{\Vec_{G_{32}}^\omega}_{A^\omega_{16}}=$};
\node (1) at (2.5,-3) {$v^2_{0}$};
\node at (2.5,-4) {$g_2g_3$};
\node[Blue] at (2.5,-2) {$g_3g_1^4$};
\node at (3.5,-3) {$\oplus$};
\node[red] at (3.5,-3.7) {$g_1$};
\node (2) at (4.5,-3) {$v^2_{1}$};
\node at (4.5,-4) {${}^{g_1}(g_2g_3)$};
\node[Blue] at (4.5,-2) {${}^{g_1}(g_3g_1^4)$};
\node (2b) at (4.55,-3.2) {};
\node at (5.5,-3) {$\oplus$};
\node[red] at (5.5,-3.7) {$g_1$};
\node (3) at (6.5,-3) {$\cdots$};
\node (3b) at (6,-3.3) {};
\node (3c) at (7,-3.3) {};
\node at (7.5,-3) {$\oplus$};
\node[red] at (7.5,-3.7) {$g_1$};
\node (4) at (8.5,-3) {$v^2_{7}$};
\node at (8.5,-3.9) {${}^{g_1^7}(g_2g_3)$};
\node[Blue] at (8.5,-1.95) {${}^{g_1^7}(g_3g_1^4)$};
\node (4b) at (6.55,-3.2) {};
\draw[red,-stealth] (1) [bend right=30] to (2);
\draw[red,-stealth] (2) [bend right=30] to (3b);
\draw[red,-stealth] (3c) [bend right=30] to (4);
\draw[black,-stealth] (1) edge [in=-70, out=-110,looseness=5] (1);
\draw[black,-stealth] (2) edge [in=-70, out=-110,looseness=5] (2);
\draw[black,-stealth] (4) edge [in=-70, out=-110,looseness=5] (4);
\draw[Blue,-stealth] (1) edge [in=70, out=110,looseness=5] (1);
\draw[Blue,-stealth] (2) edge [in=70, out=110,looseness=5] (2);
\draw[Blue,-stealth] (4) edge [in=70, out=110,looseness=5] (4);
\end{scope}
\end{tikzpicture}
\ee
{Note one can consider a new basis $v^{2\prime}_{k}$ such that the symmetry action on the resulting phase $\cP^{\Vec_{G_{32}}^\omega}_{A_{16}}$ is identical to its twin, but the OP OPEs will be different: } 
\be
\label{eqn:VecG32A16phase_omega-samesymm}
\begin{tikzpicture}[baseline]
\begin{scope}[shift={(-1,2.8)}]
\node at (1.2,-3) {$\cP^{\Vec_{G_{32}}^\omega}_{A^\omega_{16}}=$};
\node (1) at (2.5,-3) {$v^{2\prime}_{0}$};
\node at (2.5,-4) {$g_2g_3$};
\node[RoyalBlue] at (2.5,-2) {$g_3$};
\node at (3.5,-3) {$\oplus$};
\node[red] at (3.5,-3.7) {$g_1$};
\node (2) at (4.5,-3) {$v^{2\prime}_{1}$};
\node at (4.5,-4) {${}^{g_1}(g_2g_3)$};
\node[RoyalBlue] at (4.5,-2) {${}^{g_1}g_3$};
\node (2b) at (4.55,-3.2) {};
\node at (5.5,-3) {$\oplus$};
\node[red] at (5.5,-3.7) {$g_1$};
\node (3) at (6.5,-3) {$\cdots$};
\node (3b) at (6,-3.3) {};
\node (3c) at (7,-3.3) {};
\node at (7.5,-3) {$\oplus$};
\node[red] at (7.5,-3.7) {$g_1$};
\node (4) at (8.5,-3) {$v^{2\prime}_{7}$};
\node at (8.5,-3.9) {${}^{g_1^7}(g_2g_3)$};
\node[RoyalBlue] at (8.5,-1.95) {${}^{g_1^7}g_3$};
\node (4b) at (6.55,-3.2) {};
\draw[red,-stealth] (1) [bend right=30] to (2);
\draw[red,-stealth] (2) [bend right=30] to (3b);
\draw[red,-stealth] (3c) [bend right=30] to (4);
\draw[black,-stealth] (1) edge [in=-70, out=-110,looseness=5] (1);
\draw[black,-stealth] (2) edge [in=-70, out=-110,looseness=5] (2);
\draw[black,-stealth] (4) edge [in=-70, out=-110,looseness=5] (4);
\draw[RoyalBlue,-stealth] (1) edge [in=70, out=110,looseness=5] (1);
\draw[RoyalBlue,-stealth] (2) edge [in=70, out=110,looseness=5] (2);
\draw[RoyalBlue,-stealth] (4) edge [in=70, out=110,looseness=5] (4);
\end{scope}
\end{tikzpicture} 
\ee

For $g\in {G_{32}}$, the $g$-twisted sector operators in $\cP^{\Vec_G^\omega}_{A^\omega_{15}}$ and $\cP^{\Vec_G^\omega}_{A^\omega_{16}}$ come from elements $v_{r,n}$ such that $rnr^{-1}=g$ (recall notation \eqref{eqn:basiselm}) in algebras $A^\omega_{15}$ and $A^\omega_{16}$, respectively. The vector spaces of all twisted operators in both phases form non-associative algebras, see~\eqref{eqn:cond_alg_multiplication}. 

\vspace{2mm}
\noindent\textbf{Phase Transition between Twin Gapped Phases.} 
A phase transition between $\cP^{\Vec_G^\omega}_{A^\omega_{15}}$ and $\cP^{\Vec_G^\omega}_{A^\omega_{16}}$ can be achieved by inputting a $\Vec_{\Z_2^3}^{\omega_{\text{III}}}$-symmetric CFT, denoted by $\cP^{{\omega_{\text{III}}}}_{\text{CFT}}$ (with one possible choice being the $\Spin(4)_1$ WZW CFT studied in \cite{Ando:2026ffy}) on the physical boundary of the reduced TO $\cZ(\Vec_{\Z_2^3}^{\omega_{\text{III}}})$, where $\omega_{\text{III}}$ is the type III cocycle:
\begin{equation}
\begin{split}    
    \begin{tikzpicture}
\begin{scope}[shift={(0,0)}, scale=1.5]
    \draw[SymTFTColores, fill= SymTFTColores, opacity = 0.2]   (0,0) -- (0,2) -- (2,2) -- (2,0) --(0,0); 
    \draw[SymTFTColores, fill= SymTFTColores, opacity = 0.1]  (2,0) -- (2,2) -- (4,2) -- (4,0) --(2,0); 
    \draw [very thick]  (0,0) -- (0,2); 
    \draw [very thick]  (2,0) -- (2,2); 
    \draw [very thick]  (4,0) -- (4,2); 
    \node[above] at (0,2) {$G^\omega$};
    \node[above] at (2,2) {$\cI_{A_3^\omega}$};
    \node[above] at (4,2) {$\cP^{{\omega_{\text{III}}}}_{\text{CFT}}$};
    \node at (1,1) {$\cZ({G},\omega)$}; 
    \node at (3,1) {$\cZ({\Z_2^3},{\omega_{\text{III}}})$};
    \end{scope}
    \end{tikzpicture}
\end{split}\end{equation}
Compactifying the above SymTFT results in a phase transition between $\cP^{\Vec_G^\omega}_{A^\omega_{15}}$ and $\cP^{\Vec_G^\omega}_{A^\omega_{16}}$.

The theory for the phase transition, corresponding to the algebra $A_3^\omega$, is:
\be \label{eq:G-Aom3-theory}
\begin{tikzpicture}[baseline]
\begin{scope}[shift={(-0.5,2.8)}]
\node at (1.2,-1.5) {$\cP^{\Vec^\omega_G}_{A_3^\omega}=$};
\node (1) at (2.5,-3) {$\lb\cP^{{\omega_{\text{III}}}}_{\text{CFT}}\rb_{0}$};
\node at (2.5,-4) {$g_2g_3$};
\node[RoyalBlue] at (1.85,-2.5) {$g_3$};
\node at (2.5,-2) {\textcolor{Blue}{$g_3g_1^4$}};
\node at (3.5,-3) {$\oplus$};
\node[red] at (3.5,-3.7) {$g_1$};
\node (2) at (4.5,-3) {$\lb\cP^{{\omega_{\text{III}}}}_{\text{CFT}}\rb_{1}$};
\node at (4.5,-4) {${}^{g_1}(g_2g_3)$};
\node[RoyalBlue] at (3.85,-2.45) {${}^{g_1}g_3$};
\node at (4.5,-2) {\textcolor{Blue}{${}^{g_1}(g_3g_1^4)$}};
\node (2b) at (4.55,-3.2) {};
\node at (5.5,-3) {$\oplus$};
\node[red] at (5.5,-3.7) {$g_1$};
\node at (6.5,-4) {${}^{g_1^2}(g_2g_3)$};
\node[RoyalBlue] at (5.85,-2.4) {${}^{g_1^2}g_3$};
\node at (6.5,-2) {\textcolor{Blue}{${}^{g_1^2}(g_3g_1^4)$}};
\node (3) at (6.5,-3) {$\lb\cP^{{\omega_{\text{III}}}}_{\text{CFT}}\rb_{2}$};
\node (3b) at (6,-3.3) {};
\node (3c) at (7,-3.3) {};
\node at (7.5,-3) {$\oplus$};
\node[red] at (7.5,-3.7) {$g_1$};
\node (4) at (8.5,-3) {$\lb\cP^{{\omega_{\text{III}}}}_{\text{CFT}}\rb_{3}$};
\node at (8.5,-3.9) {${}^{g_1^3}(g_2g_3)$};
\node[RoyalBlue] at (7.85,-2.4) {${}^{g_1^3}g_3$};
\node at (8.5,-1.95) {\textcolor{Blue}{${}^{g_1^3}(g_3g_1^4)$}};
\node (4b) at (6.55,-3.2) {};
\draw[red,-stealth] (1) [bend right=30] to (2);
\draw[red,-stealth] (2) [bend right=30] to (3b);
\draw[red,-stealth] (3c) [bend right=30] to (4);
\draw[black,-stealth] (1) edge [in=-70, out=-110,looseness=5] (1);
\draw[black,-stealth] (2) edge [in=-70, out=-110,looseness=5] (2);
\draw[black,-stealth] (3) edge [in=-70, out=-110,looseness=5] (3);
\draw[black,-stealth] (4) edge [in=-70, out=-110,looseness=5] (4);
\draw[RoyalBlue,-stealth] (1) edge [in=70, out=110,looseness=5] (1);
\draw[Blue,-stealth] (1) edge [in=60, out=120,looseness=4.5] (1);
\draw[RoyalBlue,-stealth] (2) edge [in=70, out=110,looseness=5] (2);
\draw[Blue,-stealth] (2) edge [in=60, out=120,looseness=4.5] (2);
\draw[RoyalBlue,-stealth] (3) edge [in=70, out=110,looseness=5] (3);
\draw[Blue,-stealth] (3) edge [in=60, out=120,looseness=4.5] (3);
\draw[RoyalBlue,-stealth] (4) edge [in=70, out=110,looseness=5] (4);
\draw[Blue,-stealth] (4) edge [in=60, out=120,looseness=4.5] (4);
\end{scope}
\end{tikzpicture} 
\ee
Using the club sandwich SymTFT approach, inputting $\cP^{{\omega_{\text{III}}}}_{\text{CFT}}$ as physical boundary of $\cZ(\Vec_{\Z_2^3}^{\omega_{\III}})$, 
the gapped twin phases $\cP^{\Vec_G^\omega}_{A^\omega_{15}}$ and $\cP^{\Vec_G^\omega}_{A^\omega_{16}}$ can be obtained from $\cP^{\Vec_G^\omega}_{A^\omega_3}$ via the following two sequential condensations:
\begin{itemize}
    \item For $\cP^{\Vec_G^\omega}_{A^\omega_{15}}$, adding to $\cP^{{\omega_{\text{III}}}}_{\text{CFT}}$ the local relevant deformation $e_a$ spontaneously breaks $\Z_2^{\langle a \rangle}$. This leads to the gapless phase corresponding to the algebra \nameref{Aom:4}, which is a gSSB with eight vacua. Next, condensing the anyons $([c],2),([b],2),([bc],2)$ and using the map of anyons to $\cZ(\Vec_G^\omega)$ \eqref{eqn:map-of-anyon-A3omega}, each $\cP^{{\omega_{\text{III}}}}_{\text{CFT}}$ in $\cP^{\Vec_G^\omega}_{A_3^\omega}$ is driven to a gapped phase, preserving subgroups conjugate to $H_1$, while $g_2g_1^4$ is spontaneously broken: we thus obtain $\cP^{\Vec_G^\omega}_{A_{15}^\omega}$ \eqref{eqn:VecG32A15phase_omega}.
    
    \item For $\cP^{\Vec_G^\omega}_{A^\omega_{16}}$, adding to $\cP^{{\omega_{\text{III}}}}_{\text{CFT}}$ the local relevant deformation $e_b$ spontaneously breaks $\Z_2^{\langle b \rangle}$. This leads to the gapless twin phase corresponding to the algebra \nameref{Aom:5}, which has eight universes. 
    Next, condensing the anyons $([c],2),([a],2),([ac],2)$ and using the map of anyons to $\cZ(\Vec_G^\omega)$ \eqref{eqn:map-of-anyon-A3omega}, each $\cP^{{\omega_{\text{III}}}}_{\text{CFT}}$ in $\cP^{\Vec_G^\omega}_{A_3^\omega}$ is driven to a gapped phase, preserving subgroups conjugate to $H_2$, while $g_3$ is spontaneously broken: we thus obtain $\cP^{\Vec_G^\omega}_{A_{16}^\omega}$ \eqref{eqn:VecG32A16phase_omega}.
    \end{itemize}

\begin{equation}\begin{split}
\begin{tikzpicture}[baseline]
\begin{scope}[shift={(-1,2.8)}]
\node (H) at (0,0) {$\cP^{\Vec_G^\omega}_{A^\omega_3}$};
\node (H1) at (-1,-1.5) {$\cP^{\Vec^\omega_G}_{A^\omega_{15}}$};
\node[Blue,left] at (-1.35,-0.8) {$\langle A_{15}^{\Z_2^3,\omega_{\text{III}}}\mapsto A_{15}^\omega\rangle\neq 0$};
\node[RoyalBlue,right] at (1.35,-0.8) {$\langle A_{16}^{\Z_2^3,\omega_{\text{III}}}\mapsto A_{16}^\omega\rangle\neq 0$};
\node (H2) at (1,-1.5) {$\cP^{\Vec^\omega_G}_{A^\omega_{16}}$};
\draw[Blue,-stealth] (H) to (H1);
\draw[RoyalBlue,-stealth] (H) to (H2);
\end{scope}
\end{tikzpicture}
\end{split}\end{equation}
The different order/disorder parameters \eqref{eq:tilde_A15}-\eqref{eq:tilde_A16} for the input transition on the $\cZ(\Vec_{{\Z_2^3}}^{\omega_{\text{III}}})$ physical boundary are mapped to the same anyons in the twin Lagrangian algebras \eqref{eqn:map-of-anyon-A3omega}, giving rise to transition \eqref{eq:G-Aom3-theory} for twin gapped phases with anomalous $G$-symmetry.

Importantly, the anomaly prevents the pair of Landau transition between gapped phases preserving the groups $H_1<H$, as well as  $H>H_2$ with $H$ a subgroup containing both $H_1$ and $H_2$. Instead, changing the sign of the relevant deformation of the order parameter for the $\cP^{\Vec^\omega_G}_{A^\omega_{15}}$ preserving $H_1$, drives the system to a gapless phase with larger symmetry $H$ and then to gap the system by breaking minimal symmetry, it gets driven to the theory that preserves $H_2$, i.e. to the phase $\cP^{\Vec^\omega_G}_{A^\omega_{16}}$. A more in depth characterization of transitions including deconfined quantum critical points (DQCPs) through the structure of condensable algebras will be discussed in \cite{SWW}.

\subsubsection{Phase Transitions Respecting Non-Invertible Symmetries}

We now consider symmetric phases associated to the twin algebras with respect to a non-invertible symmetry $\cS^\omega:=(\Vec_{G_{32}}^\omega)_{\cM(H_1,1)}^*$, obtained by taking the symmetry Lagrangian algebra of the SymTFT to be \nameref{Aom:15}
\begin{align}
    \text{\nameref{Aom:15}}&=A(H_1, H_1, 1, 1) \\
    &\cong ([1],1) \oplus ([1],r_1) \oplus ([1],r_8) \oplus ([1],r_{10}) \nn\\ 
    & \quad \oplus ([g_2g_3],2) \oplus ([g_3],2\;g_1^6=-\sqrt{2}\zeta_{16})  \nn \\
    & \quad \oplus ([g_2],2) \oplus ([g_3],2\;g_1^6=\sqrt{2}\zeta_{16})\,. 
\end{align}
In other words, the symmetry category $\cS^\omega$ is obtained by gauging $H_1\subset G$ in $\Vec_G^\omega$. It has 8 simple objects, forming a set denoted by $\text{Irr}(\cS^\omega)$, consisting of:
\begin{itemize}
    \item Four 1-dimensional objects denoted by
    \begin{equation}
        S^\omega_{1,1_h},
    \end{equation}
    where $1_h$ denotes irreducible representation of $H_1$:
    \be
    \label{eqn:1drep-H1}
        \begin{array}{|c|c c c c|}\hline
         & 1 & g_2 g_3 & g_3 & g_2 \\
         \hline
    1_1 & 1 & 1 & 1 & 1 \\
    1_{g_2} & 1 & -1 & -1 & 1 \\
    1_{g_3} & 1 & -1 & 1 & -1 \\
    1_{g_2g_3} & 1 & 1 & -1 & -1 \\
    \hline
        \end{array}
    \ee
    
    \item Three 2-dimensional objects denoted by
    \begin{equation}
        S^{\omega}_{[g_1^2], +1}\,, \quad S^{\omega}_{[g_1^2], -1}\,, \quad S^\omega_{g_1^4}\,,
    \end{equation}

    \item One 4-dimensional object
    \begin{equation}
        S^\omega_{[g_1]}
    \end{equation}
\end{itemize}

The fusion rules are symmetric, being
\begin{equation}
    \begin{split}
        S^\omega_{1,1_h} \otimes S^\omega_{1,1_f} & \cong S^\omega_{1,1_{hf}}\,, \\
        S^\omega_{1,1_h} \otimes S^{\omega}_{[g_1^2], \pm 1} & \cong \begin{cases}
            S^{\omega}_{[g_1^2], \pm 1}, & h \in \{1, g_3\} \\
            S^{\omega}_{[g_1^2], \mp 1}, & h \in \{ g_2, g_2g_3\}
        \end{cases}\,,\\
        S^\omega_{1,1_h} \otimes S^\omega_{g_1^4} & \cong S^\omega_{g_1^4} \,, \\
        S^\omega_{1,1_h} \otimes S^\omega_{[g_1]} & \cong S^\omega_{[g_1]}\,, \\
        S^{\omega}_{[g_1^2], \pm 1} \otimes S^{\omega}_{[g_1^2], \pm 1} & \cong S^\omega_{1,1} \oplus S^\omega_{1,1_{g_3}} \oplus S^\omega_{g_1^4}\,,\\
        S^{\omega}_{[g_1^2], \pm 1} \otimes S^{\omega}_{[g_1^2], \mp 1} & \cong S^\omega_{1,1_{g_2}} \oplus S^\omega_{1,1_{g_2 g_3}} \oplus S^\omega_{g_1^4}\,,\\
        S^{\omega}_{[g_1^2], \pm 1} \otimes S^\omega_{g_1^4} & \cong S^{\omega}_{[g_1^2], + 1} \oplus S^{\omega}_{[g_1^2], - 1}\,, \\
        S^{\omega}_{[g_1^2], \pm 1} \otimes S^\omega_{[g_1]} & \cong 2 S^\omega_{[g_1]}\,,\\
        S^\omega_{g_1^4} \otimes S^\omega_{g_1^4} & \cong \bigoplus_{h\in H_1} S^\omega_{1,1_h}\,,\\
        S^\omega_{g_1^4} \otimes S^\omega_{[g_1]} & \cong 2 S^\omega_{[g_1]} \,,\\
        S^\omega_{[g_1]} \otimes S^\omega_{[g_1]} & \cong \bigoplus_{S \in \text{Irr}(\cS^\omega)\setminus S^\omega_{[g_1]}} \dim(S) S\,,
    \end{split}
\end{equation}
where $h,f\in H_1$. The 4-dimensional object $S^\omega_{[g_1]}$ behaves in a similar way to the non-invertible object in a TY-category: its self fusion decomposes into a direct sum over all other simple objects with multiplicities being their quantum dimensions.

We now study the twin phases $\cP^{\cS^\omega}_{A_{15}^\omega}$ and $\cP^{\cS^\omega}_{A_{16}^\omega}$. Starting from $\cP^{\Vec_G^\omega}_{A_{15}^\omega}$ \eqref{eqn:VecG32A15phase_omega}, we gauge $H_1$ to obtain the corresponding theory $\cP^{\cS^\omega}_{A_{15}^\omega}$ with symmetry $\cS^\omega$ (note that the following is a diagrammatic way to represent the fusion rule of $\cS^\omega$, this is expected as $\cP^{\cS^\omega}_{A_{15}^\omega}$ is the $\cS^\omega$-full SSB)
\be
\label{eqn:Som-A15phase}
\begin{tikzpicture}[baseline,yscale=1.1, xscale=1.1]
\node at (1.5,0) {$\cP^{\cS^\omega}_{A^\omega_{15}}=$};
\begin{scope}[shift={(0,0)}]
\node (S0) at (3,2) {$(v^1_0)_{1_1}$};
\node (S0g3) at (3,1) {$(v^1_0)_{1_{g_3}}$};
\node (S0g2) at (3,-2) {$(v^1_0)_{1_{g_2}}$};
\node (S0g2g3) at (3,-1) {$(v^1_0)_{1_{g_2g_3}}$};
\node (S2p) at (5,2) {$(v^1_{2,6})_{+}$};
\node (S2m) at (5,-2) {$(v^1_{2,6})_{-}$};
\node (S4) at (5,0) {$v^1_{4}$};
\node (S1) at (7,0) {$v^1_{1,3,5,7}$};
%
%
\node[RoyalBlue] at (3.7,2.5) {$S^\omega_{[g_1^2],+1}$};
\draw[RoyalBlue, Stealth-Stealth] (S0g3) -- (S2p);
\draw[RoyalBlue, Stealth-Stealth] (S0) -- (S2p);
\draw[RoyalBlue, Stealth-Stealth] (S0g2g3) -- (S2m);
\draw[RoyalBlue, Stealth-Stealth] (S0g2) -- (S2m);
\draw[RoyalBlue, Stealth-Stealth, bend right=10] (S4) to (S2p);
\draw[RoyalBlue, Stealth-Stealth, bend left=10] (S4) to (S2m);
\draw[RoyalBlue,-stealth] (S1) edge [in=70, out=110,looseness=5] (S1);
\node[ForestGreen] at (7.3,1.2) {$S^\omega_{[g_1^2],-1}$};
\draw[ForestGreen, Stealth-Stealth] (S0g3) -- (S2m);
\draw[ForestGreen, Stealth-Stealth] (S0) -- (S2m);
\draw[ForestGreen, Stealth-Stealth] (S0g2g3) -- (S2p);
\draw[ForestGreen, Stealth-Stealth] (S0g2) -- (S2p);
\draw[ForestGreen, Stealth-Stealth, bend left=10] (S4) to (S2p);
\draw[ForestGreen, Stealth-Stealth, bend right=10] (S4) to (S2m);
\draw[ForestGreen,-stealth] (S1) edge [in=60, out=120,looseness=6] (S1);
%
\node[black] at (5.5,2.5) {$S^\omega_{g_1^4}$};
\draw[black, Stealth-Stealth] (S0) -- (S4);
\draw[black, Stealth-Stealth] (S0g3) -- (S4);
\draw[black, Stealth-Stealth] (S0g2) -- (S4);
\draw[black, Stealth-Stealth] (S0g2g3) -- (S4);
\draw[black,-stealth] (S2p) edge [in=70, out=110,looseness=5] (S2p);
\draw[black,-stealth] (S2m) edge [in=-70, out=-110,looseness=5] (S2m);
\draw[black,-stealth] (S1) edge [in=-70, out=-110,looseness=5] (S1);
%
\node[red] at (6.15,1.45) {$S^\omega_{[g_1]}$};
\draw[red, Stealth-Stealth] (S0) -- (S1);
\draw[red, Stealth-Stealth] (S0g3) -- (S1);
\draw[red, Stealth-Stealth] (S0g2) -- (S1);
\draw[red, Stealth-Stealth] (S0g2g3) -- (S1);
\draw[red, Stealth-Stealth] (S2p) -- (S1);
\draw[red, Stealth-Stealth] (S2m) -- (S1);
\draw[red, Stealth-Stealth] (S4) -- (S1);
%
\end{scope}
\end{tikzpicture} 
\ee
where, upon gauging $H_1=\Z_2\times \Z_2$ symmetry:  
\begin{itemize}
    \item $v^1_0$ in \eqref{eqn:VecG32A15phase_omega} is invariant under $H_1$ symmetry, hence it leads to four vacua $(v^1_0)_{1_h}$, $h\in H_1$, each dressed with an irreducible representation $1_{h}$ of $H_1$.
    \item Similarly, $v^1_4$ in \eqref{eqn:VecG32A15phase_omega} is invariant under $H_1$, it is dressed with the non-trivial (2-dimensional) projective representation of $H_1$, hence results in $v^1_{4}$.     \item $v^1_1,v^1_3,v^1_5,v^1_7$ form an $H_1$-orbit, hence combine into $v^1_{1,3,5,7}$. 
    \item $v^1_2$ and $v^1_6$ are swapped by $g_2$, hence combine together in the $H_1$-gauged theory, and get dressed by the two irreducible representations of $\Z_2^{\langle g_3\rangle}$, hence result in $(v^1_{2,6})_{+}$ and $(v^1_{2,6})_{-}$.
\end{itemize}

Similarly, the phase $\cP^{\cS^\omega}_{A_{16}^\omega}$ corresponding to $A_{16}^\omega$ is
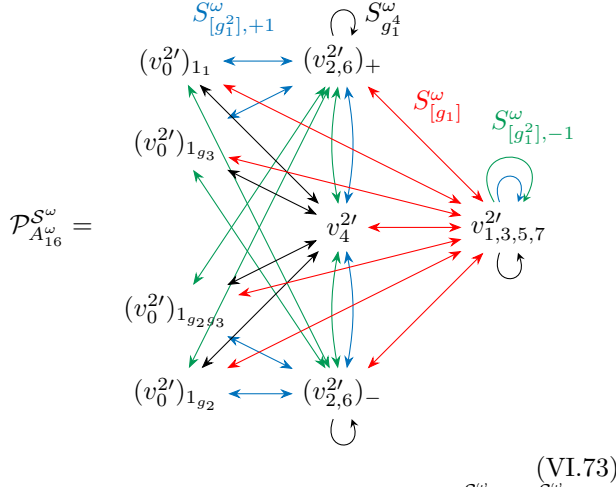
\be
\label{eqn:Som-A16phase}
\begin{tikzpicture}[baseline,yscale=1.1, xscale=1.1]
\node at (1.5,0) {$\cP^{\cS^\omega}_{A^\omega_{16}}=$};
\begin{scope}[shift={(0,0)}]
\node (S0) at (3,2) {$(v^{2\prime}_0)_{1_1}$};
\node (S0g3) at (3,1) {$(v^{2\prime}_0)_{1_{g_3}}$};
\node (S0g2) at (3,-2) {$(v^{2\prime}_0)_{1_{g_2}}$};
\node (S0g2g3) at (3,-1) {$(v^{2\prime}_0)_{1_{g_2g_3}}$};
\node (S2p) at (5,2) {$(v^{2\prime}_{2,6})_{+}$};
\node (S2m) at (5,-2) {$(v^{2\prime}_{2,6})_{-}$};
\node (S4) at (5,0) {$v^{2\prime}_{4}$};
\node (S1) at (7,0) {$v^{2\prime}_{1,3,5,7}$};
%
%
\node[RoyalBlue] at (3.7,2.5) {$S^\omega_{[g_1^2],+1}$};
\draw[RoyalBlue, Stealth-Stealth] (S0g3) -- (S2p);
\draw[RoyalBlue, Stealth-Stealth] (S0) -- (S2p);
\draw[RoyalBlue, Stealth-Stealth] (S0g2g3) -- (S2m);
\draw[RoyalBlue, Stealth-Stealth] (S0g2) -- (S2m);
\draw[RoyalBlue, Stealth-Stealth, bend right=10] (S4) to (S2p);
\draw[RoyalBlue, Stealth-Stealth, bend left=10] (S4) to (S2m);
\draw[RoyalBlue,-stealth] (S1) edge [in=70, out=110,looseness=5] (S1);
%
\node[ForestGreen] at (7.3,1.2) {$S^\omega_{[g_1^2],-1}$};
\draw[ForestGreen, Stealth-Stealth] (S0g3) -- (S2m);
\draw[ForestGreen, Stealth-Stealth] (S0) -- (S2m);
\draw[ForestGreen, Stealth-Stealth] (S0g2g3) -- (S2p);
\draw[ForestGreen, Stealth-Stealth] (S0g2) -- (S2p);
\draw[ForestGreen, Stealth-Stealth, bend left=10] (S4) to (S2p);
\draw[ForestGreen, Stealth-Stealth, bend right=10] (S4) to (S2m);
\draw[ForestGreen,-stealth] (S1) edge [in=60, out=120,looseness=6] (S1);
%
\node[black] at (5.5, 2.5) {$S^\omega_{g_1^4}$};
\draw[black, Stealth-Stealth] (S0) -- (S4);
\draw[black, Stealth-Stealth] (S0g3) -- (S4);
\draw[black, Stealth-Stealth] (S0g2) -- (S4);
\draw[black, Stealth-Stealth] (S0g2g3) -- (S4);
\draw[black,-stealth] (S2p) edge [in=70, out=110,looseness=5] (S2p);
\draw[black,-stealth] (S2m) edge [in=-70, out=-110,looseness=5] (S2m);
\draw[black,-stealth] (S1) edge [in=-70, out=-110,looseness=5] (S1);
%
\node[red] at (6.15,1.45) {$S^\omega_{[g_1]}$};
\draw[red, Stealth-Stealth] (S0) -- (S1);
\draw[red, Stealth-Stealth] (S0g3) -- (S1);
\draw[red, Stealth-Stealth] (S0g2) -- (S1);
\draw[red, Stealth-Stealth] (S0g2g3) -- (S1);
\draw[red, Stealth-Stealth] (S2p) -- (S1);
\draw[red, Stealth-Stealth] (S2m) -- (S1);
\draw[red, Stealth-Stealth] (S4) -- (S1);
%
\end{scope}
\end{tikzpicture} 
\ee

Importantly, the twin gapped phases $\cP^{\cS^\omega}_{A^\omega_{15}}$, $\cP^{\cS^\omega}_{A^\omega_{16}}$ admit the same $\cS^\omega$-symmetry action and the same symmetry breaking pattern, confirming our result that phase transitions between twin phases do not observe hidden symmetry breaking, as discussed in \Cref{sec:no-rel-SSB}.

\vspace{2mm}
\noindent\textbf{Phase Transition between Twin Gapped Phases.} 
The transition theory $\cP^{\cS^\omega}_{A_3^\omega}$ between $\cP^{\cS^\omega}_{A_{15}^\omega}$ and $\cP^{\cS^\omega}_{A_{16}^\omega}$ can similarly be obtained from $\cP^{\Vec^\omega_G}_{A_3^\omega}$~\eqref{eq:G-Aom3-theory} by gauging $H_1$:
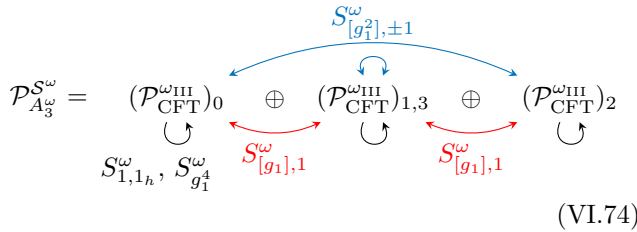
\be \label{eq:Som-A3-theory}
\begin{tikzpicture}[baseline,xscale=1.3]
\begin{scope}[shift={(-1,2.8)}]
\node at (1.2,-3) {$\cP^{\cS^\omega}_{A_3^\omega}=$};
\node (1) at (2.5,-3) {$(\cP^{{\omega_{\text{III}}}}_{\text{CFT}})_{0}$};
\node at (2.3,-4) {$S^\omega_{1,1_h}$, $S^\omega_{g_1^4}$};
\node at (3.5,-3) {$\oplus$};
\node[red] at (3.5,-3.8) {$S^\omega_{[g_1]}$};
\node[red] at (5.5,-3.8) {$S^\omega_{[g_1]}$};
\node (2) at (4.5,-3) {$(\cP^{{\omega_{\text{III}}}}_{\text{CFT}})_{1,3}$};
\node[RoyalBlue] at (4.5,-2) {$S^\omega_{[g_1^2],\pm 1}$};
\node (2b) at (4.55,-3.2) {};
\node at (5.5,-3) {$\oplus$};
\node (3b) at (6,-3.3) {};
\node (3) at (6.5,-3) {$(\cP^{{\omega_{\text{III}}}}_{\text{CFT}})_{2}$};
\node (4b) at (6.55,-3.2) {};
\draw[red,stealth-stealth] (1) [bend right=30] to (2);
\draw[red,stealth-stealth] (2) [bend right=30] to (3);
%
\draw[RoyalBlue,stealth-stealth] (1) [bend left=30] to (3);
\draw[RoyalBlue,stealth-stealth] (2) edge [in=70, out=110,looseness=5] (2);
\draw[black,-stealth] (1) edge [in=-70, out=-110,looseness=6] (1);
\draw[black,-stealth] (2) edge [in=-70, out=-110,looseness=6] (2);
\draw[black,-stealth] (3) edge [in=-70, out=-110,looseness=6] (3);
\end{scope}
\end{tikzpicture} 
\ee
Note that $S^\omega_{[g_1]}$ acts non-invertibly: it maps $(\cP^{{\omega_{\text{III}}}}_{\text{CFT}})_{1,3}$ to $(\cP^{{\omega_{\text{III}}}}_{\text{CFT}})_{0}$ and $(\cP^{{\omega_{\text{III}}}}_{\text{CFT}})_{2}$.
\begin{itemize}
    \item For $\cP^{\cS^\omega}_{A^\omega_{15}}$, adding to $\cP^{{\omega_{\text{III}}}}_{\text{CFT}}$ the local relevant deformation $e_a$ spontaneously breaks $\Z_2^{\langle a \rangle}$. This leads to the gapless phase corresponding to the algebra \nameref{Aom:4}, which has four universes, with an extra idempotent corresponding to $v^1_4$. Next, we condense the anyons $([c],2),([b],2),([bc],2)$ in the reduced TO $\cZ(\Vec_{\Z_2^3}^{\omega_\III})$. Using the map of anyons to $\cZ(\Vec_G^\omega)$ \eqref{app:map-of-anyon-G32omega}, we obtain four new idempotents, since $([c],2),([bc],2)$ are each mapped to an anyon in \nameref{Aom:15}, while $([b],2)$ is mapped to two distinct anyons in \nameref{Aom:15}. We thus obtain the gapped phase $\cP^{\cS^\omega}_{A_{15}^\omega}$ \eqref{eqn:Som-A15phase} with eight vacua.
    
    \item For $\cP^{\cS^\omega}_{A^\omega_{16}}$, adding to $\cP^{{\omega_{\text{III}}}}_{\text{CFT}}$ the local relevant deformation $e_b$ spontaneously breaks $\Z_2^{\langle b \rangle}$. This leads to the gapless phase corresponding to the algebra \nameref{Aom:5}, which has four universes, with an extra idempotent corresponding to $v^1_4$. Next, we condense the anyons $([c],2),([a],2),([ac],2)$ in the reduced TO $\cZ(\Vec_{\Z_2^3}^{\omega_\III})$. Using the map of anyons to $\cZ(\Vec_G^\omega)$ \eqref{app:map-of-anyon-G32omega}, we obtain four new idempotents, since $([c],2),([a],2)$ are each mapped to an anyon in \nameref{Aom:16}, while $([ac],2)$ is mapped to two distinct anyons in \nameref{Aom:16}. We thus obtain the gapped phase $\cP^{\cS^\omega}_{A_{16}^\omega}$ \eqref{eqn:Som-A16phase} with eight vacua.
    \end{itemize}
\begin{equation}\begin{split}
\begin{tikzpicture}
\begin{scope}[shift={(-1,2.8)}, scale =0.8]
\node (H) at (0,0) {$\cP^{\cS^\omega}_{A^\omega_3}$};
\node (H1) at (-1,-1.5) {$\cP^{\cS^\omega}_{A^\omega_{15}}$};
\node[Blue,left] at (-1.35,-0.8) {$\langle A_{15}^{\Z_2^3,\omega_{\text{III}}}\mapsto A_{15}^\omega\rangle\neq 0$};
\node[RoyalBlue,right] at (1.35,-0.8) {$\langle A_{16}^{\Z_2^3,\omega_{\text{III}}}\mapsto A_{16}^\omega\rangle\neq 0$};
\node (H2) at (1,-1.5) {$\cP^{\cS^\omega}_{A^\omega_{16}}$};
\draw[Blue,-stealth] (H) to (H1);
\draw[RoyalBlue,-stealth] (H) to (H2);
\end{scope}
\end{tikzpicture}
\end{split}\end{equation}
The different order/disorder parameters \eqref{eq:tilde_A15}-\eqref{eq:tilde_A16} for the input transition on the $\cZ(\Vec_{{\Z_2^3}}^{\omega_{\text{III}}})$ physical boundary are mapped to the same anyons in the twin Lagrangian algebras \eqref{eqn:map-of-anyon-A3omega}, giving rise to transition \eqref{eq:Som-A3-theory} for twin gapped phases with anomalous non-invertible symmetry $\cS^\omega$.

\section{Twin Algebras and Phases for $G_{\text{Bog}}$}
The smallest group admitting a Bogomolov multiplier is of order 64. We focus on the group $G_{\text{Bog}}$ appearing in \cite{Kobayashi:2025pxs}. This is a central extension of $G_{32}$ studied in \Cref{sec:G32-43_examples}\footnote{In GAP, $G_{\text{Bog}}:=\text{SmallGroup}(64,149)$.}
\begin{equation}
    1\rightarrow \Z_2^{\langle x\rangle} \rightarrow G_{\text{Bog}} \rightarrow G_{32} \rightarrow 1.
\end{equation}
In terms of generator and relations, denote its generators by $x, g_1, g_2, g_3$, with
\begin{equation}
    \begin{split}
        x^2=g_2^2=g_3^2=g_1^8=1, \\
        {}^{g_2}g_1 = (g_1)^3, \quad {}^{g_3}g_1 = (g_1)^5,\\
        [g_2, g_3] = x, \quad [g_2,x] = [g_3,x] = [g_1,x] = 1.
    \end{split}
\end{equation}
Each group element $g\in G_{\text{Bog}}$ can be denoted by
\begin{equation}
    g=g_2^{n} g_3^{m} g_1^{l} x^k
\end{equation}
for $n,m,k = 0,1$ and $l=0, 1, \cdots, 7$.

In terms of the above notation, the Bogomolov multiplier is
\begin{equation}
    \beta(g, g')=(-1)^{ml'}.
\end{equation}
The 2-cocycle $\beta$ represents a nontrivial class of $H^2(G_{\text{Bog}},U(1))\cong\Z_2\times\Z_2$ and is symmetric on commuting pairs. Thus $\cL(G_{\text{Bog}},1)$ and $\cL(G_{\text{Bog}},\beta)$ are distinct twin Lagrangian algebras with the same anyon decomposition $\bigoplus_{[g]}([g],1)$. These twin gapped phases were discussed in \cite{Kobayashi:2025pxs}.

We now turn to the transition between them. Their largest common subalgebras have dimension $16$ and, up to isomorphism, are $A(G_{\text{Bog}},N,1,1)$ for
\be\ba
N&=\langle g_2,g_1^2,x\rangle\cong\Z_2\times D_8\,, \cr
N&=\langle g_1^2,xg_1^6,g_1g_2\rangle\cong\Z_2\times Q_8\,, \cr
N&=\langle xg_1^6,g_3g_2\rangle\cong\Z_4\rtimes\Z_4\,, \cr
N&=\langle xg_1^6,g_3g_1g_2\rangle\cong\Z_4\rtimes\Z_4\,.
\ea\ee
Each $N$ is normal in $G_{\text{Bog}}$, with $K=G_{\text{Bog}}/N\cong\Z_2\times\Z_2$ and reduced TO $\cZ(\Vec_{\Z_2\times\Z_2})$. For each quotient map $p:G_{\text{Bog}}\to K$, the class of $\beta$ is the pullback $p^*\phi$ of the nontrivial class $\phi\in H^2(\Z_2\times\Z_2,U(1))$. Hence, across the interface determined by $A(G_{\text{Bog}},N,1,1)$, $\cL(G_{\text{Bog}},1)$ becomes the Lagrangian algebra $\cL(K,1)$ of the trivial $\Z_2\times\Z_2$ SPT, while $\cL(G_{\text{Bog}},\beta)$ becomes $\cL(K,\phi)$ of the nontrivial $\Z_2\times\Z_2$ SPT. The anyon map from the reduced TO across the SymTFT interface sends both $\cL(K,1)$ and $\cL(K,\phi)$ to $\bigoplus_{[g]}([g],1)$, so the twins are distinguished only by their algebra structure.

These are not the only twins in $\cZ(\Vec_{G_{\text{Bog}}})$. Let $\gamma'$ represent one of the two classes of $H^2(G_{\text{Bog}},U(1))$ other than $1$ and $\beta$. Then $\cL(G_{\text{Bog}},\gamma')$ and $\cL(G_{\text{Bog}},\gamma'\beta)$ form a second pair of twin Lagrangian algebras, whose 2-cocycles again differ by the Bogomolov class $\beta$. There are also non-maximal twins of $\epsilon$-type.

\vspace{2mm}
\noindent\textbf{Phase Transition between the Twin SPTs.} Fixing the symmetry Lagrangian algebra to be $\cL(1,1)$ defines the Dirichlet $\Vec_{G_{\text{Bog}}}$ symmetry boundary. The twin gapped phases $\cP^{\Vec_{G_{\text{Bog}}}}_{\cL(G_{\text{Bog}},1)}$ and $\cP^{\Vec_{G_{\text{Bog}}}}_{\cL(G_{\text{Bog}},\beta)}$ are then the trivial $G_{\text{Bog}}$-SPT and the $G_{\text{Bog}}$-SPT with 2-cocycle $\beta$. Since $\beta$ is trivial on every abelian subgroup, the two SPTs have the same anyon content and are not distinguished by local or ordinary string order parameters. A phase transition between them is obtained by inputting a $\Z_2\times\Z_2$-symmetric CFT, such as $\Ising^2$, on the physical boundary of the reduced TO $\cZ(\Vec_{\Z_2\times\Z_2})$ of $A(G_{\text{Bog}},N,1,1)$, for any of the four $N$ above:
\begin{equation}
\begin{split}
    \begin{tikzpicture}
\begin{scope}[shift={(0,0)}, scale=1.5]
    \draw[SymTFTColores, fill= SymTFTColores, opacity = 0.2]   (0,0) -- (0,2) -- (2,2) -- (2,0) --(0,0);
    \draw[SymTFTColores, fill= SymTFTColores, opacity = 0.1]  (2,0) -- (2,2) -- (4,2) -- (4,0) --(2,0);
    \draw [very thick]  (0,0) -- (0,2);
    \draw [very thick]  (2,0) -- (2,2);
    \draw [very thick]  (4,0) -- (4,2);
    \node[above] at (0,2) {$G_{\text{Bog}}$};
    \node[above] at (2,2) {$\cI_{A(G_{\text{Bog}},N,1,1)}$};
    \node[above] at (4,2) {$\Bphys=\Ising^2$};
    \node at (1,1) {$\cZ(G_{\text{Bog}})$};
    \node at (3,1) {$\cZ(\Z_2\times\Z_2)$};
    \end{scope}
    \end{tikzpicture}
\end{split}\end{equation}
Compactifying the sandwich gives a $G_{\text{Bog}}$-symmetric gapless phase $\cP^{\Vec_{G_{\text{Bog}}}}_{A(G_{\text{Bog}},N,1,1)}$, on which $G_{\text{Bog}}$ acts through the quotient $p:G_{\text{Bog}}\to K\cong\Z_2\times\Z_2$. Deforming $\Ising^2$ by the relevant operators that drive it to the trivial $\Z_2\times\Z_2$ SPT, $\cL(K,1)=\langle m_1,m_2\rangle$, or to the nontrivial one, $\cL(K,\phi)=\langle m_1e_2,m_2e_1\rangle$, gives the two twin phases:
\begin{equation}\begin{split}
\begin{tikzpicture}[baseline]
\begin{scope}[shift={(-1,2.8)}]
\node (H) at (0,0) {$\cP^{\Vec_{G_{\text{Bog}}}}_{A(G_{\text{Bog}},N,1,1)}$};
\node (H1) at (-1.3,-1.5) {$\cP^{\Vec_{G_{\text{Bog}}}}_{\cL(G_{\text{Bog}},1)}$};
\node[Blue,left] at (-0.75,-0.65) {\footnotesize$\langle m_1,m_2\rangle\neq 0$};
\node[RoyalBlue,right] at (0.75,-0.65) {\footnotesize$\langle m_1e_2,m_2e_1\rangle\neq 0$};
\node (H2) at (1.3,-1.5) {$\cP^{\Vec_{G_{\text{Bog}}}}_{\cL(G_{\text{Bog}},\beta)}$};
\draw[Blue,-stealth] (H) to (H1);
\draw[RoyalBlue,-stealth] (H) to (H2);
\end{scope}
\end{tikzpicture}
\end{split}\end{equation}
The string order parameters of the two input $\Z_2\times\Z_2$ SPTs come from different anyons (respectively $m_1,m_2$ and $m_1e_2,m_2e_1$) of $\cZ(\Vec_{\Z_2\times\Z_2})$, but the interface maps both to $\bigoplus_{[g]}([g],1)$. Hence we obtain a twin transition between twin SPT phases of $G_{\text{Bog}}$, which remains beyond-Landau for any change of symmetry boundary.

\section{Conclusions and Outlook}

The three main results in this paper are as follows: first the sharpening of the classification of condensable algebras for Drinfeld centers of group-theoretical fusion categories. Second, the definition of the concept of twin algebras in general, and their characterization in terms of group-theoretical symmetries. Then in combination with the more physics-oriented \cite{Warman:2026gfz} companion paper, we study implications of the twin algebras, both Lagrangian and non-Lagrangian, which have interesting new effects such as no hidden symmetry breaking and intrinsically non-Landau transitions. 

This systematic study of twins is possible due to the fact that we have a very precise and concrete characterization of twin algebras. It would be very interesting to extend this to non-group-theoretical fusion categories, such as the Tambara-Yamagami or metaplectic fusion categories and explore twins in this context. Another interesting extension is to higher-dimensions. In particular bosonic fusion 2-categories are group-theoretical and their algebra classification is very well understood \cite{Bhardwaj:2024qiv, Xu:2024pwd, Bullimore:2024khm,  Bhardwaj:2025piv,  Inamura:2025cum, Decoppet:2024htz}, and has similarities to the group-theoretical case studied here. Studying extensions to higher dimensional DQCPs and twin-phase transitions is certainly an interesting avenue to pursue.

\subsection*{Acknowledgments}
We are grateful to Andrea Antinucci, Thomas Bartsch, Christian Copetti, Andr\'{e} Henriques, Kansei Inamura, Ryohei Kobayashi, Kantaro Ohmori, Sal Pace, Alex Turzillo,  Rui Wen, Yunqin Zheng   for discussions. 
SSN thanks the anonymous referee(s?) who over the years have kept insisting in their reports on asking about the algebra structure of condensable algebras, for initiating this exciting exploration. 
A.W. thanks Xiao-Gang Wen for suggesting the computer algebra system GAP \cite{GAP4}. The code used for these computations is \texttt{AlgebrasTQD} \cite{AlgebrasTQD} which uses some functions from \cite{gruen2021computing,mignard2017moritaequivalencepointedfusion}. Anthropic's Claude Code was used to check this paper and suggest improvements, which were then implemented and verified by the authors. This work is supported by the UKRI Frontier Research Grant, underwriting the ERC Advanced Grant ``Generalized Symmetries in Quantum Field Theory and Quantum Gravity''.

\appendix

\section{Mathematical Preliminary}
\label{append:math-details}

\subsection{Definition of Condensable Algebras}
\label{app:defn-cond-alg}
We assume some familiarity with basic knowledge on tensor categories and algebras, see e.g., \cite{Fuchs:2002cm, Frohlich:2003hm, EGNO}. 
Here we present the following definitions in an algebraic way, see \cite{Fuchs:2002cm} for the same definitions using string diagrams.
\begin{definition}
    An \emph{algebra} $(A, \mu, \iota)$ in a tensor category $\cC$ consists of an object $A$ in $\cC$, multiplication $\mu\in \Hom_{\cC} (A\otimes A, A)$ and unit $\iota\in \Hom_{\cC}(1_{\cC}, A)$ satisfying compatibility conditions. 

    Similarly, a \emph{coalgebra} $(A, \Delta, \eps)$ in a tensor category $\cC$ consists of an object $A$ in $\cC$, comultiplication $\Delta \in \Hom_\cC(A, A\otimes A)$ and counit $\eps\in \Hom_{\cC}(\A, 1_\cC)$ satisfying compatibility conditions.

    A \emph{Frobenius algebra} $A$ is an object $A\in \cC$ that is both an algebra $(A, \mu, \iota)$ and a coalgebra $(A, \Delta, \eps)$ such that
    \begin{equation}
        \begin{split}
            \Delta \circ \mu & = (\id_A \otimes \mu) \circ (\Delta \otimes \id_A) \\
            & = (\mu \otimes \id_A) \circ (\id_A \otimes \Delta).
        \end{split}
    \end{equation}
\end{definition}
In this paper, we only consider pivotal monoidal categories $\cC$: in particular\footnote{We omit the full definition of pivotal category but the notation relevant for our discussion.}, each object $X\in \cC$ has an associated right dual object $X^\vee \in\cC$ with coevaluation morphism:
\begin{equation}
    \text{coev}_X\in \Hom_{\cC}(1, X\otimes X^\vee),
\end{equation}
and a left dual object ${}^\vee \! X=X^\vee$ with coevaluation morphism\footnote{See e.g., \cite{Fuchs:2002cm} for the relation between $\text{coev}_X$ and $\wt{\text{coev}}_X$.}
\begin{equation}
    \wt{\text{coev}}_X \in \Hom_{\cC}(1, X^\vee\otimes X).
\end{equation}

The following definitions can be found in \cite{Fuchs:2002cm, Frohlich:2003hm}:
\begin{definition}
    A Frobenius algebra $(A, \mu, \iota, \Delta, \eps)$ in a (pivotal) tensor category $\cC$ is 
    \begin{itemize}
        \item \emph{symmetric} if 
        \begin{equation}
            \begin{split}
                \bigl( (\eps \circ \mu) \otimes \id_{A^\vee} \bigr) \circ \alpha_{A,A,A^\vee}^{-1} \circ ( \id_{A} \otimes \text{coev}_{A}) \\
                = \bigl( \id_{A^\vee} \otimes (\eps \circ \mu) \bigr) \alpha_{A^{\vee}, A, A} \circ (\wt{\text{coev}}_{A}\otimes \id_A).
            \end{split}
        \end{equation}
        \item \emph{normalized-special} if
        \begin{equation}
            \mu \circ \Delta = \id_A, \quad \eps \circ \iota = \dim(A) \id_{1_\cC}.
        \end{equation}
    \end{itemize}
\end{definition}

\begin{definition}
    An algebra $(A, \mu, \iota)$ in a braided tensor category $\cB$ is
    \begin{itemize}
        \item \emph{connected} if $\dim(\Hom_{\cB}(1_\cB, A))=1$; 
        \item \emph{commutative} if $\mu = \mu \circ c_{A,A}$, where $c$ denotes the braiding in $\cB$.
    \end{itemize}
\end{definition}

Recall the following definition from \cite{Kong:2013aya}
\begin{definition}
    A \emph{condensable algebra} $A$ in a braided tensor category $\cB$ is a connected commutative symmetric normalized-special Frobenius algebra in $\cB$.
\end{definition}

\begin{definition}
    Let $A$ be a condensable algebra in a modular tensor category $\cB$, its local modules form a modular tensor category $\cB^\text{loc}_{A}$, called the \emph{reduced TO} associated to $A$.
\end{definition}
In the main text, where we focus on Drinfeld center of group-theoretical fusion categories,
an explicit equivalence of modular tensor categories $\cZ(\Vec_{H/N}^{\alpha}) \cong \cZ(\Vec_G^\omega)_{A(H, N, \gamma, \epsilon)}^{\text{loc}}$ with $\cZ(\Vec_G^\omega)_{A(H, N, \gamma, \epsilon)}^{\text{loc}}$ expressed as a subcategory of $\cZ(\Vec_G^\omega)$ is called an \emph{interface} associated to the condensable algebra $A(H, N, \gamma, \epsilon)$.

Lagrangian algebras can be defined in terms of condensable algebras with trivial reduced TOs, or equivalently in terms of quantum dimensions \cite{davydov2013witt}.
\begin{definition}
\label{defn:Lagalg}
    A \emph{Lagrangian algebra} $A$ in a braided tensor category $\cB$ is a condensable algebra such that 
    \begin{equation}
        \cB^\text{loc}_{A} \cong \Vect.
    \end{equation}
    Equivalently, a Lagrangian algebra is a condensable algebra of maximal dimension. More precisely, $\dim(A)^2=\dim(\cB)$.
\end{definition}

\subsection{Classification Data of Condensable Algebras in $\cZ(\Vec_G^\omega)$}
This appendix consists of the technical details in the classification of condensable algebras and their reduced TOs in $\cZ(\Vec_G^\omega)$. From Theorem~\ref{thm:DS-cond-alg}, such a condensable algebra can be labeled by $(H, N, \gamma, \epsilon)$, where $H\subset G$ a subgroup, $N\triangleleft H$ a normal subgroup, $\gamma:N\times N \rightarrow U(1)$ such that $d\gamma=\omega|_N$ and $\epsilon:H \times N \rightarrow U(1)$ satisfying the following compatibility conditions:
    \begin{align}
        \frac{\epsilon(h_1h_2,n)}{\epsilon(h_1,{}^{h_2}n)\epsilon(h_2,n)}&=\omega(h_1,h_2|n), \label{eqn:epsilon_relation_transgression}\\
    	\frac{\epsilon(h,n_1n_2)}{\epsilon(h,n_1)\epsilon(h,n_2)} &= \omega(h|n_1,n_2)\frac{\gamma({}^{h}n_1,{}^{h}n_2)}{\gamma(n_1,n_2)},  \label{eqn:epsilon_relation_multiplication}\\
    	\epsilon(n_1,n_2)& = \frac{\gamma(n_1,n_2)}{\gamma({}^{n_1}n_2,n_1)}\,,\label{eps_from_gamma}
    \end{align}
    where $h,h_1,h_2\in H$, $n,n_1,n_2\in N$, ${}^{g}\!f = gfg^{-1}$ for any $g,f\in G$, and for any $f,g,h\in G$ we denote 
    \begin{align} \label{eqn:alpha_proj_action}
    	\omega(f,g|h) & =\frac{\omega(f,{}^{g}h,g)}{\omega(f,g,h)\omega({}^{fg}h,f,g)},\\
        \Omega_f(g,h)=\omega(f|g,h)& =\frac{\omega(f,g,h)\omega({}^{f}\!g,{}^{f}h,f)}{\omega({}^{f}\!g,f,h)}, 
    \end{align} 
    where $\omega(f,g|h)$ is the transgression.

One can note that from \eqref{eps_from_gamma}, $\eps|_{N\times N}$ is uniquely determined from $\gamma$, while \eqref{eqn:epsilon_relation_multiplication} implies that, for any $h\in H$, $\eps(h,\cdot):N\to U(1)$ is a 1-dimensional projective representation of $N$.

\subsection{Derivation of Twist \eqref{eqn:twist_reduced_TO} of the Reduced TO}
\label{append:twist_reduced_TO}
Here, we follow the exact same steps as in \cite[Theorem~3.13]{davydov2017lagrangian} to derive formula \eqref{eqn:twist_reduced_TO}.

By the proof of \cite[Theorem~3.16]{davydov2017lagrangian}, the local $A(H, N, \gamma, \eps)$-modules in $\cZ(\Vec_G^\omega)$ are the local $B$-modules in $\cZ(\Vec_H^\omega)$, where $B=\bC[N, \gamma, \eps]$.
Recall from \cite{davydov2017lagrangian} that an object in $\cZ(\Vec_H^\omega)$ is an $H$-graded vector space $V=\bigoplus_{h}V_h$. For $v\in V_x$, $u\in U_y$, $w\in W_z$, the associator and braiding are
\begin{align}
        a(v \otimes (u\otimes w)) & = \omega(x,y,z) (v\otimes u) \otimes w, \label{eqn:H_associator}\\
        c_{V,U}(v\otimes u) & = (x\cdot u) \otimes v. \label{eqn:H_braiding}
    \end{align}
Here $h\cdot v$ is the $\omega$-projective $H$-action, satisfying
\be\ba
(h_1h_2)\cdot v&=\omega(h_1,h_2\mid x)\,h_1\cdot(h_2\cdot v)\,, \cr
h\cdot(v\otimes u)&=\Omega_h(x,y)\,(h\cdot v)\otimes(h\cdot u)\,.
\ea\ee
Commutative algebra $B$ has basis $e_n$, $n\in N$ with multiplication and $H$-action
\begin{equation}
    e_n e_{n'}=\gamma(n,n')e_{nn'}, \quad h\cdot e_{n}= \eps(h,n) e_{{}^hn}.
\end{equation}

A right $B$-module $M$ has action $\rho(m\otimes e) = m\lhd e$ which satisfies the following compatibility condition
\begin{equation}
\label{eqn:B-local_module_action}
    (m\lhd e_n) \lhd e_{n'} = \frac{\gamma(n,n')}{\omega(x,n,n')}m \lhd e_{nn'}
\end{equation}
where $m\in M_{x}$.
The module is local if $\rho \circ c_{B,M} \circ c_{M,B} = \rho$, which reads
\begin{equation}
\label{eqn:B_mod_locality}
    m\lhd e_n = \eps(x,n) ({}^xn \cdot m) \lhd e_{{}^xn}.
\end{equation}
Recall the monoidal structure on local $B$-modules.
Given two local $B$-modules $M$ and $L$, their tensor product $M\otimes_{B}L$ satisfies that for $m\in M_x$, $v\in L_{y}$, 
\begin{equation}
\label{eqn:MotimesN-coequalizer}
    (m\lhd e_n) \otimes v = \frac{1}{\omega(x,n,y)} m\otimes ((n\cdot v) \lhd e_n).
\end{equation}
Combining with \eqref{eqn:B_mod_locality}, we obtain
\begin{equation}
    \label{eqn:MotimesN-coequalizer_explicit}
    (m \lhd e_n) \otimes v = \frac{1}{\omega(x,n,y) \eps(y, n^y)} m\otimes (v\lhd e_{n^y}).
\end{equation}
Moreover, $M\otimes_B L$ has module structure
\begin{equation}
\label{eqn:MotimesN_module}
    (m\otimes v) \lhd e_n = \frac{1}{\omega(x,y,n)} m \otimes (v\lhd e_n).
\end{equation}

We are now ready to derive formula \eqref{eqn:twist_reduced_TO}. For $x\in H/N$ and $s:H/N \rightarrow H$ a section, set $V_x=M_{s(x)}$ and let $U_z=L_{s(z)}$ be the $s(z)$-graded component of the second module. As in the proof of \cite[Theorem~3.13]{davydov2017lagrangian}, for $k\in H/N$ the $s(k)$-graded component of $M\otimes_B L$ is identified with $\bigoplus_{yz=k} V_y \otimes U_z$ via
\begin{equation}
    \Phi(m\otimes v) = m \otimes (v \lhd e_{\tau(y,z)})
\end{equation}
where $m\in V_y$ and $v\in U_z$. The associator $\alpha$ has component $\alpha(x,y,z)$ such that for $m\in M_{s(x)}$, $v\in L_{s(y)}$ and $w \in L_{s(z)}$, 
\begin{equation}
    a\big(\Phi(m\otimes \Phi( v\otimes w)) \big) = \alpha(x,y,z) \Phi\big( \Phi(m\otimes v) \otimes w \big).
\end{equation}
Applying \eqref{eqn:MotimesN_module} and \eqref{eqn:B-local_module_action} to the LHS, and introduce $\tau_1=\tau(y,z)$, $\tau_2=\tau(x, yz)$, $u=s(y)^{-1} s(x)^{-1} s(xyz)$ we obtain
\begin{equation}
    \begin{split}
         a\big(\Phi(m\otimes \Phi( v\otimes w)) \big) = c_L (m\otimes v)\otimes (w\lhd e_{\tau_1\tau_2}), \\
    c_L=\frac{\gamma(\tau_1, \tau_2) \omega(s(x), s(y), u)}{\omega(s(y), s(y)^{-1}s(yz), \tau_2) \omega(s(z), \tau_1, \tau_2)}.
    \end{split}
\end{equation}
Introduce shorthand $\tau_3=\tau(x,y)$, $\tau_4=\tau(xy,z)$, equations \eqref{eqn:MotimesN_module}, \eqref{eqn:MotimesN-coequalizer_explicit} and \eqref{eqn:B-local_module_action} give
\begin{equation}
    \begin{split}
        & \hspace{-2.7em}\Phi(\Phi(m\otimes v) \otimes w) = c_R (m\otimes v) \otimes (w \lhd e_{\tau_1 \tau_2}),\\
        c_R={}& \frac{\omega(s(x), s(y), \tau_3) }{\omega(s(x) s(y), \tau_3, s(z) \tau_4) \eps\big(s(z) \tau_4, \tau_3^{s(z) \tau_4}\big) }\\
        & \times \frac{\gamma(\tau_4, \tau_3^{s(z) \tau_4})}{\omega(s(z), \tau_4, \tau_3^{s(z) \tau_4}) }.
    \end{split}
\end{equation}
Combining these two expressions, we obtain
\begin{equation}
    \alpha(x,y,z) = \frac{c_L}{c_R},
\end{equation}
which is \eqref{eqn:twist_reduced_TO}.

\section{Proofs of Algebra Properties}

\subsection{Isomorphic Condensable Algebras}
\label{append:alg_Morita_equiv}
Here we prove Lemma~\ref{lem:Morita_equiv}. 
This is achieved via Morita equivalence relation on algebras\footnote{Not to be confused with the categorical Morita equivalence, which also plays an important role in the proof, see \cite{EGNO}.}, which is equivalent to isomorphism relation for condensable algebras.


In \Cref{sec:monoidal_structure_conjugation}, we give details on conjugation action $\ad_g$ on $\Vect_G^\omega$ following \cite{Natale2017}. In \Cref{sec:induced_functor_on_Morita_dual}, we induce braided autoequivalence $\wt{\ad}_{g_1,g_2}$ on $\cZ(\Vec_G^\omega)$ from the above conjugation functor. Finally in \Cref{sec:Morita_Equiv_Cond_algs}, we show that the induced functors $\wt{\ad}_{g_1,g_2}$ identify the Morita equivalent, hence isomorphic, condensable algebras in $\cZ(\Vec_G^\omega)$. 

\subsubsection{Preliminary: Conjugation Action on $\Vec_G^\omega$}
\label{sec:monoidal_structure_conjugation}
This section consists of the conjugation action on $\Vec_G^\omega$, which identifies Morita equivalent algebras in $\Vec_G^\omega$ \cite{Natale2017}. 

For $g\in G$, the functor
\be
\label{eqn:conjugation_Vec_G}
    \ad_g: \Vec_G^{\omega} \rightarrow \Vec_G^\omega
\ee
sends an object $V$ of $\Vec_G^{\omega}$ to ${}^{g}V$, such that ${}^{g}V=V$ as a vector space, with $G$-grading $({}^gV)_{{}^{g}x} = V_{x}$
for $x\in G$, and maps a morphism $f$ of $\Vec_G^\omega$ to $f$. This defines a monoidal functor with monoidal structure
\begin{align}
    (\ad_g^2)_{U,V}: {}^{g}U \ot {}^{g}V & \rightarrow {}^{g}(U \ot V)\nn\\
    u \ot v & \mapsto \Omega_g(x,y)^{-1} u \ot v,
\end{align}
where $u\in U_x$, $v\in V_y$ and $\Omega$ is defined in \eqref{eq:Omega_f}. For notational simplicity, we often write 
\be
    (\ad_g^2)_{x,y} := (\ad_g^2)_{U,V}
\ee
if the objects $U$ and $V$ are simple and graded by $x\in G$ and $y\in G$, respectively.

For any subgroup $H\subset G$ and a 2-cochain $\gamma\in C^2(H, U(1))$ such that $d\gamma = \omega|_H$, $A(H,\gamma) = \mathbb{C}_{\gamma}[H]$ denotes the (Frobenius) group algebra with multiplication twisted by $\gamma$, considered as an algebra in $\Vec_G^\omega$ (any semisimple indecomposable algebra $A$ in $\Vec_G^\omega$ such that $\dim\lb \Hom_{\Vec_G^\omega}(1,A) \rb=1$ is of this form).
The following result was shown in \cite{Natale2017} 
\begin{theorem}[\cite{Natale2017}, Thm. 1.2]
\label{thm:Natale_MorEquivAlg}
    Algebras $A(H_1,\gamma_1)$ and $A(H_2,\gamma_2)$ in $\Vec_G^\omega$ are Morita equivalent if and only if they are related by $\ad_g$,
\be
    \ad_g(A(H_1,\gamma_1))\cong A(H_2,\gamma_2),
\ee
for some $g\in G$.
\end{theorem}
We denote the $\Vec_G^\omega$-module category associated to $A(H,\gamma)$ by $\cM(H,\gamma)$ (i.e., $\cM(H, \gamma)$ is the category of $A(H, \gamma)$-modules in $\Vec_G^\omega$).



\subsubsection{Induced Functors on Morita Duals}
\label{sec:induced_functor_on_Morita_dual}
We now describe the induced tensor functor on a Morita dual category $\cC_\cM^*$, from a tensor functor on $\cC$ that preserves the Morita classes of algebras (the precise meaning will be explained in what follows). Such induced functors will be used to exhaust the Morita equivalence relations on the condensable algebras in $\cZ(\Vec_G^\omega)$ in the next section. In what follows, all $\cC$-module categories are assumed to be indecomposable.

For any fusion category $\cC$ and an indecomposable $\cC$-module category $\cM$, let $A_\cM$ denote the algebra in $\cC$ such that $\cM=\Mod_\cC(A_\cM)$ and the Morita dual $\cC_\cM^*$ is 
\be
    \cC_\cM^* :=\End_\cC(\cM) \cong \Bimod_\cC(A_\cM).
\ee
Let $F: \cC \rightarrow \cC$ be a tensor autoequivalence that preserves the Morita equivalence class, i.e., 
\be
\label{eqn:Morita_class_preserving}
    \Mod_\cC(F(A)) \cong \Mod_\cC(A)
\ee 
as $\cC$-module categories for any algebra $A$ in $\cC$. 
It induces a tensor endofunctor on the Morita dual category $\cC_{\cM}^*$
\begin{equation}
    \begin{split}
        \wt{F}: \Bimod_{\cC}(\cA_{\cM}) & \rightarrow \Bimod_{\cC}(F(\cA_{\cM})), \\
        (N,p,q) & \mapsto (F(N), \wt{p}, \wt{q})
    \end{split}
\end{equation}
where $(N,p,q)$ is the $A_\cM$-bimodule consists of $N \in \cC$, left $A_\cM$-action $p:A_\cM \ot N \rightarrow N$ and right action $q:N \otimes A_\cM\rightarrow N$, to the $F(\cA_{\cM})$-bimodule $(F(N), \wt{p}, \wt{q})$ with left action
\be
    \wt{p} = F(p) \circ F^2_{A_\cM, N}
\ee
and right action
\be
    \wt{q} = F(q) \circ F^2_{N,A_\cM},
\ee
where $F^2$ denotes the tensor structure of $F$. The induced functor $\wt{F}$ preserves the Morita equivalent classes of algebras.

Equivalently, using the description of the Morita dual in terms of $\cC$-module endofunctors $\cC_\cM^* = \End_{\cC}(\cM)$, the induced functor can be described by composing with $\cC$-module equivalences. More explicitly, by assumption (see equation~\eqref{eqn:Morita_class_preserving}), there are $\cC$-module equivalences
\begin{align}
    \phi& : F_*(\cM) \rightarrow \cM,\\
    \psi& : \cM \rightarrow F_*(\cM),
\end{align}
where $F_*(\cM)$ denotes the $\cC$-module structure on $\cM$. The induced functor can be expressed as
\begin{equation}
    \begin{split}
        \wt{F}: \End_\cC(\cM) & \rightarrow \End_\cC(\cM)\\
    \Phi & \mapsto \phi \circ F_*(\Phi) \circ \psi. \label{eqn:induced_functor_module_endof}
    \end{split}
\end{equation}

Moreover, if $F$ a tensor autoequivalence on $\cC \boxtimes \cC^{\text{op}}$, for some fusion category $\cC$, i.e., $F: \cC \boxtimes \cC^{\text{op}} \rightarrow \cC \boxtimes \cC^{\text{op}}$. $\cC$ is a $\cC \boxtimes \cC^{\text{op}}$-module category and
the induced functor on $\cZ(\cC)=(\cC \boxtimes \cC^{\text{op}})_{\cC}^*$,
\be
    \wt{F}: \cZ(\cC) \rightarrow \cZ(\cC)
\ee
is a braided autoequivalence.

From the module endofunctor perspective, the objects in $\cZ(\cC)$ are generated by tensoring with objects in $\cC$ (i.e., the left and right $\cC$-module actions), and the half-braidings in $\cZ(\cC)$ are the same as the commutation relation between the (both left and right) $\cC$-module actions -- since the functors in \eqref{eqn:induced_functor_module_endof} are all $(\cC\boxtimes \cC^{\text{op}})$-module functors, $\wt{F}$ is braided. Note that different monoidal autoequivalences $F$ can induce the same $\wt{F}$.

This will be applied in the next section to identify Morita equivalent algebras in $\cZ(\Vec_G^\omega)$.


\subsubsection{Morita Equivalent/Isomorphic Classes of Condensable Algebras}
\label{sec:Morita_Equiv_Cond_algs}
We now prove which condensable algebras are Morita equivalent in $\cZ(\Vec_G^\omega)$. Here, ``an algebra $A$ in $\Vec_G^\omega$'' refers to a semisimple connected algebra $A$, where connected means that $\dim\lb \Hom_{\Vec_G^\omega}(1,A) \rb=1$. Being connected in $\cC$ guarantees its category of modules in $\cC$ is an indecomposable $\cC$-module category \cite{davydov2013witt}.

From \cite{Ostrikmodule}, $\cZ(\Vec_G^\omega)$ and $\Vec_{G\times G}^{\omega_1/\omega_2}$ are Morita equivalent with $\Bimod_{\Vec_{G\times G}^{\omega_1/\omega_2}}(A(G^{\diag},1))\cong \cZ(\Vec_G^\omega)$, where 
for $i=1,2$, $\omega_i=p_i^*\omega$ denotes the pullback class along 
\begin{align}
    p_i:G \times G & \rightarrow G\nn\\
    (g_1,g_2) & \mapsto g_i,
\end{align}
projection onto the $i$-th copy of $G$, $G^\diag$ is the image of 
\begin{equation}
    \begin{split}
        \Delta^\diag:G & \rightarrow G\times G\\
        g & \mapsto (g, g),
    \end{split}
\end{equation}
and $A(G^\diag,1)$ is the algebra
\begin{equation}
    A(G^\diag,1) \cong \bigoplus_{g\in G^\diag}\bC_g,
\end{equation}
with multiplication follows from group multiplication.

We need an explicit equivalence of categories between $(\Vec_{G\times G}^{\omega_1/\omega_2})_{\cM(G^{\diag}, 1)}^*$ and $\cZ(\Vec_G^\omega)$. 
Such equivalences of categories are not unique, but the resulting Morita equivalence relation on condensable algebras does not depend on this choice. We fix the equivalence where, for $x\in G$, the $x$-graded component of an object in $\cZ(\Vec_G^\omega)$ is identified with the $(x,1)$-graded component of an object in $(\Vec_{G\times G}^{\omega_1/\omega_2})_{\cM(G^{\diag}, 1)}^*=\Bimod_{\Vec_{G\times G}^{\omega_1/\omega_2}}(A(G^{\diag},1))$, and the $\omega$-projective action of $g\in G$ in $\cZ(\Vec_G^\omega)$ is the bimodule action from the $(g,g)$-graded component of $A(G^{\diag},1)$ in $\Vec_{G\times G}^{\omega_1/\omega_2}$:

As objects in $\cZ(\Vec_G^\omega)$, the $g\in G$ action on the $x$-graded component comes from the $A(G^{\diag},1)$-bimodule structure as
\be
    (q|_{({}^gx, 1)\otimes (g,g)})^{-1} \circ p|_{(g,g)\otimes (x,1)},
\ee
where $p$ denotes the left $A(G^{\diag},1)$-module structure and $q$ the right $A(G^{\diag},1)$-module structure\footnote{One can also think of this action as given by the braiding.}. 

Under this equivalence of categories, the induced functor $\wt{\ad}_{(g,g')}$ changes the $G$-grading on an object via conjugation by $g$, and changes the $G$-action, derived from the $A(G^{\diag}, 1)$-bimodule structure, to 
\begin{align}
    {}^gx \,\wt{\cdot}\, v_{{}^gy} & = \frac{(\ad_{(g,g')}^2)_{(x,x),(y,1)}}{(\ad_{(g,g')}^2)_{({}^xy,1),(x,x)}} {}^g(x \cdot v_{y}) \nn\\
    & = \frac{\Omega_g({}^x y, x)}{\Omega_g(x,y)} {}^g(x \cdot v_{y}),
\end{align}
where $x,y\in G$, $\wt{\cdot}$ denotes the $G$-action after applying the functor $\wt{\ad}_{(g,g')}$, and $\cdot$ denotes the $G$-action before applying the functor, ${}^g(x \cdot v_{y})$ meaning conjugating the grading of $x \cdot v_{y}$ by $g$. The ratio of the $\Omega_g$ terms comes from the coherence isomorphisms of the monoidal functor $\ad_{(g,g')}$. (Note that a different choice of the above mentioned equivalence of categories can result in conjugation by $g'$, and the $G$-actions modify accordingly.)

Braided functors send condensable algebras to condensable algebras. Now we can apply the induced braided functor $\wt{\ad}_{(g,g')}$ to the condensable algebra $A(H,N,\gamma,\epsilon)$ in $\cZ(\Vec_G^\omega)$ to derive 
\be
    \wt{\ad}_{(g,g')}(A(H,N,\gamma,\epsilon)) = A({}^gH, {}^g N, \gamma',\epsilon'),
\ee
where, to determine $\gamma'$ and $\epsilon'$, we focus on basis elements $v_{1,{}^gn}$ and $v_{1,{}^gm}$, for some $m,n\in N$. $\gamma'$ can be deduced from the multiplication of $v_{1,{}^g m}, v_{1,{}^g n}\in A({}^gH, {}^g N, \gamma',\epsilon')$, where $\wt{\ad}_{(g,g')}$ modifies the multiplication by $(\ad_{(g,g')}^2)_{(m,1),(n,1)}$, i.e., 
\be
    v_{1,{}^g m} v_{1,{}^g n} = \gamma(m,n)\Omega_g(m,n)^{-1} v_{1,{}^g(mn)}.
\ee
Similarly, $\epsilon'$ can be computed from the ${}^g\!H$-action on $v_{1,{}^gn}$: for any $h\in H$, 
\begin{align}
    {}^gh \,\wt{\cdot}\, v_{1,{}^gn} & = \frac{\Omega_g({}^h n, h)}{\Omega_g(h,n)}{}^g(h\cdot v_{1,n}) \nn\\
    & = \frac{\Omega_g({}^h n, h)}{\Omega_g(h,n)} \epsilon(h,n) v_{1,{}^g(hnh^{-1})}.
\end{align}

Theorem~\ref{thm:Natale_MorEquivAlg} shows that two algebras in $\Vec_{G\times G}^{\omega_1/\omega_2}$ are Morita equivalent if and only if they are related by the conjugation action $\ad_{(g,g')}$~\eqref{eqn:conjugation_Vec_G} for some element $(g,g') \in G \times G$. Hence two condensable algebras $A(H, N, \gamma, \epsilon)$ and $A(H',N', \gamma', \epsilon')$ in $\cZ(\Vec_G^\omega)$ are Morita equivalent if and only if they are related by the induced functor $\wt{\ad}_{(g,g')}$. The backward direction is clear, and the forward direction can be shown by a contrapositive argument: if two such algebras are not related by any of the induced functor, then their equivalent counterpart as algebras in $\Vec_{G\times G}^{\omega_1/\omega_2}$ are not conjugate to each other via any $\ad_{(g,g')}$, which cannot be Morita equivalent by Theorem~\ref{thm:Natale_MorEquivAlg}, hence the original algebras in $\cZ(\Vec_G^\omega)$ are not Morita equivalent. 

Hence we conclude that two condensable algebras $A(H, N, \gamma, \epsilon)$ and $A(H',N', \gamma', \epsilon')$
are Morita equivalent in $\cZ(\Vec_{G}^\omega)$ if and only if 
\begin{align}
    H' & = {}^g\!H, & \gamma'({}^g n_1, {}^g n_2) & = \frac{\gamma(n_1,n_2)}{\Omega_g(n_1,n_2)},\nn\\
    N' & = {}^g\!N, & \epsilon'({}^{g}h, {}^{g}n) & =\frac{\Omega_g({}^{h}n,h)}{\Omega_g(h,n)}\epsilon(h,n),
\end{align}
for some $g\in G$.

Finally, condensable algebras are commutative by definition, and two commutative algebras are Morita equivalent if and only if they are isomorphic \cite{Kong:2007yv, DAVYDOV2010319}. Hence the statement follows.

\subsection{Partial Order on Condensable Algebras}
\label{append:partialorder}
The proof of Proposition~\ref{prop:partial-order} on the partial ordering of condensable algebras uses tools developed in Ref.~\cite{davydov2017lagrangian}. Rather than reproducing such constructions in full, we refer readers to that work and provide only the necessary elements for our proof.

Recall the definition of the condensable algebra $A(H,N,\gamma, \epsilon)$ in $\cZ(\Vect_G^{\omega})$ from \cite{davydov2017lagrangian}. $A(H,N,\gamma, \epsilon)$ is defined by applying a functor denoted by $E_{H}$\footnote{Here we use a slightly different notation from \cite{davydov2017lagrangian} by labeling the functor with the subgroup $H$ to emphasize the source category of the functor.} to the algebra $\bC[N,\gamma,\epsilon]$ in $\cZ(H,\omega|_H)$, i.e.,
\begin{equation}
\label{eqn:induce-cond-alg}
    A(H,N,\gamma, \epsilon):=E_H(\bC[N,\gamma,\epsilon]),
\end{equation}
where $\bC[N,\gamma,\epsilon]$ as an algebra in $\cZ(H,\omega|_H)$ has basis elements $e_n$ for $n\in N$, graded as $|e_n|=n$, with $H$-action $h\cdot e_n=\epsilon(h,n)e_{hnh^{-1}}$ for $h\in H$ and multiplication defined by $e_{n_1} e_{n_2} = \gamma(n_1,n_2) e_{n_1 n_2}$, where $n_1, n_2\in N$. 

The functor $E_H$ in general sends a $G$-graded vector space $V$ with $\omega$-projective $H$-action to the following subspace of $\text{Map}(G,V)$:
\begin{align} \label{eqn:defn_transfer_functor}
    E_H(V) = \{ a: G \rightarrow V\,|\, & a(xh)=\omega(h^{-1},x^{-1}|n) h^{-1}\cdot a(x), \nn\\
    & h\in H, x\in G, |a|=n \}.
\end{align}
$E_H(V)$ is an object of $\cZ(\Vect_G^{\omega})$, with the $G$-grading on $\text{Map}(G,V)$ defined by $|a|=g\in G$ if and only if $|a(x)|=x^{-1}gx$ for all $x\in G$, and the $\omega$-projective $G$-action on $\text{Map}(G,V)$ is defined by $(g\cdot a)(x)=\omega(x^{-1},g|f)^{-1}a(g^{-1}x)$, for a homogeneous $a\in \text{Map}(G,V)$ graded by $|a|=f\in G$. 

To prove the first ordering on the condensable algebras, note that $\bC[N',\gamma|_{N'\times N'}, \epsilon|_{H\times N'}]$ is a subalgebra of $\bC[N,\gamma,\epsilon]$ in $\cZ(\Vect_{H}^{\omega|_{H}})$ via the inclusion
\begin{align}
    \iota: \bC[N',\gamma|_{N'\times N'}, \epsilon|_{H\times N'}] & \hookrightarrow \bC[N,\gamma,\epsilon] \nn\\
    e_n & \mapsto e_n,
\end{align}
where $n\in N'$. Applying the functor $E_H$ results in the following inclusion 
\begin{align}
    E_H(\bC[N',\gamma|_{N'\times N'}, \epsilon|_{H\times N'}]) & \hookrightarrow E_H(\bC[N,\gamma,\epsilon]) \nn\\
    a & \mapsto \iota \circ a,
\end{align}
which can be checked to be an algebra homomorphism in $\cZ(\Vect_G^{\omega})$.

For $H'\subset H$, forgetting the action from element $h\in H$ such that $h\not\in H'$, the two algebras $\bC[N,\gamma,\epsilon]$ and $\bC[N,\gamma,\epsilon|_{H'\times N}]$ can be identified as algebras in $\cZ(\Vect_{H'}^{\omega|_{H'}})$. Definition \eqref{eqn:defn_transfer_functor} of the functors $E_H$ and $E_{H'}$ guarantees that an element $a\in E_{H}(\bC[N,\gamma,\epsilon])$ can be viewed as an element of $E_{H'}(\bC[N,\gamma,\epsilon|_{H'\times N}])$, hence there is a natural inclusion 
\begin{align}
    E_{H}(\bC[N,\gamma,\epsilon]) & \hookrightarrow E_{H'}(\bC[N,\gamma,\epsilon|_{H'\times N}])\nn\\
    a & \mapsto a.
\end{align}
It is now straightforward to check that this inclusion preserves both the $G$-grading and the $G$-action as defined above. Furthermore, it is an algebra homomorphism as the algebra structures on $E_{H'}(\bC[N,\gamma,\epsilon|_{H'\times N}])$ and $E_{H}(\bC[N,\gamma,\epsilon])$ are induced from those on $\bC[N,\gamma,\epsilon|_{H'\times N}]$ and $\bC[N,\gamma,\epsilon]$ respectively, whose multiplications are both defined by the same 2-cocycle $\gamma$.

\section{Distinct Maps of Anyons from Twins}
\label{app:all-maps-of-anyons}

\subsection{$G=(\Z_2\times \Z_2)\ltimes \Z_8$}

Here, we fix $G = (\Z_2 \times \Z_2) \ltimes \Z_8 $ we collect all the pairs of distinct maps of anyons defined by the folded Lagrangian for each twin algebras in $\cZ(\Vec_{G})$, as studied in \Cref{sec:G32-43_examples}. Note that although twin algebras are defined on the same anyons, their corresponding maps of anyons are distinct even on the object level.

\vspace{2mm}
\noindent\textbf{First Pair.}
For twin algebras
\begin{align}
    A_1^{\langle g_1\rangle} & =A(G,\langle g_1\rangle,1,1), \\
    A_2^{\langle g_1\rangle} & =A(G,\langle g_1\rangle,1,\epsilon^{\langle g_1\rangle}_2),
\end{align}
in $\cZ(\Vec_{(\Z_2 \times \Z_2) \ltimes \Z_8 })$.

The folded Lagrangians for $A_i^{\langle g_1\rangle}$, determines the map of anyons from the reduced TO $\cZ(\Vec_{\Z_2\times\Z_2})$. Recalling the definition \eqref{eq:Hdiag} and denoting the $\Z_2\times\Z_2$ group generators as $m_1,m_2$, the subgroups for both folded Lagrangians can be taken to be:
\be\ba
    G^\diag_1=G^\diag_2&=\langle (g_1,1), (g_2g_3,m_1), (g_3,m_2) \rangle\,.
\ea\ee 
However, while the folded Lagrangian for $A_1^{\langle g_1\rangle}$
has trivial 2-cocycle, for  $A_2^{\langle g_1\rangle}$ we must take a non-trivial 2-cocycle since $A_2^{\langle g_1\rangle}$ in order for the subalgebra condition for \eqref{eq:alg_to_include} to hold (recall that $\eps_2^{\langle g_1\rangle}$ is non-trivial).
For $A_1^{\langle g_1\rangle}$, the map of anyons is:
\be\ba 
\cZ(\Vec_{\Z_2^2}) & \rightarrow \cZ(\Vec_G) \cr
1  &\mapsto  ([1],1) \oplus ([g_1^4],1) \oplus ([g_1^2],1) \oplus ([g_1],1)\cr
e_1  &\mapsto  ([1],r_3) \oplus ([g_1^2],1) \oplus ([g_1^4],r_3) \oplus ([g_1],1)\cr
e_2  &\mapsto  ([1],r_4) \oplus ([g_1^2],g_3=-1) \oplus ([g_1^4], r_4) \oplus ([g_1],1)\cr
m_1  &\mapsto   ([g_2g_3],1) \oplus ([g_1g_2g_3],1)\cr
m_2  &\mapsto  ([g_3],1) \oplus ([g_3g_1^6],1) \oplus ([g_1g_3],1)\,,
\ea\ee
while the one for  $A_2^{\langle g_1\rangle} $ gives:
\be\ba 
\cZ(\Vec_{\Z_2^2}) & \rightarrow \cZ(\Vec_G) \cr
1  &\mapsto  ([1],1) \oplus ([g_1^4],1) \oplus ([g_1^2],1) \oplus ([g_1],1)\cr
e_1  &\mapsto  ([1],r_3) \oplus ([g_1^2],1) \oplus ([g_1^4],r_3) \oplus ([g_1],1)\cr
e_2  &\mapsto  ([1],r_4) \oplus ([g_1^2],g_3=-1) \oplus ([g_1^4], r_4) \oplus ([g_1],1)\cr
m_1  &\mapsto  ([g_2g_3],g_3=-1) \oplus ([g_1g_2g_3],1)\cr
m_2  &\mapsto  ([g_3],g_2g_3=-1) \oplus ([g_3g_1^6],1) \oplus ([g_1g_3],1)\,.
\ea\ee
the electric anyons map in the same way, while the map of magnetic anyons is different, as determined from the non-trivial 2-cocycle in the data for the folded Lagrangian data for $A_2^{\langle g_1\rangle}$.

\vspace{2mm}
\noindent\textbf{Second Pair.}
Here we consider the other pair of twin algebras distinguished by generalized string order parameters
\begin{align}
    A_1^{\langle g_1 g_3 \rangle} & = A(G, \langle g_1 g_3 \rangle, 1, 1), \\ 
    A_2^{\langle g_1 g_3 \rangle} & = A(G, \langle g_1 g_3 \rangle, 1, \epsilon^{\langle g_1 g_3 \rangle}_2),
\end{align}
in the same TO as above. 

Similarly to the previous case, we can compute the folded Lagrangians and hence maps of anyons: they are now specified by
\be\ba
    G^\diag_1=G^\diag_2&=\langle (g_1g_3,1), (g_2g_3,m_1), (g_3,m_2) \rangle\,,
\ea\ee 
with trivial 2-cocycle for $A_1^{\langle g_1g_3\rangle}$ and non-trivial one for $A_2^{\langle g_1g_3\rangle}$.
For $A_1^{\langle g_1g_3\rangle}$, the map of anyons is:
\be\ba 
\cZ(\Vec_{\Z_2^2}) & \rightarrow \cZ(\Vec_G) \cr
1  &\mapsto  ([1],1) \oplus ([g_1^4],1) \oplus ([g_1^2], 1) \oplus ([g_1g_3],1)\cr
e_1  &\mapsto  ([1],r_3) \oplus ([g_1^2],1) \oplus ([g_1^4],r_3) \oplus ([g_1g_3],1)\cr
e_2  &\mapsto  ([1],r_5) \oplus ([g_1^2],g_3=-1\,g_1=-1) \cr
&\quad\; \oplus ([g_1^4], r_5) \oplus ([g_1g_3],1)\cr
m_1  &\mapsto   ([g_2g_3],1) \oplus ([g_1g_2],1)\cr
m_2  &\mapsto  ([g_3],1) \oplus ([g_1],1) \oplus ([g_3g_1^6],1) \,,
\ea\ee
while the one for  $A_2^{\langle g_1g_3\rangle} $ gives:
\be\ba 
\cZ(\Vec_{\Z_2^2}) & \rightarrow \cZ(\Vec_G) \cr
1  &\mapsto  ([1],1) \oplus ([g_1^4],1) \oplus ([g_1^2], 1) \oplus ([g_1g_3],1)\cr
e_1  &\mapsto  ([1],r_3) \oplus ([g_1^2],1) \oplus ([g_1^4],r_3) \oplus ([g_1g_3],1)\cr
e_2  &\mapsto  ([1],r_5) \oplus ([g_1^2],g_3=-1\,g_1=-1) \\
& \quad\; \oplus ([g_1^4], r_5) \oplus ([g_1g_3],1)\cr
m_1  &\mapsto  ([g_2g_3], g_3=-1) \oplus ([g_1g_2],1)\cr
m_2  &\mapsto   ([g_3],g_2g_3=-1)\oplus ([g_1],1)\oplus ([g_3g_1^6],1)\,.
\ea\ee
As for the previous pair, the electric anyons map in the same way, while the map of magnetic anyons differ because of the 2-cocycle in the folded Lagrangian data.

\vspace{2mm}
\noindent\textbf{Third Pair: Gassmann Triple Twin.}
For $i=1,2$, algebras $A(H_i,1,1,1)$ in ${\cZ(\Vec_{(\Z_2 \times \Z_2) \ltimes \Z_8})}$ both define interfaces to reduced TO $\cZ(\Vec_{\Z_2\times\Z_2})$, as shown in \eqref{fig:G3243condquiche}. Their corresponding maps of anyons from $\cZ(\Vec_{\Z_2\times\Z_2})$ to ${\cZ(\Vec_{(\Z_2 \times \Z_2) \ltimes \Z_8})}$ are computed from the folded Lagrangian, as we explained in \Cref{sec:Fold}. Although the anyon decomposition of the twin algebras $A(H_i,1,1,1)$ is the same, their corresponding folded Lagrangians, hence maps of anyons, are inequivalent, even at the level of simple objects (anyons). In other words, they define inequivalent interfaces.

The above mentioned folded Lagrangians are determined from subgroups $H^\diag_i\subset G\times\Z_2\times\Z_2$, which, recalling \eqref{eq:Hdiag}-\eqref{eqn:HinG32-43} and denoting the $\Z_2\times\Z_2$ group generators as $m_1,m_2$, are taken to be: 
\be\ba
    H^\diag_1&=\langle (g_2g_3,m_1), (g_3,m_2) \rangle\,,\\
    H^\diag_2&=\langle (g_2g_3,m_1), (g_3g_1^4,m_2) \rangle\,.
\ea\ee

The folded Lagrangian for $A(H_1, 1, 1, 1)$, determines the map of anyons:\footnote{For $i=1,2,$ $C_G(g_3 g_1^{4(i-1)})\cong\Z_2^{(g_3g_1^{4(i-1)})}\times D_8^{(g_1^2,g_2)}$ and $([g_3g_1^{4(i-1)}],2)$ carries the 2-dim irrep of $D_8$.}
\be\ba \label{eq:G32-43_FL1}
\cZ(\Vec_{\Z_2^2}) & \rightarrow \cZ(\Vec_G) \cr
1  & \mapsto ([1],1)\oplus([1],r_1)\oplus([1],r_8)\oplus([1],r_{10})\cr
e_1  & \mapsto   ([1], r_2) \oplus ([1], r_3) \oplus([1],r_8)\oplus([1],r_{10})\cr
e_2  & \mapsto  ([1], r_4) \oplus ([1], r_5) \oplus ([1], r_{9}) \oplus ([1], r_{10} )\cr
m_1  & \mapsto ([g_2 g_3],1) \oplus ([g_2 g_3], g_1^4=-1)\cr
m_2  & \mapsto ([g_3],1) \oplus ([g_3], g_1^6=-1) \oplus ([g_3],2)\,.\cr
\ea\ee
In turn the folded Lagrangian for $A(H_2, 1, 1, 1)$ gives maps of anyons:
\be\ba \label{eq:G32-43_FL2}
\cZ(\Vec_{\Z_2^2}) & \rightarrow \cZ(\Vec_G) \cr
1  & \mapsto ([1],1)\oplus([1],r_1)\oplus([1],r_8)\oplus([1],r_{10})\cr
e_1  & \mapsto   ([1], r_2) \oplus ([1], r_3) \oplus([1],r_8)\oplus([1],r_{10})\cr
e_2  & \mapsto  ([1], r_4) \oplus ([1], r_5) \oplus ([1], r_{9}) \oplus ([1], r_{10} )\cr
m_1  & \mapsto  ([g_2 g_3],1) \oplus ([g_2 g_3], g_3=-1 \,g_1^4=-1)\cr
m_2  & \mapsto  ([g_3^{\phantom{-1}}\!\!\!\! g_1^4],1) \oplus ([g_3^{\phantom{-1}}\!\!\!\! g_1^4],\,g_1^6=-1) \oplus ([g_3^{\phantom{-1}}\!\!\!\! g_1^4],2)\,.\cr
\ea\ee
While the electric anyons $e_1, e_2$ map in the same way, the magnetic anyons $m_1, m_2$ map differently.

In \eqref{eq:G32-43_FL1} the anyons $m_1,m_2$, which are twisted sector, uncharged anyons for $\Z_2\times\Z_2$, are mapped to anyons in twisted sectors for $H_1$ (generated by $g_2g_3,g_3$) that are also uncharged under $H_1$, since they all carry irreps in which the group elements in $H_1$ are represented trivially. 

In turn, in \eqref{eq:G32-43_FL2}, $m_1,m_2$ are mapped to anyons in twisted sectors for $H_2$, generated by $g_2g_3,g_3g_1^4$ and carry trivial $H_2$ charge (since all $H_2$ group elements are represented trivially). \footnote{To compare with \eqref{eq:G32-43_FL1}, note that $g_3= g_1^{-1}(g_3^{\phantom{-1}}\!\!\!\! g_1^4) g_1^{\phantom{-1}}\!\!\!\! \,,$
from which we can write
\begin{align}
&([g_3^{\phantom{-1}}\!\!\!\! g_1^4],1) \oplus ([g_3^{\phantom{-1}}\!\!\!\! g_1^4],\,g_1^6=-1) \oplus ([g_3^{\phantom{-1}}\!\!\!\! g_1^4],2)\\
=&([g_3],1) \oplus ([g_3], g_2 g_3=-1 \,g_1^6=-1) \oplus ([g_3],2)
\end{align}
where in the irrep carried by $([g_3^{\phantom{-1}}\!\!\!\! g_1^4],\,g_1^6=-1)$, $g_2g_3g_1^6=-1$, which then implies that in the corresponding anyon in the second line $g_1^{-1}(g_2g_3g_1^6)g_1=g_2g_3=-1$: the image of $m_2$ in \eqref{eq:G32-43_FL2} is therefore different and inequivalent to its image in \eqref{eq:G32-43_FL1}. This is clear also for the image of $m_1$.}

In conclusion, although $A(H_i,1,1,1)$ have the same anyon decomposition, the anyons in their respective folded Lagrangians differ in how the $\Z_2\times\Z_2$ twisted sector anyons map: there exists a choice of representatives such that $A(H_i,1,1,1)$ corresponds to a $\Vec_G$-symmetric phase preserving $H_i$, consistently with the $G$-action described in \Cref{sec:basis-diff-action-same-OPE}.

\subsection{Maps of Anyons}
\label{app:map-of-anyon-G32omega}
\nameref{Aom:3} defines a map of anyons to be:
\begin{align}
    1 &\mapsto 1 \oplus ([1],r_1) \oplus ([1],r_8) \cr 
e_{ac} &\mapsto ([1],r_2) \oplus ([1],r_3) \oplus ([1],r_8) \cr 
e_{ab} &\mapsto ([1],r_4) \oplus ([1],r_5) \oplus ([1],r_9) \cr 
e_{bc} &\mapsto ([1],r_6) \oplus ([1],r_7) \oplus ([1],r_9) \cr 
e_{a} &\mapsto ([1],r_{10}) \cr 
e_{c} &\mapsto ([1],r_{10}) \cr 
e_{b} &\mapsto ([1],r_{10}) \cr 
e_{abc} &\mapsto ([1],r_{10}) \cr 
([c],2) &\mapsto ([g_2g_3],2) \cr 
([c],-2) &\mapsto ([g_2g_3],-2) \cr 
([b],2) &\mapsto ([g_3],2\;g_1^6=-\sqrt{2}\zeta_{16}) \oplus ([g_3],2\;g_1^6=\sqrt{2}\zeta_{16}) \cr 
([b],-2) &\mapsto ([g_3],-2\;g_1^6=-\sqrt{2}\zeta_{16}) \oplus ([g_3],-2\;g_1^6=\sqrt{2}\zeta_{16}) \cr 
([bac],2i) &\mapsto ([g_1^4],2i\;g_1g_2g_3=-2\zeta_8) \oplus ([g_1^4],2i\;g_1g_2g_3=2\zeta_8) \cr
& \qquad \oplus ([g_1^4],2i\;g_1g_2=-2) \oplus ([g_1^4],2i\;g_2g_1^7=-2) \cr 
([bac],-2i) &\mapsto ([g_1^4],-2i\;g_1=\sqrt{2}\zeta_{16}\;g_3g_1^6=2\zeta_8^3) \cr
& \qquad \oplus ([g_1^4],-2i\;g_1=\sqrt{2}\zeta_{16}\;g_3g_1^6=-2\zeta_8^3) \cr
&\qquad\oplus ([g_1^4],-2i\;g_1=-\sqrt{2}\zeta_{16}\;g_3g_1^6=2\zeta_8^3) \cr
& \qquad \oplus ([g_1^4],-2i\;g_1=-\sqrt{2}\zeta_{16}\;g_3g_1^6=-2\zeta_8^3) \cr 
([bc],2) &\mapsto ([g_2],2) \cr 
([bc],-2) &\mapsto ([g_2],-2) \cr 
([ab],2) &\mapsto ([g_2g_3],2) \cr 
([ab],-2) &\mapsto ([g_2g_3],-2) \cr 
([ac],2) &\mapsto ([g_3],2\;g_1^6=-\sqrt{2}\zeta_{16}) \oplus ([g_3],2\;g_1^6=\sqrt{2}\zeta_{16}) \cr 
([ac],-2) &\mapsto ([g_3],-2\;g_1^6=-\sqrt{2}\zeta_{16}) \oplus ([g_3],-2\;g_1^6=\sqrt{2}\zeta_{16}) \cr 
([a],2) &\mapsto ([g_2],2) \cr 
([a],-2) &\mapsto ([g_2],-2) \cr
\end{align}

\clearpage
\clearpage
\onecolumngrid

\begin{figure}
\begin{minipage}{\textwidth}
\begin{center}
\begin{tikzpicture}[vertex/.style={draw}, scale=0.9]
 \begin{scope}[shift={(0,0)}]
    \foreach \coord/\i/\j in {
(0.,-1.35)/1/{\nameref{alg:1}},
(0.,-2.7)/2/{\nameref{alg:2}},
(0.,-4.05)/3/{\nameref{alg:3}},
(-2.5,-5.4)/4/{\nameref{alg:4}$=A(H_1,1,1,1)$},
(2.5,-5.4)/5/{\nameref{alg:5}$=A(H_2,1,1,1)$},
(-7.7,-8)/6/{\nameref{alg:6}},
(-6.6,-8)/7/{\nameref{alg:7}},
(-5.5,-8)/8/{\nameref{alg:8}},
(-4.4,-8)/9/{\nameref{alg:9}},
(-3.3,-8)/10/{\nameref{alg:10}},
(-2.2,-8)/11/{\nameref{alg:11}},
(2.2,-8)/12/{\nameref{alg:12}},
(3.3,-8)/13/{\nameref{alg:13}},
(4.4,-8)/14/{\nameref{alg:14}},
(5.5,-8)/15/{\nameref{alg:15}},
(6.6,-8)/16/{\nameref{alg:16}},
(7.7,-8)/17/{\nameref{alg:17}},
(-1.1,-8)/18/{\nameref{alg:18}},
(0.,-8)/19/{\nameref{alg:19}},
(1.1,-8)/20/{\nameref{alg:20}},
(-7.7,-11)/21/{\nameref{alg:21}},
(-5.5,-11)/22/{\nameref{alg:22}},
(5.5,-11)/23/{\nameref{alg:23}},
(7.7,-11)/24/{\nameref{alg:24}},
(-3.3,-11)/25/{\nameref{alg:25}},
(-1.1,-11)/26/{\nameref{alg:26}},
(1.1,-11)/27/{\nameref{alg:27}},
(3.3,-11)/28/{\nameref{alg:28}}}
     {
      \node[vertex,align=center] (p\i) at \coord {\j};
     }
  \foreach [count=\r] \row in 
{{0,1,0,0,0,0,0,0,0,0,0,0,0,0,0,0,0,0,0,0,0,0,0,0,0,0,0,0},
{0,0,1,0,0,0,0,0,0,0,0,0,0,0,0,0,0,0,0,0,0,0,0,0,0,0,0,0},
{0,0,0,1,1,0,0,0,0,0,0,0,0,0,0,0,0,0,0,0,0,0,0,0,0,0,0,0},
{0,0,0,0,0,1,1,1,1,1,1,0,0,0,0,0,0,1,1,1,0,0,0,0,0,0,0,0},
{0,0,0,0,0,0,0,0,0,0,0,1,1,1,1,1,1,1,1,1,0,0,0,0,0,0,0,0},
{0,0,0,0,0,0,0,0,0,0,0,0,0,0,0,0,0,0,0,0,1,0,0,0,1,0,0,0},
{0,0,0,0,0,0,0,0,0,0,0,0,0,0,0,0,0,0,0,0,0,1,0,0,1,0,0,0},
{0,0,0,0,0,0,0,0,0,0,0,0,0,0,0,0,0,0,0,0,1,0,0,0,0,1,0,0},
{0,0,0,0,0,0,0,0,0,0,0,0,0,0,0,0,0,0,0,0,0,1,0,0,0,1,0,0},
{0,0,0,0,0,0,0,0,0,0,0,0,0,0,0,0,0,0,0,0,1,0,0,0,0,0,1,0},
{0,0,0,0,0,0,0,0,0,0,0,0,0,0,0,0,0,0,0,0,0,1,0,0,0,0,1,0},
{0,0,0,0,0,0,0,0,0,0,0,0,0,0,0,0,0,0,0,0,0,0,1,0,0,0,1,0},
{0,0,0,0,0,0,0,0,0,0,0,0,0,0,0,0,0,0,0,0,0,0,0,1,0,0,1,0},
{0,0,0,0,0,0,0,0,0,0,0,0,0,0,0,0,0,0,0,0,0,0,1,0,0,1,0,0},
{0,0,0,0,0,0,0,0,0,0,0,0,0,0,0,0,0,0,0,0,0,0,0,1,0,1,0,0},
{0,0,0,0,0,0,0,0,0,0,0,0,0,0,0,0,0,0,0,0,0,0,1,0,1,0,0,0},
{0,0,0,0,0,0,0,0,0,0,0,0,0,0,0,0,0,0,0,0,0,0,0,1,1,0,0,0},
{0,0,0,0,0,0,0,0,0,0,0,0,0,0,0,0,0,0,0,0,0,0,0,0,1,0,0,1},
{0,0,0,0,0,0,0,0,0,0,0,0,0,0,0,0,0,0,0,0,0,0,0,0,0,1,0,1},
{0,0,0,0,0,0,0,0,0,0,0,0,0,0,0,0,0,0,0,0,0,0,0,0,0,0,1,1},
{0,0,0,0,0,0,0,0,0,0,0,0,0,0,0,0,0,0,0,0,0,0,0,0,0,0,0,0},
{0,0,0,0,0,0,0,0,0,0,0,0,0,0,0,0,0,0,0,0,0,0,0,0,0,0,0,0},
{0,0,0,0,0,0,0,0,0,0,0,0,0,0,0,0,0,0,0,0,0,0,0,0,0,0,0,0},
{0,0,0,0,0,0,0,0,0,0,0,0,0,0,0,0,0,0,0,0,0,0,0,0,0,0,0,0},
{0,0,0,0,0,0,0,0,0,0,0,0,0,0,0,0,0,0,0,0,0,0,0,0,0,0,0,0},
{0,0,0,0,0,0,0,0,0,0,0,0,0,0,0,0,0,0,0,0,0,0,0,0,0,0,0,0},
{0,0,0,0,0,0,0,0,0,0,0,0,0,0,0,0,0,0,0,0,0,0,0,0,0,0,0,0},
{0,0,0,0,0,0,0,0,0,0,0,0,0,0,0,0,0,0,0,0,0,0,0,0,0,0,0,0}}
   {
     \foreach [count=\c] \cell in \row{
            \ifnum\cell=1%
                \draw[-stealth] (p\r) edge [thick] (p\c);
            \fi
        }
    }
\end{scope}
\end{tikzpicture}
\caption{Part of the Hasse diagram for $\cZ(\Vec_{(\Z_2\times \Z_2)\ltimes \Z_8})$ that connects to the twin algebras $A_4=A(H_1, 1,1,1)$ and $A_5=A(H_2, 1, 1, 1)$. The algebras $A_{18}, A_{19}, A_{20}$ are  condensable algebras that contain both twins. 
The algebras are detailed in table \ref{tab:G3248alg}. The top most algebra is the identity algebra. The bottom row are the Lagrangian algebras, that correspond to topological boundary conditions. }
\label{fig:HasseG3243}
\end{center}
\end{minipage}
\end{figure}
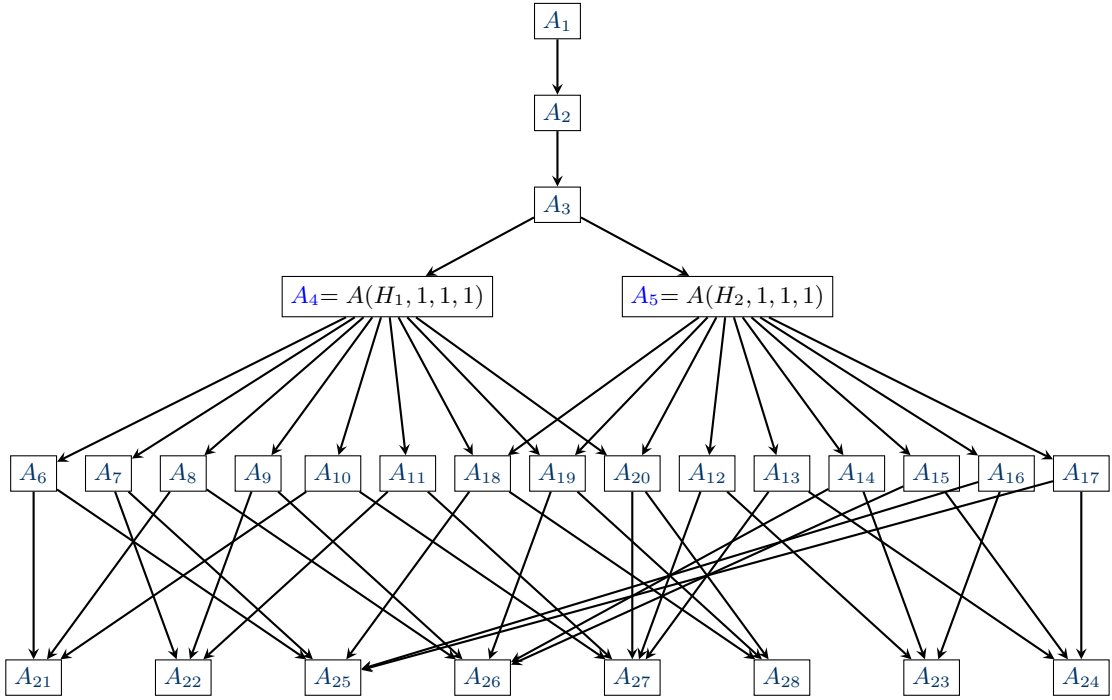

\clearpage
\newpage 
\onecolumngrid
\section{Tables of Condensable Algebras in $\cZ(\Vec_{(\Z_2 \times \Z_2) \ltimes \Z_8 }^\omega)$}
\label{app:TabConds}

\onecolumngrid

\vspace*{-5mm}
\begin{table}[H]
\begin{center}
{\scriptsize
\begin{tabular}{|c|}
\hline
Condensable Algebra in $\cZ(\Vec_{(\Z_2\times \Z_2)\ltimes \Z_8})$ \\ 
\hline
\nameref{alg:1}$=A(G=(\Z_2\times\Z_2) \ltimes \Z_8=\langle g_2,g_3,g_1 \rangle, 1, 1, 1)$\\
$\phantom{\cZ(\Vec_{(\Z_2\times\Z_2) \ltimes \Z_8^\omega})}
\xlabel[$A_{1}$]{alg:1} \hfill ([1],1) \hfill 
\cZ(\Vec_{(\Z_2\times\Z_2) \ltimes \Z_8})$\\
\hline

\nameref{alg:2}$=A(\Z_2\times D_8=\langle g_2g_3,g_3,g_1^6 \rangle, 1, 1, 1)$\\
$\phantom{\cZ(\Vec_{\Z_2\times D_8})}
\xlabel[$A_{2}$]{alg:2} \hfill ([1],1)\oplus([1],r_1) \hfill 
\cZ(\Vec_{\Z_2\times D_8})$\\
\hline

\nameref{alg:3}$=A(\Z_2\times\Z_2\times\Z_2=\langle g_2g_3,g_3,g_1^4 \rangle, 1, 1, 1)$\\
$\phantom{\cZ(\Vec_{\Z_2\times\Z_2\times\Z_2})}
\xlabel[$A_{3}$]{alg:3} \hfill ([1],1)\oplus([1],r_1)\oplus([1],r_8) \hfill 
\cZ(\Vec_{\Z_2\times\Z_2\times\Z_2})$\\
\hline

\textcolor{blue}{\nameref{alg:4}$=A(H_1=\Z_2\times\Z_2=\langle g_2g_3,g_3 \rangle, 1, 1, 1)$}\\
$\phantom{\cZ(\Vec_{\Z_2\times\Z_2})}
\xlabel[$\textcolor{blue}{A_{4}}$]{alg:4} \hfill 
\textcolor{blue}{([1],1)\oplus([1],r_1)\oplus([1],r_8)\oplus([1],r_{10})} \hfill 
\cZ(\Vec_{\Z_2\times\Z_2})$\\
\hline

\textcolor{blue}{\nameref{alg:5}$=A(H_2=\Z_2\times\Z_2=\langle g_2g_3,g_3g_1^4 \rangle, 1, 1, 1)$}\\
$\phantom{\cZ(\Vec_{\Z_2\times\Z_2})}
\xlabel[$\textcolor{blue}{A_{5}}$]{alg:5} \hfill 
\textcolor{blue}{([1],1)\oplus([1],r_1)\oplus([1],r_8)\oplus([1],r_{10})} \hfill 
\cZ(\Vec_{\Z_2\times\Z_2})$\\
\hline

\nameref{alg:6}$=A(\Z_2\times\Z_2=\langle g_2g_3,g_3 \rangle, \Z_2=\langle g_2g_3 \rangle, 1, 1)$\\
$\phantom{\cZ(\Vec_{\Z_2})}
\xlabel[$A_{6}$]{alg:6} \hfill ([1],1)\oplus([1],r_1)\oplus([1],r_8)\oplus([1],r_{10})\oplus([g_2g_3],1)\oplus([g_2g_3],g_1^4=-1) \hfill 
\cZ(\Vec_{\Z_2})$\\
\hline

\nameref{alg:7}$=A(\Z_2\times\Z_2=\langle g_2g_3,g_3 \rangle, \Z_2=\langle g_2g_3 \rangle, 1, \eps)$\\
$\phantom{\cZ(\Vec_{\Z_2})}
\xlabel[$A_{7}$]{alg:7} \hfill ([1],1)\oplus([1],r_1)\oplus([1],r_8)\oplus([1],r_{10})\oplus([g_2g_3],g_3=-1)\oplus([g_2g_3],g_3=-1\;g_1^4=-1) \hfill 
\cZ(\Vec_{\Z_2})$\\
\hline

\nameref{alg:8}$=A(\Z_2\times\Z_2=\langle g_2g_3,g_3 \rangle, \Z_2=\langle g_3 \rangle, 1, 1)$\\
$\phantom{\cZ(\Vec_{\Z_2})}
\xlabel[$A_{8}$]{alg:8} \hfill ([1],1)\oplus([1],r_1)\oplus([1],r_8)\oplus([1],r_{10})\oplus([g_3],1)\oplus([g_3],g_1^6=-1)\oplus([g_3],2) \hfill 
\cZ(\Vec_{\Z_2})$\\
\hline

\nameref{alg:9}$=A(\Z_2\times\Z_2=\langle g_2g_3,g_3 \rangle, \Z_2=\langle g_3 \rangle, 1, \eps)$\\
$\phantom{\cZ(\Vec_{\Z_2})}\phantom{\cZ(\Vec_{\Z_2})}
\xlabel[$A_{9}$]{alg:9} \hfill ([1],1)\oplus([1],r_1)\oplus([1],r_8)\oplus([1],r_{10})\oplus([g_3],g_2g_3=-1)\oplus([g_3],g_2g_3=-1\;g_1^6=-1)\oplus([g_3],2) \phantom{\cZ(\Vec_{\Z_2})}
\hfill 
\cZ(\Vec_{\Z_2})$\\
\hline

\nameref{alg:10}$=A(\Z_2\times\Z_2=\langle g_2g_3,g_3 \rangle, \Z_2=\langle g_2 \rangle, 1, 1)$\\
$\phantom{\cZ(\Vec_{\Z_2})}
\xlabel[$A_{10}$]{alg:10} \hfill ([1],1)\oplus([1],r_1)\oplus([1],r_8)\oplus([1],r_{10})\oplus([g_2],1)\oplus([g_2],g_1^4=-1) \hfill 
\cZ(\Vec_{\Z_2})$\\
\hline

\nameref{alg:11}$=A(\Z_2\times\Z_2=\langle g_2g_3,g_3 \rangle, \Z_2=\langle g_2 \rangle, 1, \eps)$\\
$\phantom{\cZ(\Vec_{\Z_2})}
\xlabel[$A_{11}$]{alg:11} \hfill ([1],1)\oplus([1],r_1)\oplus([1],r_8)\oplus([1],r_{10})\oplus([g_2],g_2g_3=-1)\oplus([g_2],g_2g_3=-1\;g_1^4=-1) \hfill 
\cZ(\Vec_{\Z_2})$\\
\hline

\nameref{alg:12}$=A(\Z_2\times\Z_2=\langle g_2g_3,g_3g_1^4 \rangle, \Z_2=\langle g_2g_1^4 \rangle, 1, 1)$\\
$\phantom{\cZ(\Vec_{\Z_2})}
\xlabel[$A_{12}$]{alg:12} \hfill ([1],1)\oplus([1],r_1)\oplus([1],r_8)\oplus([1],r_{10})\oplus([g_2],1)\oplus([g_2],g_2g_3=-1\;g_1^4=-1) \hfill 
\cZ(\Vec_{\Z_2})$\\
\hline

\nameref{alg:13}$=A(\Z_2\times\Z_2=\langle g_2g_3,g_3g_1^4 \rangle, \Z_2=\langle g_2g_1^4 \rangle, 1, \eps)$\\
$\phantom{\cZ(\Vec_{\Z_2})}
\xlabel[$A_{13}$]{alg:13} \hfill ([1],1)\oplus([1],r_1)\oplus([1],r_8)\oplus([1],r_{10})\oplus([g_2],g_1^4=-1)\oplus([g_2],g_2g_3=-1) \hfill 
\cZ(\Vec_{\Z_2})$\\
\hline

\nameref{alg:14}$=A(\Z_2\times\Z_2=\langle g_2g_3,g_3g_1^4 \rangle, \Z_2=\langle g_3g_1^4 \rangle, 1, 1)$\\
$\phantom{\cZ(\Vec_{\Z_2})}
\xlabel[$A_{14}$]{alg:14} \hfill ([1],1)\oplus([1],r_1)\oplus([1],r_8)\oplus([1],r_{10})\oplus([g_3],1)\oplus([g_3],g_2g_3=-1\;g_1^6=-1)\oplus([g_3],2) \hfill 
\cZ(\Vec_{\Z_2})$\\
\hline

\nameref{alg:15}$=A(\Z_2\times\Z_2=\langle g_2g_3,g_3g_1^4 \rangle, \Z_2=\langle g_3g_1^4 \rangle, 1, \eps)$\\
$\phantom{\cZ(\Vec_{\Z_2})}
\xlabel[$A_{15}$]{alg:15} \hfill ([1],1)\oplus([1],r_1)\oplus([1],r_8)\oplus([1],r_{10})\oplus([g_3],g_1^6=-1)\oplus([g_3],g_2g_3=-1)\oplus([g_3],2) \hfill 
\cZ(\Vec_{\Z_2})$\\
\hline

\nameref{alg:16}$=A(\Z_2\times\Z_2=\langle g_2g_3,g_3g_1^4 \rangle, \Z_2=\langle g_2g_3 \rangle, 1, 1)$\\
$\phantom{\cZ(\Vec_{\Z_2})}
\xlabel[$A_{16}$]{alg:16} \hfill ([1],1)\oplus([1],r_1)\oplus([1],r_8)\oplus([1],r_{10})\oplus([g_2g_3],1)\oplus([g_2g_3],g_3=-1\;g_1^4=-1) \hfill 
\cZ(\Vec_{\Z_2})$\\
\hline

\nameref{alg:17}$=A(\Z_2\times\Z_2=\langle g_2g_3,g_3g_1^4 \rangle, \Z_2=\langle g_2g_3 \rangle, 1, \eps)$\\
$\phantom{\cZ(\Vec_{\Z_2})}
\xlabel[$A_{17}$]{alg:17} \hfill ([1],1)\oplus([1],r_1)\oplus([1],r_8)\oplus([1],r_{10})\oplus([g_2g_3],g_3=-1)\oplus([g_2g_3],g_1^4=-1) \hfill 
\cZ(\Vec_{\Z_2})$\\
\hline

\nameref{alg:18}$=A(\Z_2=\langle g_2g_3 \rangle, 1, 1, 1)$\\
$\phantom{\cZ(\Vec_{\Z_2})}
\xlabel[$A_{18}$]{alg:18} \hfill ([1],1)\oplus([1],r_1)\oplus([1],r_4)\oplus([1],r_5)\oplus([1],r_8)\oplus([1],r_9)\oplus2([1],r_{10}) \hfill 
\cZ(\Vec_{\Z_2})$\\
\hline

\nameref{alg:19}$=A(\Z_2=\langle g_3 \rangle, 1, 1, 1)$\\
$\phantom{\cZ(\Vec_{\Z_2})}
\xlabel[$A_{19}$]{alg:19} \hfill ([1],1)\oplus([1],r_1)\oplus([1],r_2)\oplus([1],r_3)\oplus2([1],r_8)\oplus2([1],r_{10}) \hfill 
\cZ(\Vec_{\Z_2})$\\
\hline

\nameref{alg:20}$=A(\Z_2=\langle g_2 \rangle, 1, 1, 1)$\\
$\phantom{\cZ(\Vec_{\Z_2})}
\xlabel[$A_{20}$]{alg:20} \hfill ([1],1)\oplus([1],r_1)\oplus([1],r_6)\oplus([1],r_7)\oplus([1],r_8)\oplus([1],r_9)\oplus2([1],r_{10}) \hfill 
\cZ(\Vec_{\Z_2})$\\
\hline
\hline \\[-2mm]

\nameref{alg:21}$=A(\Z_2\times\Z_2=\langle g_2g_3,g_3 \rangle, \Z_2\times\Z_2=\langle g_2g_3,g_3 \rangle, 1, 1)$\\
$\xlabel[$A_{21}$]{alg:21} ([1],1)\oplus([1],r_1)\oplus([1],r_8)\oplus([1],r_{10})\oplus([g_2g_3],1)\oplus([g_2g_3],g_1^4=-1)\oplus([g_3],1)$\\
$\oplus([g_3],g_1^6=-1)\oplus([g_3],2)\oplus([g_2],1)\oplus([g_2],g_1^4=-1)$\\
\hline

\nameref{alg:22}$=A(\Z_2\times\Z_2=\langle g_2g_3,g_3 \rangle, \Z_2\times\Z_2=\langle g_2g_3,g_3 \rangle, \gamma, \eps)$\\
$\xlabel[$A_{22}$]{alg:22} ([1],1)\oplus([1],r_1)\oplus([1],r_8)\oplus([1],r_{10})\oplus([g_2g_3],g_3=-1)\oplus([g_2g_3],g_3=-1\;g_1^4=-1)\oplus([g_3],g_2g_3=-1)$\\
$\oplus([g_3],g_2g_3=-1\;g_1^6=-1)\oplus([g_3],2)\oplus([g_2],1\
 g_2g_3=-1)\oplus([g_2],g_2g_3=-1\;g_1^4=-1)$\\
\hline

\nameref{alg:23}$=A(\Z_2\times\Z_2=\langle g_2g_3,g_3g_1^4 \rangle, \Z_2\times\Z_2=\langle g_2g_3,g_3g_1^4 \rangle, 1, 1)$\\
$\xlabel[$A_{23}$]{alg:23} ([1],1)\oplus([1],r_1)\oplus([1],r_8)\oplus([1],r_{10})\oplus([g_2g_3],1)\oplus([g_2g_3],g_3=-1\;g_1^4=-1)\oplus([g_3],1)$\\
$\oplus([g_3],g_2g_3=-1\;g_1^6=-1)\oplus([g_3],2)\oplus([g_2],1)\oplus([g_2],g_2g_3\
=-1\;g_1^4=-1)$\\
\hline

\nameref{alg:24}$=A(\Z_2\times\Z_2=\langle g_2g_3,g_3g_1^4 \rangle, \Z_2\times\Z_2=\langle g_2g_3,g_3g_1^4 \rangle, \gamma, \eps)$\\
$\xlabel[$A_{24}$]{alg:24} ([1],1)\oplus([1],r_1)\oplus([1],r_8)\oplus([1],r_{10})\oplus([g_2g_3],g_3=-1)\oplus([g_2g_3],g_1^4=-1)\oplus([g_3],g_1^6=-1)$\\
$\oplus([g_3],g_2g_3=-1)\oplus([g_3],2)\oplus([g_2],g_1^4=-1)\oplus([g_2]\
,1\;g_2g_3=-1)$\\
\hline

\nameref{alg:25}$=A(\Z_2=\langle g_2g_3 \rangle, \Z_2=\langle g_2g_3 \rangle, 1, 1)$\\
$\xlabel[$A_{25}$]{alg:25} ([1],1)\oplus([1],r_1)\oplus([1],r_4)\oplus([1],r_5)\oplus([1],r_8)\oplus([1],r_9)\oplus2([1],r_{10})\oplus([g_2g_3],1)$\\
$\oplus([g_2g_3],g_3=-1)\oplus([g_2g_3],g_1^4=-1)\oplus([g_2g_3],g_3=-1\;g_1^4=-\
1)$\\
\hline

\nameref{alg:26}$=A(\Z_2=\langle g_3 \rangle, \Z_2=\langle g_3 \rangle, 1, 1)$\\
$\xlabel[$A_{26}$]{alg:26} ([1],1)\oplus([1],r_1)\oplus([1],r_2)\oplus([1],r_3)\oplus2([1],r_8)\oplus2([1],r_{10})\oplus([g_3],1)\oplus([g_3],g_1^6=-1)$\\
$\oplus([g_3],g_2g_3=-1)\oplus([g_3],g_2g_3=-1\;g_1^6=-1)\oplus2([g_3],2)$\\
\hline

\nameref{alg:27}$=A(\Z_2=\langle g_2 \rangle, \Z_2=\langle g_2 \rangle, 1, 1)$\\
$\xlabel[$A_{27}$]{alg:27} ([1],1)\oplus([1],r_1)\oplus([1],r_6)\oplus([1],r_7)\oplus([1],r_8)\oplus([1],r_9)\oplus2([1],r_{10})\oplus([g_2],1)$\\
$\oplus([g_2],g_1^4=-1)\oplus([g_2],g_2g_3=-1)\oplus([g_2],g_2g_3=-1\;g_1^4=-1)$\\
\hline

\nameref{alg:28}$=A(1, 1, 1, 1)$\\
$\xlabel[$A_{28}$]{alg:28} ([1],1)\oplus([1],r_1)\oplus([1],r_2)\oplus([1],r_3)\oplus([1],r_4)\oplus([1],r_5)\oplus([1],r_6)$\\
$\oplus([1],r_7)\oplus2([1],r_8)\oplus2([1],r_9)\oplus4([1],r_{10})$\\
\hline
\end{tabular}
}
\end{center}
\caption{Condensable algebras for $\cZ(\Vec_{(\Z_2\times \Z_2)\ltimes \Z_8})$ that either contain, or are contained within the twin algebras \nameref{alg:4} and \nameref{alg:5}. The top line in each entry gives the data $A(H, N, \gamma, \epsilon)$, the second row the decomposition into anyons and reduced topological order (TO). For the Lagrangian algebras (shown below the double line) the reduced TO is trivial.  The irreps are described in Table \ref{tab:G3243chars} and the associated Hasse diagram is in Figure \ref{fig:HasseG3243}. \label{tab:G3248alg}}
\end{table}

\onecolumngrid
\begin{table}[H]
\renewcommand{\arraystretch}{1.24}
\begin{center}
{\scriptsize
\begin{tabular}{|c|}
\hline
Condensable Algebra in $\cZ(\Vec_{(\Z_2\times \Z_2)\ltimes \Z_8})$ \\[0.5pt] 
\hline
\nameref{Aom:1}$=A(G=(\Z_2\times\Z_2) \ltimes \Z_8=\langle g_2,g_3,g_1 \rangle, 1, 1, 1)$\\
$\xlabel[$A^\omega_{1}$]{Aom:1}([1],1)$\\
$\cZ(\Vec_{G}^\omega)$\\
\hline

\nameref{Aom:2}$=A(\Z_2 \times D_8=\langle g_2g_3,g_3,g_1^6 \rangle, 1, 1, 1)$\\
$\xlabel[$A^\omega_{2}$]{Aom:2}([1],1) \oplus ([1],r_1)$\\
$\cZ(\Vec_{\Z_2 \times D_8}^\alpha)$\\[1pt]
\hline

\nameref{Aom:3}$=A(\Z_2 \times \Z_2 \times \Z_2=\langle g_2g_3,g_3,g_1^4 \rangle, 1, 1, 1)$\\
$\xlabel[$A^\omega_{3}$]{Aom:3}([1],1) \oplus ([1],r_1) \oplus ([1],r_8)$\\
$\cZ(\Vec_{\Z_2 \times \Z_2 \times \Z_2}^{\omega_\III})$\\[1.5pt]
\hline

\nameref{Aom:4}$=A(H_1=\langle g_2g_3,g_3 \rangle, 1, 1, 1)$\\
$\xlabel[$\textcolor{blue}{A^\omega_{4}}$]{Aom:4}([1],1) \oplus ([1],r_1) \oplus ([1],r_8) \oplus ([1],r_{10})$\\
$\cZ(\Vec_{\Z_2 \times \Z_2})$\\
\hline

\nameref{Aom:5}$=A(H_2=\langle g_2g_3,g_3g_1^4 \rangle, 1, 1, 1)$\\
$\xlabel[$\textcolor{blue}{A^\omega_{5}}$]{Aom:5}([1],1) \oplus ([1],r_1) \oplus ([1],r_8) \oplus ([1],r_{10})$\\
$\cZ(\Vec_{\Z_2 \times \Z_2})$\\
\hline

\nameref{Aom:6}$=A(H_1=\langle g_2g_3,g_3 \rangle, \Z_2=\langle g_2g_3 \rangle, 1, 1)$\\
$\xlabel[$\textcolor{cyan}{A^{\omega}_{6}}$]{Aom:6}([1],1) \oplus ([1],r_1) \oplus ([1],r_8) \oplus ([1],r_{10}) \oplus ([g_2g_3],2)$\\
$\cZ(\Vec_{\Z_2})$\\
\hline

\nameref{Aom:11}$=A(H_2=\langle g_2g_3,g_3g_1^4 \rangle, \Z_2=\langle g_2g_3 \rangle, 1, 1)$\\
$\xlabel[$\textcolor{cyan}{A^{\omega}_{11}}$]{Aom:11}([1],1) \oplus ([1],r_1) \oplus ([1],r_8) \oplus ([1],r_{10}) \oplus ([g_2g_3],2)$\\
$\cZ(\Vec_{\Z_2})$\\
\hline

\nameref{Aom:8}$=A(H_1=\langle g_2g_3,g_3 \rangle, \Z_2=\langle g_2 \rangle, 1, 1)$\\
$\xlabel[$\textcolor{JungleGreen}{A^\omega_{8}}$]{Aom:8}([1],1) \oplus ([1],r_1) \oplus ([1],r_8) \oplus ([1],r_{10}) \oplus ([g_2],2)$\\
$\cZ(\Vec_{\Z_2})$\\
\hline

\nameref{Aom:9}$=A(H_2=\langle g_2g_3,g_3g_1^4 \rangle, \Z_2=\langle g_2g_1^4 \rangle, 1, 1)$\\
$\xlabel[$\textcolor{JungleGreen}{A^\omega_{9}}$]{Aom:9}([1],1) \oplus ([1],r_1) \oplus ([1],r_8) \oplus ([1],r_{10}) \oplus ([g_2],2)$\\
$\cZ(\Vec_{\Z_2})$\\
\hline

\nameref{Aom:7}$=A(H_1=\langle g_2g_3,g_3 \rangle, \Z_2=\langle g_3 \rangle, 1, 1)$\\
$\xlabel[$\textcolor{PineGreen}{A^\omega_{7}}$]{Aom:7}([1],1) \oplus ([1],r_1) \oplus ([1],r_8) \oplus ([1],r_{10}) \oplus ([g_3],2\;g_1^6=-\sqrt{2}\zeta_{16}) \oplus ([g_3],2\;g_1^6=\sqrt{2}\zeta_{16})$\\
$\cZ(\Vec_{\Z_2})$\\
\hline

\nameref{Aom:10}$=A(H_2=\langle g_2g_3,g_3g_1^4 \rangle, \Z_2=\langle g_3g_1^4 \rangle, 1, 1)$\\
$\xlabel[$\textcolor{PineGreen}{A^\omega_{10}}$]{Aom:10}([1],1) \oplus ([1],r_1) \oplus ([1],r_8) \oplus ([1],r_{10}) \oplus ([g_3],2\;g_1^6=-\sqrt{2}\zeta_{16}) \oplus ([g_3],2\;g_1^6=\sqrt{2}\zeta_{16})$\\
$\cZ(\Vec_{\Z_2})$\\
\hline

\nameref{Aom:12}$=A(\Z_2=\langle g_2g_3 \rangle, 1, 1, 1)$\\
$\xlabel[$A^\omega_{12}$]{Aom:12}([1],1) \oplus ([1],r_1) \oplus ([1],r_4) \oplus ([1],r_5) \oplus ([1],r_8) \oplus ([1],r_9) \oplus 2([1],r_{10})$\\
$\cZ(\Vec_{\Z_2})$\\
\hline

\nameref{Aom:13}$=A(\Z_2=\langle g_3 \rangle, 1, 1, 1)$\\
$\xlabel[$A^\omega_{13}$]{Aom:13}([1],1) \oplus ([1],r_1) \oplus ([1],r_2) \oplus ([1],r_3) \oplus 2([1],r_8) \oplus 2([1],r_{10})$\\
$\cZ(\Vec_{\Z_2})$\\
\hline

\nameref{Aom:14}$=A(\Z_2=\langle g_2 \rangle, 1, 1, 1)$\\
$\xlabel[$A^\omega_{14}$]{Aom:14}([1],1) \oplus ([1],r_1) \oplus ([1],r_6) \oplus ([1],r_7) \oplus ([1],r_8) \oplus ([1],r_9) \oplus 2([1],r_{10})$\\
$\cZ(\Vec_{\Z_2})$\\
\hline
\hline

\nameref{Aom:15}$=A(H_1=\langle g_2g_3,g_3 \rangle, H_1=\langle g_2g_3,g_3 \rangle, 1, 1)$\\
$\xlabel[$\textcolor{purple}{A^\omega_{15}}$]{Aom:15}([1],1) \oplus ([1],r_1) \oplus ([1],r_8) \oplus ([1],r_{10}) \oplus ([g_2g_3],2) \oplus ([g_3],2\;g_1^6=-\sqrt{2}\zeta_{16}) \oplus ([g_3],2\;g_1^6=\sqrt{2}\zeta_{16}) \oplus ([g_2],2)$\\
{Trivial}\\
\hline

\nameref{Aom:16}$=A(H_2=\langle g_2g_3,g_3g_1^4 \rangle, H_2=\langle g_2g_3,g_3g_1^4 \rangle, 1, 1)$\\
$\xlabel[$\textcolor{purple}{A^\omega_{16}}$]{Aom:16}([1],1) \oplus ([1],r_1) \oplus ([1],r_8) \oplus ([1],r_{10}) \oplus ([g_2g_3],2) \oplus ([g_3],2\;g_1^6=-\sqrt{2}\zeta_{16}) \oplus ([g_3],2\;g_1^6=\sqrt{2}\zeta_{16}) \oplus ([g_2],2)$\\
{Trivial}\\
\hline

\nameref{Aom:17}$=A(\Z_2=\langle g_2g_3 \rangle, \Z_2=\langle g_2g_3 \rangle, 1, 1)$\\
$\xlabel[$A^\omega_{17}$]{Aom:17}([1],1) \oplus ([1],r_1) \oplus ([1],r_4) \oplus ([1],r_5) \oplus ([1],r_8) \oplus ([1],r_9) \oplus 2([1],r_{10}) \oplus 2([g_2g_3],2)$\\
{Trivial}\\
\hline

\nameref{Aom:18}$=A(\Z_2=\langle g_3 \rangle, \Z_2=\langle g_3 \rangle, 1, 1)$\\
$\xlabel[$A^\omega_{18}$]{Aom:18}([1],1) \oplus ([1],r_1) \oplus ([1],r_2) \oplus ([1],r_3) \oplus 2([1],r_8) \oplus 2([1],r_{10}) \oplus 2([g_3],2\;g_1^6=-\sqrt{2}\zeta_{16}) \oplus 2([g_3],2\;g_1^6=\sqrt{2}\zeta_{16})$\\
{Trivial}\\
\hline

\nameref{Aom:19}$=A(\Z_2=\langle g_2 \rangle, \Z_2=\langle g_2 \rangle, 1, 1)$\\
$\xlabel[$A^\omega_{19}$]{Aom:19}([1],1) \oplus ([1],r_1) \oplus ([1],r_6) \oplus ([1],r_7) \oplus ([1],r_8) \oplus ([1],r_9) \oplus 2([1],r_{10}) \oplus 2([g_2],2)$\\
{Trivial}\\
\hline

\nameref{Aom:20}$=A(1, 1, 1, 1)$\\
$\xlabel[$A^\omega_{20}$]{Aom:20}([1],1) \oplus ([1],r_1) \oplus ([1],r_2) \oplus ([1],r_3) \oplus ([1],r_4) \oplus ([1],r_5) \oplus ([1],r_6) \oplus ([1],r_7) \oplus 2([1],r_8) \oplus 2([1],r_9) \oplus 4([1],r_{10})$\\
{Trivial}\\
\hline
\end{tabular}
}
\end{center}
\caption{Condensable algebras for $\cZ(\Vec_{(\Z_2\times \Z_2)\ltimes \Z_8}^\omega)$ that either contain, or are contained within the twin algebras \nameref{Aom:4} and \nameref{Aom:5}. The top line in each entry gives the data $A(H, N, \gamma, \epsilon)$, the second row the decomposition into anyons and reduced topological order (TO). For the Lagrangian algebras (shown below the double line) the reduced TO is trivial.  The irreps are described in Table \ref{tab:G3243chars} and the associated Hasse diagram is in Figure \ref{fig:Hasse3243_om}. 
\label{tab:G_om}}
\end{table}

\twocolumngrid

\clearpage
\newpage

\clearpage
\newpage
\twocolumngrid
\bibliographystyle{ytphys}
\small 
\baselineskip=.94\baselineskip
\let\bbb\bibitem\def\bibitem{\itemsep4pt\bbb}
\bibliography{GenSym}


\end{document}